\documentclass[aps,oneside,12pt]{CUEDthesisPSnPDF}  

\title{Quantum And Relativistic Protocols For Secure Multi-Party
Computation}

  \author{Roger Colbeck}
  \collegeordept{Trinity College}
  \university{University of Cambridge}
\crest{\includegraphics[width=30mm]{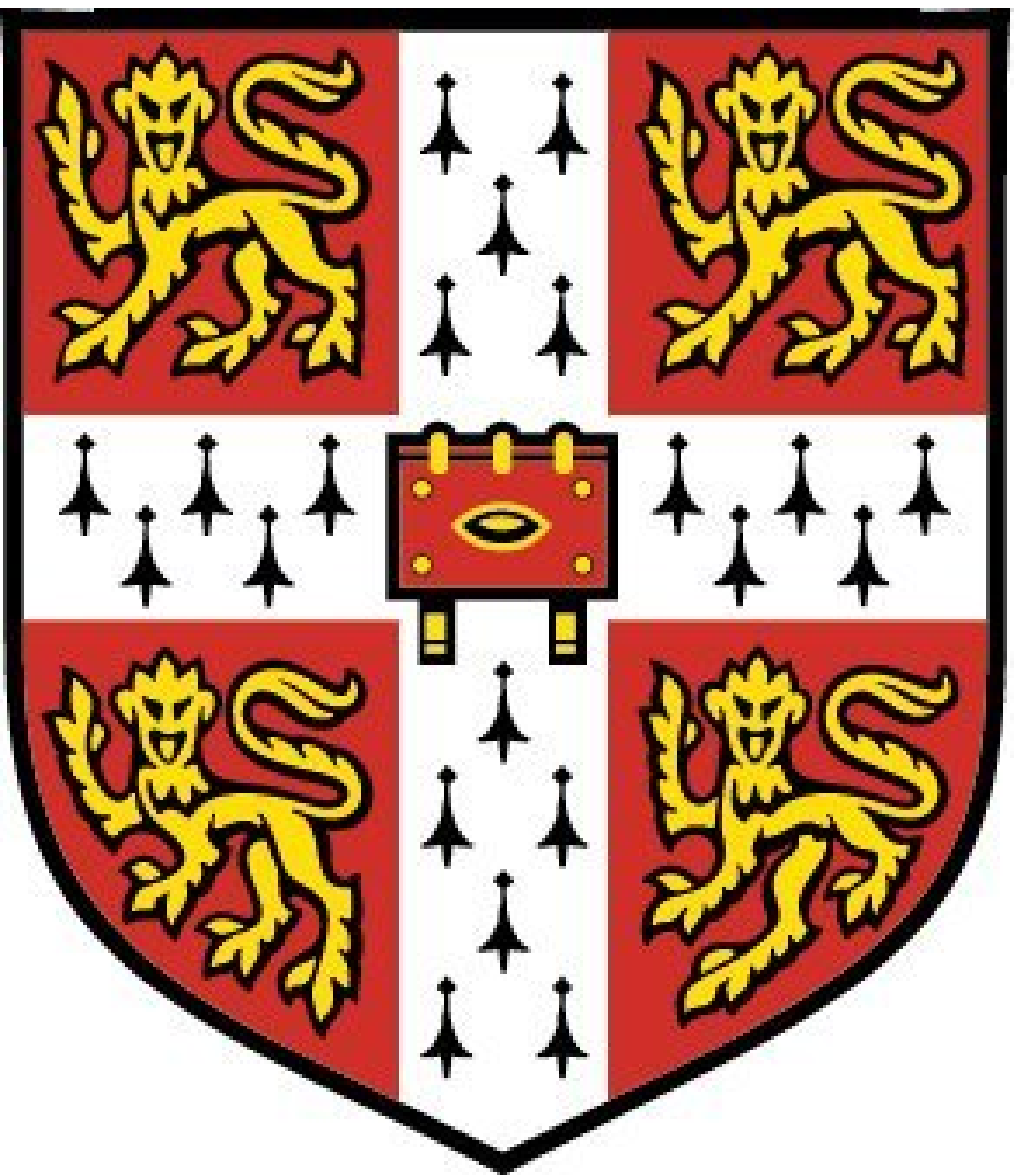}\qquad\includegraphics[width=30mm]{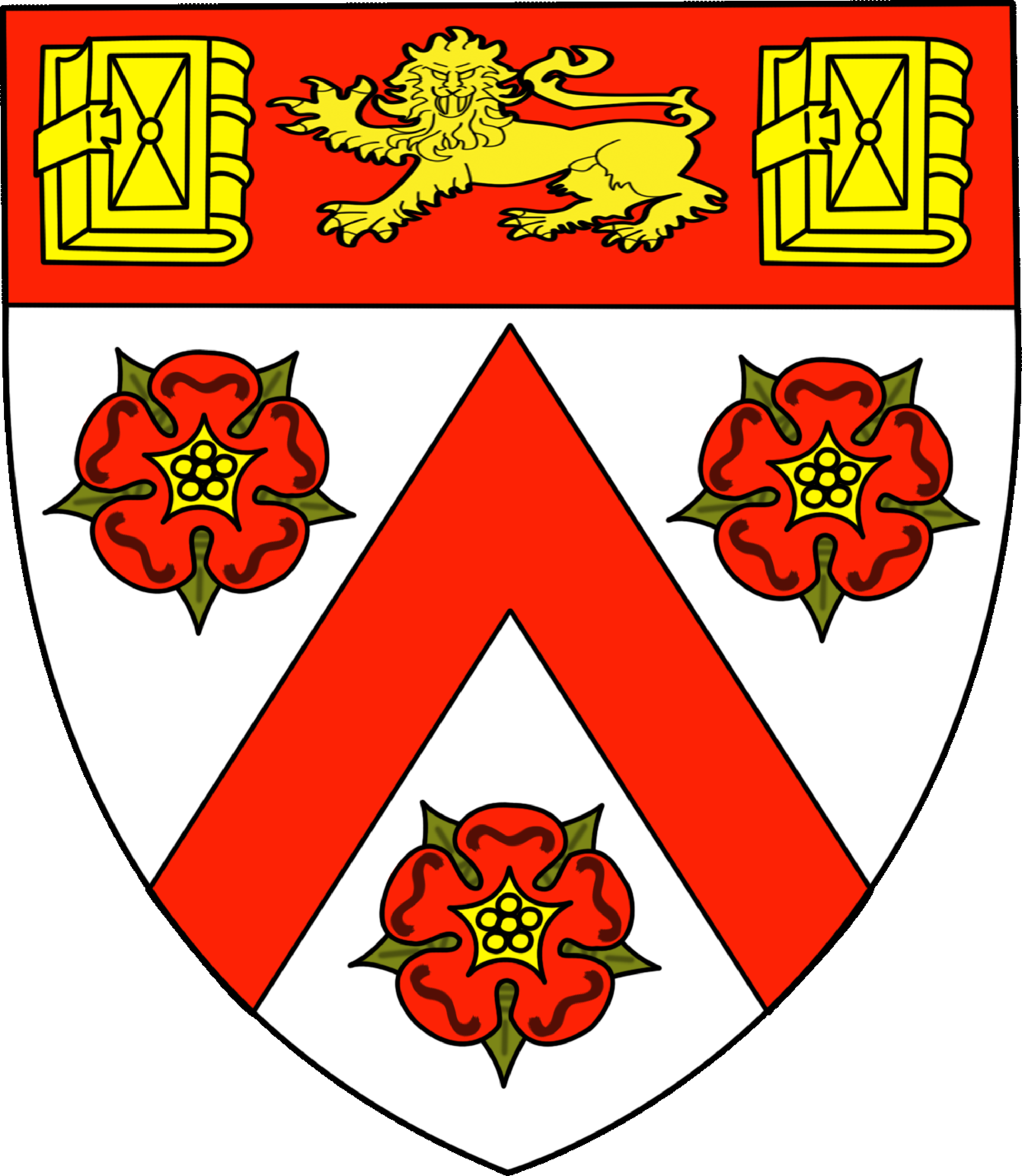}}

\degree{Doctor of Philosophy}
\degreedate{December 2006}

\usepackage{amsmath}
\usepackage{watermark}
\usepackage{pifont}                              
\usepackage{wasysym}                            
\usepackage{multirow}                           
\usepackage{amsthm}                             
\usepackage{rotating}   
\usepackage{enumerate}
\newcounter{assum}
\newtheorem{assumption}[assum]{Assumption}

\newcommand{\comment}[1]{}
\newcommand{\ket}[1]{\left |  #1 \right\rangle}
\newcommand{\bra}[1]{\left \langle #1  \right |}

\newcommand{\ketbra}[2]{|#1\rangle\!\langle#2|}
\def\openone{\leavevmode\hbox{\small1\kern-3.8pt\normalsize1}}

\theoremstyle{plain}
\newtheorem{conjecture}{Conjecture}[chapter]
\newtheorem{theorem}{Theorem}[chapter]
\newtheorem{lemma}{Lemma}[chapter]
\newtheorem{ib}{Ideal Behaviour}
\newtheorem{idf}{Ideal Functionality}[chapter]
\newtheorem{corollary}{Corollary}[chapter]

\theoremstyle{definition}
\newtheorem{protocol}{Protocol}[chapter]
\newtheorem{definition}{Definition}[chapter]
\newtheorem*{SC}{Security Condition}

\begin{document}


\renewcommand\baselinestretch{1.2}
\baselineskip=18pt plus1pt  


\fancyfoot{}
\renewcommand{\headrulewidth}{0pt}
\renewcommand{\headrulewidth}{0.5pt}


\maketitle

\setcounter{secnumdepth}{3}
\setcounter{tocdepth}{2}

\begin{declaration}
This thesis is the result of my own work and includes nothing which is
the outcome of work done in collaboration with others, except where
specifically indicated in the text.
\end{declaration}


\begin{dedication} 
{\large
\it To my parents, Lorelei and John.}

\bigskip

\bigskip
{\small
  \ding{72}\ding{73}\ding{75}\ding{76}\ding{77}\ding{78}\ding{79}{\large \ding{74}}\ding{79}\ding{78}\ding{77}\ding{76}\ding{75}\ding{73}\ding{72}}

\end{dedication}




\begin{acknowledgements}      
The work comprising this thesis was carried out over the course of
three years at the Centre For Quantum Computation, DAMTP, Cambridge,
under the supervision of Adrian Kent.  I am indebted to him for
support and guidance.  Discussions with Adrian are always a pleasure
and his eagerness to find loopholes in any new proposal is second to
none.  His clear thinking and attention to detail have had huge impact
on my work, and indeed my philosophy towards research.

It is with great pleasure that I thank the members of the quantum
information group for an enjoyable three years.  Particular thanks go
to Matthias Christandl, Robert K\"onig, Graeme Mitchison and Renato
Renner for numerous useful discussions from which my work has
undoubtedly benefited, and to Jiannis Pachos for his constant
encouragement and support.

My two office co-inhabitants deserve special mention here, not least
for putting up with me!  Alastair Kay for withstanding (and almost
always answering) a battering of questions on forgotten physics,
\LaTeX, linux and much more besides, and Roberta Rodriquez
for providing unrelenting emotional support on all matters from
physics to life itself.

I am very grateful to the Engineering and Physical Sciences Research
Council for a research studentship and to Trinity College, Cambridge
for a research scholarship and travel money.  A junior research
fellowship from Homerton College, Cambridge has provided financial
support during the final stages of writing this thesis.

Finally I would like to thank my examiners Robert Spekkens and Andreas
Winter for their thorough analysis of this thesis.

\end{acknowledgements}



\fancyhead[RO]{{\bf Abstract}}


\begin{abstracts}        
Secure multi-party computation is a task whereby mistrustful parties
attempt to compute some joint function of their private data in such a
way as to reveal as little as possible about it.  It encompasses many
cryptographic primitives, including coin tossing and oblivious
transfer.  Ideally, one would like to generate either a protocol or a
no-go theorem for any such task.

Very few computations of this kind are known to be possible with
unconditional security.  However, relatively little investigation into
exploiting the cryptographic power of a relativistic theory has been
carried out.  In this thesis, we extend the range of known results
regarding secure multi-party computations.  We focus on two-party
computations, and consider protocols whose security is guaranteed by
the laws of physics.  Specifically, the properties of quantum systems,
and the impossibility of faster-than-light signalling will be used to
guarantee security.

After a general introduction, the thesis is divided into four parts.
In the first, we discuss the task of coin tossing, principally in
order to highlight the effect different physical theories have on
security in a straightforward manner, but, also, to introduce a new
protocol for non-relativistic strong coin tossing. This protocol
matches the security of the best protocol known to date while using a
conceptually different approach to achieve the task.  It provides a
further example of the use of entanglement as a resource.

In the second part, a new task, variable bias coin tossing, is
introduced.  This is a variant of coin tossing in which one party
secretly chooses one of two biased coins to toss.  It is shown that
this can be achieved with unconditional security for a specified range
of biases, and with cheat-evident security for any bias.  We also
discuss two further protocols which are conjectured to be
unconditionally secure for any bias.

The third section looks at other two-party secure computations for
which, prior to our work, protocols and no-go theorems were unknown.
We introduce a general model for such computations, and show that,
within this model, a wide range of functions are impossible to compute
securely.  We give explicit cheating attacks for such functions.

In the final chapter we investigate whether cryptography is possible
under weakened assumptions.  In particular, we discuss the task of
expanding a private random string, while dropping the assumption that
the protocol's user trusts her devices.  Instead we assume that all
quantum devices are supplied by an arbitrarily malicious adversary.
We give two protocols that we conjecture securely perform this task.
The first allows a private random string to be expanded by a finite
amount, while the second generates an arbitrarily large expansion of
such a string.

\end{abstracts}



\fancyhead[RO]{{\bf Contents}}
\fancyfoot[C]{\thepage}
\begin{romanpages}
\tableofcontents
\listoffigures
\fancyhead[RO]{{\bf List Of Figures}}
\cleardoublepage 
\markboth{\nomname}{\nomname} 

\fancyhead[RO]{\bfseries\rightmark}

\end{romanpages}

\chapter{Introduction}

\begin{quote}
{\it ``If you wish another to keep your secret, first keep it to
  yourself.''} -- Lucius Annaeus Senec
\end{quote}

\section{Preface}
Secrecy has been an important aspect of life since the birth of
civilization, if not before -- even squirrels hide their nuts.  While
the poor squirrel has to rely on unproven assumptions about the
intelligence and digging power of its adversaries, we, today, seek a
more powerful predicate.  We demand unconditional security, that is
security guaranteed by the laws of physics.  

One common example is that of a base communicating with a field agent.
In the standard incarnation of this problem, Alice, at base, uses a
key to encrypt her data, before sending the encryption in the clear to
agent Bob.  An eavesdropper, Eve, hearing only the encrypted message
can discover little about the data. Shannon's pioneering work on
information theory implies that to achieve perfect secrecy, so that,
{\em even if she possesses the entire encrypted message}, Eve can do
no better than simply guess Alice's message, requires a key that is at
least as long as the message.

This is an inconvenient result.  Distributing, carrying and securely
storing long keys is expensive.  In the 1970s, a band of classical
cryptographers came up with a set of practical ciphers to which they
entrusted their private communications, and indeed many of us do
today.  These evade Shannon's requirement on the key by assuming
something about the power of an eavesdropper.  The Rivest, Shamir and
Adleman cipher ({\sc RSA}\nomenclature[zRSA]{{\sc RSA}}{The Rivest
Shamir Adleman cipher}), for instance, assumes that an eavesdropper
finds it hard to factor a large number into the product of two primes.
Security then has a finite lifetime, the factoring time.  Although for
matters of national security, such a cryptosystem is inappropriate, it
is of considerable use to protect short-lived secrets.  For instance
if it is known that a hacker takes 20 years to find out a credit card
number sent over the internet, one simply needs to issue new credit
card numbers at least every 20 years.  But it is not that simple.  It
may take 20 years running the best known algorithm on the fastest
computer available today to break the code, but this could change
overnight.  The problem with relying on a task such as factoring is
that no one actually knows how hard it is.  In some sense, we believe
it is secure because very clever people have spent large amounts of
time trying to find a fast factoring algorithm implementable on
today's computers, and have failed.  More alarmingly, we actually know
of a factoring algorithm that works efficiently on a quantum computer.
We are then relying for security on no one having successfully built
such a computer.  Perhaps one already exists in the depths of some
shady government organization.  There are bigger secrets than one's
credit card numbers, and for these, we cannot risk such possibilities.

As we have mentioned, a quantum computer can efficiently break the
cryptosystems we use today.  Quantum technology also allows us to
build cryptosystems with improved power, and in fact such that they
are provably unbreakable within our assumptions.  The usefulness of
quantum mechanics in cryptography went un-noticed for many years.
Wiesner made the first step in 1970 in a work that remained
unpublished until 1983.  In 1984, Bennett and Brassard extended
Wiesner's idea to introduce the most notorious utilization of quantum
mechanics in cryptography -- quantum key distribution.  This allows a
key to be generated remotely between two parties with unconditional
security, thus circumventing the problem of securely storing a long
key.  The principle behind Bennett and Brassard's scheme is that
measuring a quantum state necessarily disturbs it.  If an eavesdropper
tries to tap their quantum channel, Alice and Bob can detect this.  If
a disturbance is detected, they simply throw away their key and start
again.  Any information Eve gained on the key is useless to her since
this key will be discarded.  Alice and Bob can then be assured of
their privacy.  Remote key distribution is impossible classically,
hence quantum mechanics is, at least in this respect, a
cryptographically more powerful theory.

Other cryptographic primitives have been applied to the quantum
setting.  These so-called post cold war applications focus on exchange
of information between potentially mistrustful parties.  Multiple
parties wish to protect their own data (perhaps only temporarily)
while using it in some protocol.  Bit commitment is one such example.
In an everyday analogy of this primitive, Alice writes a bit on a
piece of paper and locks it in a safe.  She sends the safe to Bob, but
keeps the key.  At some later time Alice can unveil the bit to Bob by
sending him the key, thus proving that she was committed to the bit
all along.  Of course, this scheme is not fundamentally secure---it
relies on unproven assumptions about the size of sledgehammer
available to Bob.  Mayers, Lo and Chau showed that a large class of
quantum bit commitment schemes are impossible.  This cast major doubt
on the possible success of other such primitives, but all was not
lost.  In 1999, Kent noticed that exploiting another physical theory
might rescue the situation. (He was not in fact the first to consider
using this theory, but seems to be the first to obtain a working
result.)  Special relativity\footnote{Strictly, special relativity and
the assumption of causality.} demands that information does not travel
faster than the speed of light.  The essence of its usefulness is that
in a relativistic protocol, we can demand that certain messages be
sent simultaneously by different parties.  The receipt times can then
be used to guarantee that these messages were generated independently.
Coin tossing, for example becomes very straightforward.  Alice and Bob
simply simultaneously send one another a random bit.  If the bits are
equal, they assign heads, if different, they assign tails.
Relativistic protocols have been developed to realise bit commitment
with, at present, conjectured security.

\section{Synopsis}
This thesis is divided into five chapters.

{\bf Introduction :}\qquad The remainder of this chapter is used to
introduce several concepts that will be important throughout the
thesis.  We discuss quantum key distribution, and in particular the
use and security of $universal_2$ hash functions as randomness
extractors in privacy amplification.  We introduce types of security,
and discuss the assumptions underlying the standard cryptographic
model, before describing the general physical frameworks in which our
protocols will be constructed.  Finally, we describe some important
cryptographic primitives.

\bigskip
{\bf The Power Of The Theory -- Strong Coin Tossing:}\qquad We
introduce the task of strong coin tossing and use it to highlight the
fact that different physical theories generate different amounts of
power in cryptography.  Our contribution here is a new protocol
applicable in the non-relativistic quantum case.  It equals the best
known bias to date for such protocols, but does so using a
conceptually different technique to that of protocols found in the
literature.  It provides a further example of the use of entanglement
as a resource.  Our protocol, Protocol \ref{colbeck_protocol}, and an
analysis of its security has appeared in \cite{Colbeck}.

\bigskip
{\bf Variable Bias Coin Tossing:}\qquad In this chapter we divide
secure two-party computations into several classes before showing that
a particular class is achievable using a quantum relativistic
protocol.  The simplest non-trivial computation in this class, a
variable bias coin toss, will be discussed in detail.  Such tasks have
not been considered in the literature to date, so this chapter
describes a new positive result in cryptography.  We prove that this
task can be achieved with unconditional security for a specified range
of biases, and with cheat-evident security for any bias.  We also
discuss two further protocols which are conjectured to be
unconditionally secure for any bias.  Most of the work covered by this
chapter has appeared in \cite{CK1}.

\bigskip
{\bf Secure Two-Party Computation:}\qquad In this chapter, we study
the remaining classes of two-party computation for which, prior to our
work, neither protocols nor no-go theorems were known.  We set up a
general model for such computations, before giving a cheating attack
which shows that a wide range of functions within these classes are
impossible to compute securely.  The culmination of these results is
given in Table \ref{fns2} (see page \pageref{fns2}).  A publication on
these results is in preparation.

\bigskip
{\bf Randomness Expansion Under Relaxed Cryptographic
Assumptions:}\qquad In the final chapter, we discuss a cryptographic
task, expanding a private random string, while relaxing the standard
assumption that each party trusts all of the devices in their
laboratory.  Specifically, we assume that all quantum devices are
provided by a malicious supplier.  We give two protocols that are
conjectured to securely perform this task.  The first allows a private
random string to be expanded by a finite amount, while the second
generates an arbitrarily large expansion of such a string.
Constructing formal security proofs for our protocols is currently
under investigation.

\section{Preliminaries}
The reader well versed in quantum information theory notions can
skip this section; for the non-specialist reader, we provide an
outline of some of the aspects that we draw upon regularly in the
forthcoming chapters.

\subsection{Local Operations}
We will often talk of local operations.  These describe any operation
that a party can do on the part of the system they hold locally, as
dictated by the laws of physics (specifically quantum mechanics).  For
quantum systems, these fall into three classes: altering the size of
the system, performing unitary operations, and performing
measurements.  A local operation can comprise any combination of
these.

\smallskip
{\bf System Size Alteration :}\qquad This is operationally trivial.  A
system can be enlarged simply by combining it with another system, and
contracted by discarding the other system.  When systems are enlarged,
the combined system then lives in the tensor product of the spaces of
the original systems, and its state is given by the tensor product of
the states of the two individual systems, which we denote with the
symbol $\otimes$.\nomenclature[X\otimes]{$\otimes$}{Tensor product}

\smallskip
{\bf Unitary Operations :}\qquad These are implemented by applying
some Hamiltonian to the system in question, for example by placing it
in an external field.  The system's dynamics follow that of the
time-dependent Schr\"odinger equation.  This defines a unitary
operation on the Hilbert space of the system.  In theory, any unitary
can be enacted on the system by varying external fields appropriately,
and applying them for the correct time periods.  However,
technologically, this represents a considerable challenge.

\smallskip
{\bf Measurement :}\qquad The most general type of measurement that
one needs to consider is a projective measurement.  Such a measurement
is defined by a set of operators, $\{\Pi_i\}$ with the property that
$\Pi_i^2=\Pi_i$, and $\sum_i \Pi_i=\openone$, the identity
operator\nomenclature[X\openone]{$\openone$}{The identity operator}.
The postulates of quantum mechanics demand the outcome of such a
measurement on a system in state $\rho$ to be $i$ with probability
${\rm tr}(\Pi_i\rho)$, and that the subsequent state of the system on
measuring $i$ is $\frac{\Pi_i\rho \Pi_i}{{\rm tr}(\Pi_i\rho)}$.  

While this is the most general type of measurement we need, it will
often be convenient to use the positive operator valued measure ({\sc
POVM})\nomenclature[zPOVM]{{\sc POVM}}{Positive Operator Valued
Measure: A measurement defined by a set of positive operators,
$\{E_i\}$ satisfying $\sum_i E_i=\openone$} formalism, whereby a
measurement is defined by a set of positive operators $\{E_i\}$ which
obey $\sum_i E_i=\openone$.  An outcome $i$ leaves the state of the
system as $\frac{\sqrt{E_i}\rho\sqrt{E_i}}{{\rm tr}(E_i\rho)}$ and
occurs with probability ${\rm tr}(E_i\rho)$.

Any {\sc POVM} can be realized as the combination of an enlargement of
the system, a unitary operation, and projective measurement (this
result is often called Neumark's theorem \cite{Peres}).  The following
is equivalent to performing the {\sc POVM} with elements $\{E_i\}$ on
a system in state $\rho$: Introduce an ancilla in state $\ket{0}$, and
perform the unitary operation, $U$, given by
$U\ket{0}\ket{\psi}=\sum_i(\openone\otimes\sqrt{E_i})\ket{i}\ket{\psi}$\footnote{This
unitary is only partially specified, since we have only defined its
operation when the first register is in the state $\ket{0}$.  However,
it is easily extended to a unitary over the entire space
\cite{Nielsen&Chuang}.}.  Then measure the projector onto
$\{\ketbra{i}{i}\otimes\openone\}$ generating the state,
$\frac{(\openone\otimes\sqrt{E_i})(\ketbra{i}{i}\otimes\rho)(\openone\otimes\sqrt{E_i})}{{\rm
tr}(E_i\rho)}$ with probability ${\rm tr}(E_i\rho)$.  On discarding
the ancillary system, this operation is equivalent to that of the {\sc
POVM}.

\bigskip
Any combination of these operations forms what we term a local
operation.  It is easy to verify that any large sequence of such
operations can be reduced to at most 4 steps: First, the system is
enlarged, then it is measured, then a unitary operation is performed
on it (possibly one that depends on the result), and finally part of
the system is discarded (again, possibly depending on the result).

Local operations have the property that for a system comprising
subsystems $Q$ and $R$, no local operation on $Q$ can be used to
affect the outcome probabilities of any measurement on system $R$,
even if the two systems are entangled.  This property means that
quantum theory does not permit superluminal signalling.

Another important property of local operations is that on average they
cannot increase entanglement between separated subsystems
\cite{Nielsen&Chuang}.

\subsubsection{Keeping Measurements Quantum}
\label{keep_quantum}
Rather than perform a measurement as prescribed by a protocol, it
turns out that one can instead introduce an ancillary register, and
perform a controlled {\sc NOT} between the system that was to be measured
and this ancilla.  The additional register in effect stores the result
of the measurement, such that if it is later measured, the entire
system collapses to that which would have been present if the
measurement had been performed as the protocol prescribed.  This
result holds for any sequence of operations that occur on the system
in the time between the controlled operation and the measurement on
the ancilla.  If one of these further operations should be dependent
on the measurement outcome, then, instead, a controlled operation is
performed with the outcome register as the control bit.  The process
of delaying a measurement in this way is often referred to as
``keeping the measurement quantum''.  Figure \ref{fig:keep_quantum}
illustrates this procedure.

\begin{figure}
\begin{center}
$(a\ket{0}+b\ket{1})\ket{\psi}\left\{{\begin{tabular}{c}$\xrightarrow{\text{probability }|a|^2}\ket{0}\ket{\psi}\xrightarrow{U_0}U_0(\ket{0}\ket{\psi})$\\
\\
$\xrightarrow{\text{probability }|b|^2}\ket{1}\ket{\psi}\xrightarrow{U_1}U_1(\ket{1}\ket{\psi})$
\end{tabular}}
\right.$

\bigskip
(a)

\bigskip
\bigskip
$(a\ket{0}\ket{\psi}\ket{0}_A+b\ket{1}\ket{\psi}\ket{1}_A)\xleftarrow{\text{controlled
    {\sc NOT}}}(a\ket{0}+b\ket{1})\ket{\psi}\ket{0}_A$

\bigskip
$\left\downarrow\begin{tabular}{c}\\$U_0\otimes\ketbra{0}{0}_A+U_1\otimes\ketbra{1}{1}_A$\\
\\ \\ \end{tabular}\right.$\\
$aU_0(\ket{0}\ket{\psi})\ket{0}_A+bU_1(\ket{1}\ket{\psi})\ket{1}_A\left\{\begin{tabular}{c}$\xrightarrow{\text{probability }|a|^2}U_0(\ket{0}\ket{\psi})\ket{0}_A$\\
\\
$\xrightarrow{\text{probability }|b|^2}U_1(\ket{1}\ket{\psi})\ket{1}_A$
\end{tabular}
\right.$

\bigskip
(b)
\end{center}
\caption{Sequence of operations for the implementation of a
  measurement in the $z$ basis on the first part of a state followed
  by a two-qubit unitary dependent on the outcome in the case (a)
  where the measurement is performed explicitly, and (b) where the
  measurement is kept at the quantum level until the end.  In the
  latter case an ancillary system indexed by $A$ has been introduced,
  and the unitary operation is now controlled on this system.  Note
  that the end result is the same in both cases.}
\label{fig:keep_quantum}
\end{figure}

\subsection{Distinguishing Quantum States}
The problem of how best to distinguish quantum states dates back
several decades.  For a good account see for example \cite{Chefles}.

Alice is to prepare a state
$\rho\in\{\rho_0,\rho_1,\ldots\rho_{n-1}\}$ and send it to Bob.  She
is to choose $\rho_i$ with probability $\eta_i$.  Bob, who knows the
identity of the states $\{\rho_0,\rho_1,\ldots\rho_{n-1}\}$ and their
probability distribution, is required to guess the value of $i$ by
performing any operations of his choice on $\rho$.  It is well known
that Bob cannot guess the value of $i$ with guaranteed success, unless
the states are orthogonal.

There are two flavours to this problem.  One, sometimes called quantum
hypothesis testing, involves maximizing the probability of guessing
the state correctly.  The other, unambiguous state discrimination,
seeks to never make an error.  This can be achieved only if the states
to be distinguished are linearly independent, and at the expense that
(unless the states are orthogonal) sometimes an inconclusive outcome
will be returned.  It is the first of these two problems that will be
relevant to us.

It follows from the discussion of local operations in the previous
section, and the fact that we don't need the system after the
measurement that it is sufficient for Bob to simply do a {\sc POVM} on
$\rho$.  This {\sc POVM} should have $n$ outcomes, with outcome $i$
corresponding to a best guess of Alice's preparation being
$\rho_i$\footnote{It is clear that we can always put the optimal
strategy in this form.  For a general {\sc POVM}, each element can be
associated with a state that is the best guess for the outcome
corresponding to that element.  If two elements have the same best
guess, we can combine their {\sc POVM} elements by addition to give a new
{\sc POVM}.  This generates a {\sc POVM} with at most $n$ outcomes.  If there are
fewer than $n$, then we can always pad the {\sc POVM} with zero operators.
A simple relabelling then ensures that outcome $i$ corresponds to a
best guess of $\rho_i$.}.  The task is to maximize
\begin{eqnarray}
\label{maxp}
\sum_i\eta_i{\rm tr}(E_i\rho_i)
\end{eqnarray}
over all {\sc POVM}s $\{E_i\}$.  

In general, it is not known how to obtain an analytic solution to this
problem \cite{Chefles}, although numerical techniques have been
developed \cite{JRF}.  However, a solution is known for the case
$n=2$.  In Appendix \ref{AppA}, we give a proof that in this case, the
maximum probability is $\frac{1}{2}\left(1+{\rm
tr}\left|\eta_0\rho_0-\eta_1\rho_1\right|\right)$ \cite{Helstrom}.  In
the case $\eta_0=\eta_1=\frac{1}{2}$, this expression is usually
written as $\frac{1}{2}(1+D(\rho_0,\rho_1))$, where
$D(\rho_0,\rho_1)\equiv\frac{1}{2}{\rm tr}|\rho_0-\rho_1|$ is the
trace distance between the two density matrices. \footnote{This is a
generalization of the classical distance between probability
distributions, for which we also use the symbol $D$, see Section
\ref{classdist}.}

\label{tracedist}\nomenclature[aD]{$D(\rho_0,\rho_1)$}{The trace
distance between density matrices $\rho_0$ and $\rho_1$, defined by
$\frac{1}{2}{\rm tr}\mid\rho_0-\rho_1\mid$}

Other cases for which analytic results are known involve cases where
the set of states to be distinguished are symmetric and occur with a
uniform probability distribution.  In such cases, the so-called square
root measurement is optimal \cite{HJW,HW}.  Another result that we
will find useful is the following theorem.

\begin{theorem} \cite{Holevo,Yuen&,Helstrom}
Consider using a set of $M$ measurement operators, $\{E_j\}$, to
discriminate between a set of $M$ states, $\{\rho_j\}$, which occur
with prior probabilities, $\{\eta_j\}$, where the outcome corresponding
to operator $E_j$ indicates that the best guess of the state is
$\rho_j$.  The set $\{E_j\}$ is optimal if and only if
\begin{eqnarray}
\label{con1}
E_j\left(\eta_j\rho_j-\eta_l\rho_l\right)E_l&=&0\;\;\;\forall\;j,l\\ 
\label{con2}
\sum_j E_j\eta_j\rho_j-\eta_l\rho_l&\geq &0\;\;\;\forall\; l.  
\end{eqnarray} 
\label{optimal_cond}
\end{theorem}

\subsection{Entanglement, Bell's Theorem And Non-locality}
\label{Bell}
The title of this section could easily be that of a book.  A wealth of
previous research exists in this area, and a number of debates still
rage about the true nature of non-locality; some of which date back to
the famous Einstein-Podolsky-Rosen paper of 1935 \cite{EPR}, or even
before.  There is no evidence to date that contradicts the predictions
of quantum theory, but some find its philosophical consequences so
outrageous that they seek alternative theories that are more closely
aligned with what is ultimately their own beliefs.  Furthermore, there
exist experimental loopholes which sustain the belief that quantum
theory could be local\footnote{No experiment to date has properly
ensured space-like separation of measurements.}.  In this section I
briefly discuss some aspects of what is often referred to as quantum
non-locality.

The term entanglement describes the property of particles that leads
to their behaviour being linked even if they have a large separation.
Consider a pair of electrons created in a process which guarantees
that their spins are opposite.  According to quantum theory, until
such a time that a measurement is made, the state of the entire system
is
$\frac{1}{\sqrt{2}}\left(\ket{\uparrow_1\downarrow_2}+\ket{\downarrow_1\uparrow_2}\right)$.
Measuring either particle in the $\{\uparrow,\downarrow\}$ basis
causes the state of the entire two particle system to collapse to
either $\ket{\uparrow_1\downarrow_2}$ or
$\ket{\downarrow_1\uparrow_2}$ with probability half each.  What is
philosophically challenging about this is that measuring the first
particle affects the properties of the second particle {\it
instantaneously}.  If one were to perform the experiment, the natural
conclusion would be that each particle was assigned to be in either
the $\ket{\uparrow}$ or $\ket{\downarrow}$ state when they were
created, and hence such results are not at all surprising: we simply
didn't know how the state was assigned until the measurement.  A
famous analogy is that of Bertlmann's socks \cite{Bell}.  Which colour
he wears on a particular day is unpredictable, but his socks are never
coloured the same.  On seeing that his right sock is pink, {\it
instantaneously} one knows that his left is not.  Based on this
experiment alone, both hypotheses are tenable.

It was Bell who realized a way in which these hypotheses can be
distinguished \cite{Bell}.  He developed an inequality that any local
realistic theory must obey\footnote{A local theory is one in which no
influence can travel faster than light; a realistic theory is one in
which values of quantities exist prior to being measured.}, before
showing that suitable entangled quantum systems can violate this.  The
most common recasting of his ideas is the {\sc CHSH} inequality (named after
Clauser, Horne, Shimony and Holt\nomenclature[zCHSH]{{\sc CHSH}}{Clauser,
Horne, Shimony and Holt, the collaboration that developed a simple
version of Bell's inequality}).  Consider the following abstract
scenario.  Two spatially separated boxes each have two possible
settings (inputs), and two possible outputs, $+1$ and $-1$.  We label
the inputs $P$ and $Q$ for each box, and use $P_i\in\{1,-1\}$ and
$Q_i\in\{1,-1\}$ to denote the output of the box for input $P$ or $Q$
respectively, with index $i\in\{1,2\}$ corresponding to the box to
which we are referring.  (In a quantum mechanical context, the inputs
represent choices of measurement basis, and the outputs the
measurement result.)  The {\sc CHSH} test involves the quantity $\langle
P_1P_2+P_1Q_2+Q_1P_2-Q_1Q_2\rangle$, where $\langle X\rangle$ denotes
the expectation value of random variable $X$\nomenclature[xA]{$\langle
  X\rangle$}{The expectation value of random variable $X$}.  The following theorem
gives the limit of this quantity for local hidden variable theories.
\begin{theorem}
There is no assignment of values $\{P_1,P_2,Q_1,Q_2\}\in\{\pm 1,\pm
1,\pm 1,\pm 1\}$ (and hence no local hidden variable theory), for
which $\langle P_1P_2+P_1Q_2+Q_1P_2-Q_1Q_2\rangle > 2$
\end{theorem}
Nevertheless, values as high as $2\sqrt{2}$ are possible using quantum
systems \cite{Cirelson} although these fall short of the maximum
algebraic limit of 4.  The achievability of $2\sqrt{2}$ rules out the
possibility of a local hidden variable theory for explaining the data
(modulo a few remaining loopholes, see for example
\cite{Aspect,BCHKP,Kent_loophole} for discussions).  A non-local but
realistic theory can evade the theorem by allowing the value of
quantities defined on one particle to change when a measurement is
made on another, no matter how separated the particles are.  It is not
straightforward to drop the realistic assumption while keeping the
theory local, since the concept of locality itself is inherently
linked with realism.
Hence it is common terminology in the literature to use the phrase
``non-local effects'' to allude to violations of Bell-type
inequalities such as that of {\sc CHSH}.  

\subsection{Entropy Measures}
\label{entropy_measures}
\subsubsection{Random Events}
When we discuss random events, we assume that they occur according to
a pre-defined ensemble of possible outcomes and their associated
probabilities.  Event $X$ is a single instance drawn from the ensemble
$\{1,2,\ldots,|X|\}$ with probabilities
$\{P_X(1),P_X(2),\ldots,P_X(|X|)\}$.  We call this probability
distribution $P_X$.  \nomenclature[aP_X]{$P_X$}{The probability
distribution of a random variable $X$} The terminology $X=x$ refers to
a single instance drawn from this distribution taking the value $x$.
One similarly defines distributions over more than one random
variable.  For instance, $P_{XY}$ is the joint distribution of $X$ and
$Y$, and $P_{X|Y=y}$ is the distribution of $X$ conditioned on the
fact that $Y$ takes value $y$.

\subsubsection{Shannon Entropy}
It was Shannon who pioneered the mathematical formulation of
information \cite{Shannon}.  In essence his insight was that an event
that occurs with probability $p$ could be associated with an amount of
information $-\log p$. \footnote{In information theory, as in this
thesis, all logarithms are taken in base 2 and hence entropies and
related quantities are measured in
bits.}\nomenclature[xlog]{$\log$}{The binary logarithm function}
Consider many independent repetitions of random event $X$.  The
average information revealed by each instance of $X$ is given by the
Shannon entropy of $X$ defined as follows.
\begin{definition}
The Shannon entropy associated with an event $x$ drawn from random
distribution $X$ is $H(X)\equiv\sum_{x\in X}-P_X(x)\log P_X(x)$.
\end{definition}
Likewise, one can define conditional Shannon entropies.  $H(X|Y=y)$
denotes the Shannon entropy of $X$ given $Y$.  It measures the average
amount of information one learns from a single instance of $X$ if one
possesses string $y\in Y$, where $X,Y$ are chosen according to joint
distribution $P_{XY}$.  One can average this quantity to form
$H(X|Y)$, the conditional Shannon entropy.
\begin{definition}
The conditional Shannon entropy of an event $X$ given $Y$ is defined
by $H(X|Y)\equiv\sum_{x\in X, y\in Y}-P_Y(y)P_{X|Y=y}(x)\log P_{X|Y=y}(x)$.
\end{definition}
This leads one to define the mutual Shannon information between $X$
and $Y$ by $I(X:Y)\equiv H(X)-H(X|Y)=H(Y)-H(Y|X)$.  In some sense,
this is the amount of information in common to the two strings $X$ and
$Y$.

Shannon information was first used to solve problems of compression,
and communication over a noisy channel, as given in the following
theorems \cite{Shannon}.
\begin{theorem} (Source coding theorem)
Consider a source emitting independent and identically distributed
{\sc (IID)}  \nomenclature[ziid]{{\sc IID}}{Independent and Identically
Distributed} random variables drawn from distribution $P_X$.  For any
$\epsilon>0$ and $R>H(X)$, there exists an encoder such that for
sufficiently large $N$, any sequence drawn from $P_X^N$ can be
compressed to length $NR$, and a decoder such that, except with
probability $<\epsilon$, the original sequence can be restored from
the compressed string.  
\end{theorem}
Furthermore, if one tries to compress the same source using $R<H(X)$
bits per instance, it is virtually certain that information will be
lost.

\begin{definition}
For a discrete, memoryless channel, in which Alice sends a random
variable drawn from $X$ to Bob who receives $Y$, the {\it channel
capacity} is defined by $C\equiv \max_{P_X}I(X:Y)$.
\end{definition}
\begin{theorem} (Noisy channel coding theorem)
Consider Alice communicating with Bob via a discrete memoryless channel
which has the property that if Alice draws from an {\sc IID} source $X$,
Bob receives $Y$.  For any $\epsilon>0$ and $R<C$, for large enough
$N$, there exists an encoding of length $N$ and a decoder such that
$\geq RN$ bits of information are conveyed by the channel for each
encoder-channel-decoder cycle, except with probability $<\epsilon$.
\end{theorem}

Notice that in the noisy channel coding theorem, the channel is
memoryless, and Alice has an {\sc IID} source.  In other words, all uses
of the channel are independent of one another.  This is the situation
in which Shannon information is useful.  However, in cryptographic
scenarios where the channel may be controlled by an eavesdropper, such
an assumption is not usually valid.  Instead, other entropy measures
have been developed that apply for these cases, as discussed in the
next section.

The relative entropy, which is a measure of the closeness of two
probability distributions, will also be of use.
\begin{definition}
The relative entropy of $P_X$ and $Q_X$ is given by,
\begin{equation}
H(P_X||Q_X)\equiv\sum_x P_X(x)\log\frac{P_X(x)}{Q_X(x)}.
\end{equation}
\end{definition}

\subsubsection{Beyond Shannon Entropy}
\label{beyondshan}
R\'enyi \cite{Renyi} introduced the following generalization of the
Shannon entropy.
\begin{definition}
The R\'enyi entropy of order $\alpha$ is defined by
\nomenclature[aH_]{$H_{\alpha}(X)$}{The R\'enyi entropy of order
  $\alpha$ of random variable $X$.  Special cases are the max-entropy,
  $H_0(X)$, and the min-entropy, $H_{\infty}(X)$}
\begin{equation}
H_{\alpha}(X)\equiv\frac{1}{1-\alpha}\log\sum_{x\in X}P_X(x)^{\alpha}.
\end{equation}
\end{definition}
We have, $H_1(X)\equiv\lim_{\alpha\rightarrow 1}H_{\alpha}(X)=H(X)$.
Two other important cases are $H_0(X)=\log|X|$ and
$H_{\infty}(X)=-\log\max_{x\in X}P_X(x)$.  A useful property is that,
for $\alpha\leq\beta$, $H_{\alpha}(X)\geq H_{\beta}(X)$.

$H_{\infty}(X)$ is sometimes called the min-entropy of $X$.  We will
see that it is important for privacy amplification.  There, the
presence of an eavesdropper means that it no longer suffices to
consider each use of the channel as independent.  The min-entropy
represents the maximum amount of information that could be learned
from the event $X$, so describes the worst case scenario.  In a
cryptographic application, one wants to be assured security even in
the worst case.

In general, R\'enyi entropies are strongly
discontinuous. \footnote{Consider the two distributions $P_X$ and
$Q_X$ defined on $x\in\{1,\ldots,2^n\}$.  Take
$P_X(x=1)=2^{-\frac{n}{4}}$, $P_X(x\neq
1)=\frac{1-2^{-\frac{n}{4}}}{2^n-1}$, and $Q_X$ to be the uniform
distribution.  Comparing min-entropies gives
$H^Q_{\infty}(X)-H^P_{\infty}(X)=\frac{3n}{4}$.  In the large $n$
limit, the two distributions have distance $\approx 2^{-\frac{n}{4}}$
(see Definition \ref{class_dist}), which is exponentially small, while
the difference in min-entropies becomes arbitrarily large.}  However,
smoothed versions of these quantities have been introduced which
remove such discontinuities.  In essence, these smoothed quantities
involve optimizing such quantities over a small region of probability
space.  They have operational significance in cryptography in that
they provide the relevant quantities for information reconciliation
and privacy amplification as will be discussed in Section \ref{QKD}.
It will be the conditional versions of these entropies that concern
us, hence we provide a definition of these directly.
\begin{definition} \cite{RennerWolf}
For a distribution $P_{XY}$, and smoothing parameter $\epsilon>0$, we
define the following smoothed R\'enyi entropies:
\begin{eqnarray}
\label{H_0^e}
H_0^{\epsilon}(X|Y)&\equiv&\min_{\Omega}\max_y\log|\{x:P_{X\Omega|Y=y}(x)>0\}|\\
\label{def_1}
H_{\infty}^{\epsilon}(X|Y)&\equiv&\max_{\Omega}\left(-\log\max_y\max_x P_{X\Omega|Y=y}(x)\right),
\end{eqnarray}
where $\Omega$ is a set of events with total probability at least
$1-\epsilon$, and  $P_{X\Omega|Y=y}(x)$ denotes the probability that $X$
takes the value $x$, and event $\Omega$ occurs given that $Y$ takes
the value $y$.
\end{definition}
More generally, the smooth R\'enyi entropy of order $\alpha$ can be
defined \cite{RennerWolf}, but since, up to an additive constant these
equal either $H_0^{\epsilon}$ (for $\alpha<1$) or
$H_{\infty}^{\epsilon}$ (for $\alpha>1$), we ignore such quantities in
our discussion.  It is also worth noting that for a large number of
independent repetitions of the same experiment, the R\'enyi entropies
tend to the Shannon entropy, that is,
\begin{equation}
\lim_{\epsilon\rightarrow 0}\lim_{n\rightarrow\infty}\frac{H_{\alpha}^{\epsilon}(X^n|Y^n)}{n}=H(X|Y).
\end{equation}

\subsubsection{Quantum Entropic Quantities}
The entropy of a quantum state, $\rho$, is commonly expressed using
the von Neumann entropy,
\begin{equation}
H(\rho)\equiv-\text{tr}\left(\rho\log\rho\right).
\end{equation}
This is the quantum analogue of the Shannon entropy, and is equal to
the Shannon entropy of the probability distribution formed if $\rho$
is measured in its diagonal basis.  Hence, if $\rho$ is classical,
that is $\rho=\sum_x P_X(x)\ketbra{x}{x}$, for some orthonormal basis,
$\{\ket{x}\}$, then $H(\rho)=H(X)$.

In a similar way, one defines the quantum relative entropy between the
states $\rho$ and $\sigma$ by,
\begin{equation}
H(\rho||\sigma)=\text{tr}(\rho\log\rho)-\text{tr}(\rho\log\sigma).
\end{equation}
This again reduces to the classical version if $\sigma$ and $\rho$
have the same diagonal basis.  It has the following important
property \cite{Nielsen&Chuang}.
\begin{theorem}
(Klein's inequality) $H(\rho||\sigma)\geq 0$, with equality if and
only if $\rho=\sigma$.
\end{theorem}

The conditional von Neumann entropy can be defined
by\footnote{Alternative definitions are sometimes given, e.g.\ in
  \cite{Nielsen&Chuang}, which do not contain the
  $H(\rho_B||\sigma_B)$ part \cite{Renner}}
\begin{eqnarray}
H(\rho_{AB}|\sigma_B)&\equiv&-{\text
  tr}(\rho_{AB}(\log\rho_{AB}-\log\openone_A\otimes\sigma_B))\\
&=&H(\rho_{AB})-H(\rho_B)-H(\rho_B||\sigma_B),
\end{eqnarray}
where $\rho_B=\text{tr}_A\rho_{AB}$.  We can also define a version
for the extremal case $\sigma_B=\rho_B$,
\begin{equation}
H(\rho_{AB}|B)\equiv H(\rho_{AB}|\rho_B).
\end{equation}
Likewise, we define quantum min-entropies,
\begin{eqnarray}
H_{\infty}(\rho_A)&\equiv&-\log\lambda_{\max}(\rho_A),\\
H_{\infty}(\rho_{AB}|\sigma_B)&\equiv&-\log\lambda,\\
H_{\infty}(\rho_{AB}|B)&\equiv&\min_{\sigma_B}H_{\infty}(\rho_{AB}|\sigma_B),
\end{eqnarray}
where $\lambda_{\max}(\rho)$ is the largest eigenvalue of
$\rho$, and $\lambda$ is the minimum real number such that
$\lambda\openone_A\otimes\sigma_B-\rho_{AB}\geq 0$.  

\begin{lemma}
Consider the case where system $A$ is classical, that is,
$\rho_{AB}=\sum_iP_I(i)\ketbra{i}{i}\otimes\rho_B^i$.  For such
states, 
\begin{enumerate}[(a)]
\item \label{pointa} $H_{\infty}(\rho_{AB}|B)\geq 0$, and
\item \label{pointb} $H_{\infty}(\rho_{AB}|B)=0$ if there exists some
$j$ such that $\rho_B^j$ is not contained within the support of
$\{\rho_B^i\}_{i\neq j}$.
\end{enumerate}
\end{lemma}
\begin{proof}
For a state of this form, $\lambda\openone_A\otimes\sigma_B-\rho_{AB}$
is block diagonal with block entries $\lambda\sum_{i\neq
j}P_I(i)\rho_B^i+(\lambda-1)P_I(j)\rho_B^j$, for some $j$.  If
$\lambda>1$, these are positive for all $j$, from which (\ref{pointa})
then follows using the definition of $H_{\infty}(\rho_{AB}|B)$.

In order that $H_{\infty}(\rho_{AB}|B)=0$, there must exist a $j$ such
that for all $\epsilon\equiv 1-\lambda>0$, $(1-\epsilon)\sum_{i\neq
j}P_I(i)\rho_B^i-\epsilon P_I(j)\rho_B^j$ is negative.  This implies
that for some $j$, $\rho_B^j$ is not contained within the support of
$\{\rho_B^i\}_{i\neq j}$, hence establishing (\ref{pointb}).
\end{proof}
Recall from Section \ref{beyondshan} that the classical min-entropy
can be associated with information content in the worst case scenario.
The same is true here.  In the extremal case, there exists some $j$
such that $\rho_B^j$ is not contained within the support of
$\{\rho_B^i\}_{i\neq j}$.  Then there exists some measurement on system
$B$ for which at least one outcome identifies the state of system $A$
precisely, and the corresponding min-entropy is 0.

One can also define smoothed versions of these entropies.  The
$\epsilon$-smooth min-entropy of $\rho_{AB}$ given $\sigma_B$ is given
by 
\begin{equation}
H_{\infty}^{\epsilon}(\rho_{AB}|\sigma_B)\equiv\min_{\bar{\rho}_{AB}}H_{\infty}(\bar{\rho}_{AB}|\sigma_B)
\end{equation}
where the minimum is over the set of operators satisfying
$D(\bar{\rho}_{AB},\rho_{AB})\leq\epsilon$, with ${\text
tr}(\bar{\rho}_{AB})\leq 1$.  In other words, the smoothed version of the
min-entropy is formed by considering density matrices close to $\rho_
{AB}$.  This is in direct analogy with the classical case, where
nearby probability distributions were considered using parameter,
$\Omega$.

\nomenclature[aInf]{$H_{\infty}^{\epsilon}(\rho_{AB}\arrowvert B)$}{The
$\epsilon$-smooth min-entropy of $\rho_{AB}$ given the system in
$\mathcal{H}_B$}
We also define
\begin{equation}
\label{smq}
H_{\infty}^{\epsilon}(\rho_{AB}|B)\equiv\min_{\sigma_B}H_{\infty}^{\epsilon}(\rho_{AB}|\sigma_B),
\end{equation}
where we give the second Hilbert space of the system to the eavesdropper.

\section{Quantum Key Distribution}
\label{QKD}
\nomenclature[zQKD]{{\sc QKD}}{Quantum Key Distribution} Quantum key
distribution is one of the big success stories of quantum information
theory.  It allows two separated agents, Alice and Bob, to generate a
shared random string about which an eavesdropper, Eve, has no
information.  Such a random string can form the key bits of a one-time
pad, for example, and hence allow secure communication between Alice
and Bob.  This task is known to be impossible classically, without
making computational assumptions, and is historically the first
instance of a quantum protocol.  Really the task should be called key
expansion, since an initial shared random string is needed for
authentication.  We avoid this distinction by giving Alice and Bob
shared authenticated classical channels (upon which Eve can only
listen, but not modify), and a completely insecure quantum channel.

Eve can always perform a denial of service attack, blocking
communication between Alice and Bob.  However, we assume instead, that
she wants to learn some part of their message.  There are several
steps common to most key distribution protocols.  Exchange of
information over the insecure channel, {\it reconciliation} of this
information (i.e.\ error correction) and then {\it privacy
amplification} (i.e.\ reducing Eve's information to a negligible
amount).

Often, the quantum part of the protocol is restricted to the first
step.  A quantum channel is used to set up correlated random strings,
after which classical reconciliation and privacy amplification
procedures are used.  In essence, the security of the protocol relies
on the fact that an eavesdropper can neither copy a quantum state, nor
learn anything about it without disturbance.  We will not discuss the
alternative approach, where these latter procedures are also quantum,
and at the end of the protocol Alice and Bob possess shared singlets.
For concreteness, we now outline Bennett and Brassard's 1984 protocol,
{\sc BB}84\nomenclature[zBB84]{{\sc BB}84}{Bennett and Brassard's 1984 quantum key
distribution scheme}.

\begin{protocol}
\label{prot_BB84}
Define 2 bases, $\mathcal{B}_0\equiv\{\ket{0},\ket{1}\}$, and
$\mathcal{B}_1\equiv\{\ket{+},\ket{-}\}$, where
$\ket{\pm}\equiv\frac{1}{\sqrt{2}}\left(\ket{0}\pm\ket{1}\right)$.

\begin{enumerate}
\item Alice selects a set of bits uniformly at random, $\{x_i\}$,
  along with a uniform random set of bases $\{A_i\}$, where
  $x_1\in\{0,1\}$, and $A_i\in\{\mathcal{B}_0,\mathcal{B}_1\}$.
\item She encodes bit $x_i$ using basis $A_i$, where $0$ is
  encoded as $\ket{0}$ or $\ket{+}$, and we use $\ket{1}$ or $\ket{-}$
  to encode $1$.
\item Alice sends the encoded qubits to Bob through the quantum channel.
\item Bob selects a random set of bases $\{B_i\}$, with
  $B_i\in\{\mathcal{B}_0,\mathcal{B}_1\}$, and measures the $i$th
  incoming qubit in basis $B_i$.
\item Once Bob has made all his measurements, Alice announces the
  bases she used over the public channel, and Bob does the same.
\item \label{sift} 
  ({\it sifting}) Alice and Bob discard all the measurements made
  using different bases.  On average half the number of qubits sent by
  Alice remain.  In the absence of noise and an eavesdropper the leftover
  strings are identical.
\item \label{error estimate} Alice and Bob compare the values of a
  subset of their bits, selected at random.  This allows them to
  estimate the error rate.  If too high, they abort.
\item \label{recon&amp} Alice and Bob perform reconciliation and
  privacy amplification on the remaining bits.
\end{enumerate}
\end{protocol}

\subsection{Information Reconciliation}
Information reconciliation is error correction.  In essence, Alice
wants to send error correction information to Bob so that he can make
his partially correlated string identical to hers.  Since this
information will be sent over a public channel on which Eve has full
access, Alice wishes to minimize the error correction information at
the same time as providing a low probability of non-matching strings
in the result.

The task of information reconciliation can be stated as follows.
Alice has string $X$ and Bob $Y$, these being chosen with joint
distribution $P_{XY}$.  Alice also possesses some additional
independent random string $R$.  What is the minimum length of string
$S=f(X,R)$ that Alice can compute such that $X$ is uniquely obtainable
by Bob using $Y$, $S$ and $R$, except with probability less than
$\epsilon$?

In \cite{RennerWolf}, this quantity is denoted
$H_{enc}^{\epsilon}(X|Y)$ and is tightly bounded by the relation
\begin{equation}
H_0^{\epsilon}(X|Y)\leq H_{enc}^{\epsilon}(X|Y)\leq
H_0^{\epsilon_1}(X|Y)+\log\frac{1}{\epsilon_2},
\end{equation}
where $\epsilon_1+\epsilon_2=\epsilon$.

It is intuitively clear why $H_0^{\epsilon}(X|Y)$ is the correct
quantity.  Recall the definition (\ref{H_0^e})
\begin{equation*}
H_0^{\epsilon}(X|Y)\equiv\min_{\Omega}\max_y\log|\{x:P_{X\Omega|Y=y}(x)>0\}|,
\end{equation*}
where $\Omega$ is a set of events with total probability at least
$1-\epsilon$.  The size of the set of strings $x$ that could have
generated $Y=y$ given $\Omega$ is $|\{x:P_{X\Omega|Y=y}(x)>0\}|$.
Alice's additional information needs to point to one of these.  It
hence requires $\log |\{x:P_{X\Omega|Y=y}(x)>0\}|$ bits to encode.
Since Alice does not know $y$, she must assume the worst, hence we
maximize on $y$.  Furthermore, since some error is tolerable, we
minimize on $\Omega$, by cutting away unlikely events from the
probability distribution.

\begin{figure}
\begin{center}
\includegraphics[width=\textwidth]{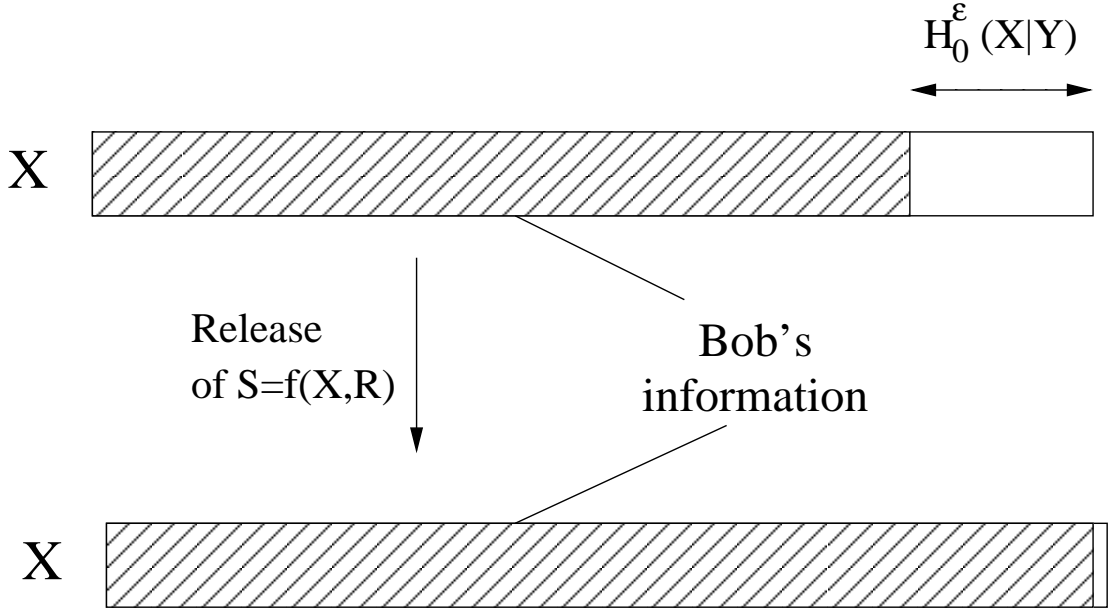}
\caption{Schematic showing information reconciliation.  The release of
$S=f(X,R)$ reduces Bob's uncertainty on Alice's string, $X$, to a
negligible amount.}
\label{fig:reconciliation}
\end{center}
\end{figure}


\subsection{Privacy Amplification}
\label{priv_amp}
In essence, this task seeks to find the maximum length of string Alice
and Bob can form from their shared string such that Eve has no
information on this string.

This task can be stated more formally as follows.  Alice possesses
string $X$ and Eve $Z$\footnote{For the moment we consider the case
where Eve's information is classical.}, distributed according to
$P_{XZ}$.  Alice also has some uncorrelated random string $R$.  What
is the maximum length of a binary string $S=f(X,R)$, such that for a
uniform random variable $U$ that is independent of $Y$ and $R$, we
have $S=U$, except with probability less than $\epsilon$?

This quantity, denoted $H_{ext}^{\epsilon}(X|Z)$, has been bounded
\cite{RennerWolf} as follows:
\begin{equation}
\label{amp_eqn}
H_{\infty}^{\epsilon_1}(X|Z)-2\log\frac{1}{\epsilon_2}\leq
H_{ext}^{\epsilon}(X|Z)\leq H_{\infty}^{\epsilon}(X|Z),
\end{equation}
where $\epsilon_1+\epsilon_2=\epsilon$.

Let us give a brief intuition as to why this is the correct quantity.
Recall the definition
\begin{equation*}
H_{\infty}^{\epsilon}(X|Z)\equiv\max_{\Omega}\left(-\log\max_z\max_x P_{X\Omega|Z=z}(x)\right)
\end{equation*}
Given that she holds $Z=z$, the minimum amount of information Eve
could learn from $X$ is $-\log \max_x P_{X|Z=z}(x)$.  We can think of
this as the information in $X$ that is independent of $Z=z$ in the
worst case.  If we minimize this quantity on $z$, which corresponds to
the worst possible case for Alice, we have $-\log \max_z \max_x
P_{X|Z=z}(x)$.  In many scenarios there is a small probability that an
eavesdropper can guess Alice's string perfectly, in which case this
quantity is zero.  We hence maximize over sets of events $\Omega$ that
have total probability at least $1-\epsilon$.  This introduces some
probability of error, but gives a significant increase in the size of
the min-entropy over its non-smoothed counterpart.

\subsubsection{Extractors And $universal_2$ Hashing}
\label{extractors}
Privacy amplification is often studied using the terminology of
extractors.  Roughly speaking, an extractor is a function that takes
as input $X$, along with some additional uniformly distributed, and
uncorrelated randomness, $R$, and returns $S=f(X,R)$ that is almost
uniformly distributed.  For a strong extractor, we have the additional
requirement that $S$ is independent of $R$.  After defining a distance
measure for classical probability distributions, we give a formal
definition of a strong extractor.
\begin{definition}
\label{classdist}
The classical distance\footnote{This is a special case of the trace
distance defined in Section \ref{tracedist} and hence we denote it by
the same symbol, $D$.  It is related to the maximum probability of
successfully distinguishing the two distributions in the same way that
the trace distance of two quantum states is related to the maximum
probability of distinguishing them (cf.\ Appendix \ref{AppA}).} between two
probability distributions $P$ and $Q$ defined on the domain $X$ is
given by
\begin{equation}
\label{class_dist}
D(P_X,Q_X)\equiv\frac{1}{2}\sum_{x\in X}\left|P_X(x)-Q_X(x)\right|.
\end{equation}
\end{definition}
\begin{definition}
\label{extractor}
Let $U_S$ be the uniform distribution over the members of $S$.  A
strong $(\tau,\kappa,\epsilon)$-extractor is a function that takes
inputs $X$ and $R$ and returns $S=f(X,R)$, where $|S|=2^{\tau}$, such
that if $H_{\infty}(X)\geq\kappa$, and $R$ is uniformly distributed
and independent of $X$, then $D(P_{SR},P_{U_S}P_R)\leq\epsilon$, where
$P_{U_S}$ is the uniform distribution on $S$.
\end{definition}
A small distance, $\epsilon$, between two probability distributions is
essentially the same as saying that the two distributions are equal,
except with probability $\epsilon$.

As an example of an extractor, consider a $universal_2$ hash function
\cite{CW,WC}.
\begin{definition}
A set of hash functions, $F$ from $X$ to $S$ is $universal_2$ if, for
any distinct $x_1$ and $x_2$ in $X$, then, for some function $f\in F$
picked uniformly at random, the probability that $f(x_1)=f(x_2)$ is at
most $\frac{1}{|S|}$.
\end{definition}
We now show that this satisfies the necessary requirements for a strong
extractor.

Consider some probability distribution $P_V$ on $V$, and take $U_V$ to
be the uniform distribution over the same set.  We have
\begin{eqnarray}
\nonumber
D(P_V,P_{U_V})&=&\frac{1}{2}\sum_{v\in
  V}\left|P_V(v)-\frac{1}{|V|}\right|\\
\nonumber
&\leq&\frac{\sqrt{V}}{2}\sqrt{\sum_{v\in V}(P_V(v)-\frac{1}{|V|})^2}\\
&=&\frac{\sqrt{V}}{2}\sqrt{\sum_{v\in V}P_V(v)^2-\frac{1}{|V|}},
\label{disteqn}
\end{eqnarray}
where we have used the Cauchy-Schwarz inequality.
\footnote{The Cauchy-Schwarz inequality states that $|{\bf x}.{\bf
y}|^2\leq |{\bf x}|^2|{\bf y}|^2$.  The version we use is for the case
{\bf y}=(1,1,\ldots, 1).}  Hence, the collision probability,
$P_C(V)\equiv \sum_{v\in V}P_V(v)^2$, i.e.\ the probability that two
events each drawn from $P_V$ are identical, allows us to bound the
distinguishability of $P_V$ from uniform.

To show that a $universal_2$ hash function is an extractor, we take
$V$ to be $SR$.  We have
\begin{eqnarray}
\nonumber
P_C(SR)
&=&P_C(R)P_C(S)\\
\nonumber 
&\leq&\frac{1}{|R|}\left(P_C(X)+\frac{1}{|S|}\right)\\
&=&\frac{1}{|R|}\left(2^{-H_2(X)}+\frac{1}{|S|}\right),
\end{eqnarray}
where the inequality follows from the definition of a $universal_2$
hash function.
Thus, using (\ref{disteqn}), we have
\begin{equation}
D(P_{SR},P_{U_S}P_R)\leq\frac{\sqrt{|S|}}{2}2^{-\frac{1}{2}H_2(X)}.
\end{equation}

Since $H_2(X)\geq H_{\infty}(X)$, we have shown that a $universal_2$
hash function is a strong
$(\tau,\kappa,\frac{1}{2}2^{\frac{1}{2}(\tau-\kappa)})$-extractor.
Alternatively, if we wish to use $universal_2$ hashing, and have
$H_{\infty}(X)\geq\kappa$, then to ensure that the output is
$\epsilon$-close to the uniform distribution, we can extract a
string whose length is bounded by
\begin{equation}
\label{maxlength}
\tau\leq\kappa-2\log\frac{1}{2\epsilon}.
\end{equation}
The use of a hash function for privacy amplification is illustrated in
Figure \ref{fig:priv_amp}.

\begin{figure}
\begin{center}
\includegraphics[width=\textwidth]{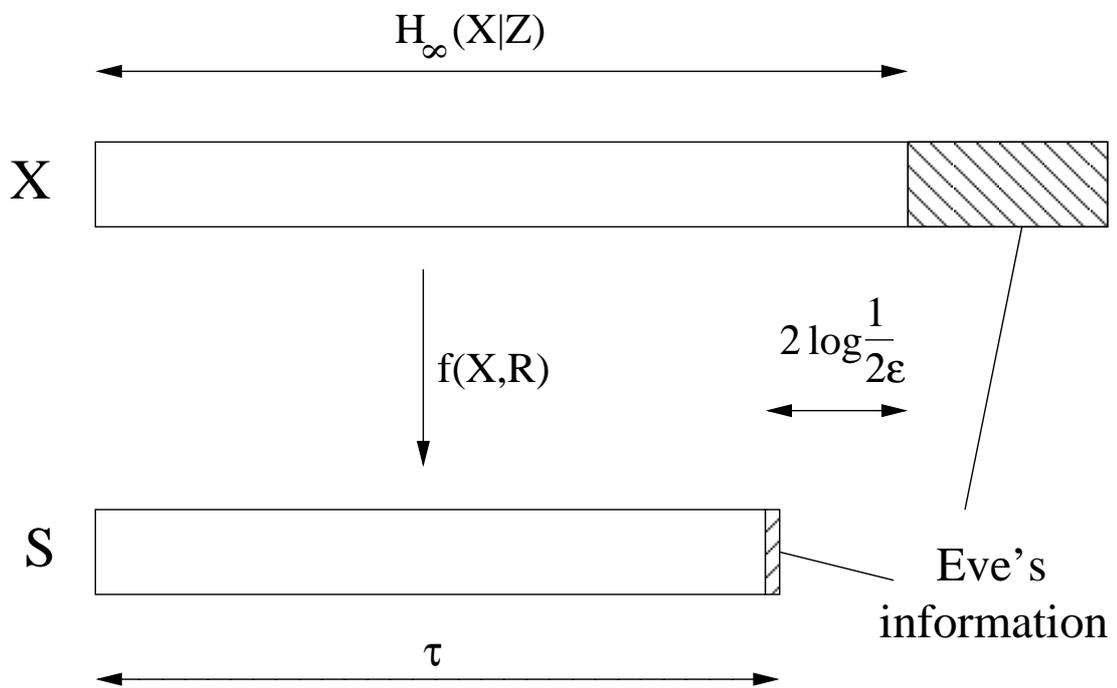}
\caption{Schematic showing privacy amplification of string $X$ to form
  $S$ using a $universal_2$ hash function.}
\label{fig:priv_amp}
\end{center}
\end{figure}

The drawback of $universal_2$ hashing is that in order to pick a
function randomly from the members of a $universal_2$ set requires a
long random string, $R$.  Many $universal_2$ classes require the
string $R$ to have length equal to that of the string being hashed,
although more efficient classes are known for cases in which the final
string is very short compared to the initial one \cite{CW,WC}.  For
more general extractors, constructions which require a much shorter
$R$ are known (see \cite{Shaltiel} for a recent review).

\subsubsection{Privacy Amplification}
In the context of privacy amplification, there is additional
information held by Eve.  We denote this using the random variable
$Z$.  Again, Alice wants to compress her string, $X$, using public
randomness\footnote{Her randomness is public because she needs to send
it to Bob in order that he can do the same.}, $R$, to form $S=f(X,R)$,
such that
\begin{equation}
\label{d2}
D(P_{SR|Z=z},P_{U_S}P_R)\leq\epsilon.
\end{equation}

Consider applying the extractor property (see Definition
\ref{extractor}) to the distribution $P_{X|Z=z}$.  This gives that
(\ref{d2}) is satisfied for $H_{\infty}(X|Z=z)\geq\kappa$.  As we
showed in the previous section, the string can be compressed to length
$H_{\infty}(X|Z=z)-2\log\frac{1}{2\epsilon}$.  Alice does not know the
value of this quantity, since she does not know $z$.  However, the
following lemma allows us to derive a useful bound.

\begin{lemma}
\label{cheb}
For a non-negative random variable, $x$, $-\log x>-\log\langle
x\rangle-t$, except with probability less than $2^{-t}$.
\end{lemma}
\begin{proof}
The probability that $-\log x>-\log\langle x\rangle-t$ is the same as
the probability that $x>2^t\langle x\rangle$.  Chebyshev's inequality
then gives the result\footnote{Chebyshev's inequality states that for
a non-negative random variable, $x$, and positive $\alpha$,
$P(x\geq\alpha)\leq\frac{\langle x\rangle}{\alpha}$, and is
straightforward to prove \cite{Mackay}.}.
\end{proof}
As a straightforward corollary to Lemma \ref{cheb}, we have
$H_{\infty}(X|Z=z)\geq H_{\infty}(X|Z)-t$, except with probability
$2^{-t}$, where the conditional min-entropy is defined by
\begin{equation}
H_{\infty}(X|Z)\equiv-\log\sum_{z\in Z}P_Z(z)\max_x P_{X|Z=z}(x)
\end{equation}
Hence, $H_{\infty}(X|Z)+\log\epsilon$ bounds $H_{\infty}(X|Z=z)$,
except with probability $\epsilon$.

In summary, Alice and Bob, by exchanging $R$ publicly, can compress
their shared random string $X$ which is correlated with a string $Z$
held by Eve, to a string $S$ of length roughly equal to
$H_{\infty}(X|Z)$, which is essentially uncorrelated with $Z$ and $R$.

To gain an intuition about privacy amplification, it is helpful to
consider an example.  The set of all functions from $k$ bits to
$\tau<k$ bits forms a $universal_2$ set (albeit an extremely large
one!).  If one picks randomly from amongst this set, then (with high
probability) the chosen function has the following property.  If two
strings are mapped under the chosen function, then the (Hamming)
distance between the resulting strings is independent of that of the
originals.  Thus nearby strings are (with high probability) mapped to
those which are not nearby.  If Eve knows the original string, but
with a few errors, then after it has been mapped, her error rate on
the final string will likely be large.  The probability of successful
amplification is bounded by the probability that Eve can guess the
initial string, since if she guesses correctly, she can discover the
final one with certainty.

\subsubsection{Significance Of Smoothed Entropies}
The bounds we have presented on the length of secure key we can
extract are not tight.  In Section \ref{priv_amp}, we alluded to the
fact that the length of extractable key is tightly bounded by smooth
versions of the R\'enyi entropies (see Equation (\ref{amp_eqn})).  We
briefly explain why this is the case.  Recall the definition in
(\ref{def_1})
\begin{equation*}
H_{\infty}^{\epsilon}(X|Y)\equiv\max_{\Omega}\left(-\log\max_y\max_xP_{X\Omega|Y=y}(x)\right),
\end{equation*}
where $\Omega$ is a set of events with total probability at least
$1-\epsilon$.

The smooth entropy quantity is formed from the sharp version by
cutting away small sections of the probability distribution, and hence
only considering the events $\Omega$.  Since the events cut away occur
with probability at most $\epsilon$, there is only a small affect on
the probability of error.  This may lead to a significant change in
the entropy quantity\footnote{See the discussion on the discontinuous
nature of R\'enyi entropies in Section \ref{beyondshan}.}, and hence a
much larger key can be extracted than that implied by the sharp
entropy quantity.

\subsubsection{Quantum Adversaries} 
\label{q_adv}
Everything we have discussed in this section so far has been with
respect to an eavesdropper holding classical information (the string
$Z$).  More generally, and of particular relevance when discussing
{\sc QKD}, the eavesdropper may attack in such a way that at the end of the
protocol she holds a quantum state that is correlated with Alice's
string.

Alice and Bob's procedure remains unchanged, so their final state at
the end of the protocol (after privacy amplification) is classical,
and corresponds to a string $S$.  Eve, on the other hand, possesses a
quantum system in Hilbert space
$\mathcal{H}_E$.\nomenclature[XHilb]{$\mathcal{H}$}{Hilbert
space} Like in the classical case, security is ensured by constraining
the trace distance.  We demand
\begin{equation}
\label{d3}
D(\rho_{SE},\rho_{U_S}\otimes\rho_E)\leq\epsilon,
\end{equation}
where $\rho_{U_S}$ denotes the maximally mixed state in
$\mathcal{H}_S$.

The trace distance cannot increase under trace-preserving quantum
operations (i.e.\ unitary operations and alterations of system size),
nor, on average, after measurements \cite{Nielsen&Chuang}.  Hence, a
key which satisfies (\ref{d3}) is secure in a {\it composable} manner.
That is, the string $S$ can be treated as random and uncorrelated with
another system in any application, except with probability $\epsilon$.
We need to show how to turn the string $X$, correlated with Eve's
quantum system, into the string $S$ which is virtually uncorrelated.
It turns out that a $universal_2$ hash function is suitable for this
purpose, like in the classical case.

Including the classical spaces used to define the string $X$ and the
random string, $R$, used to choose the hash function, the state of the
system is
\begin{equation}
\rho_{XER}=\sum_{r\in R}\sum_{x\in
  X}\left(P_R(r)P_X(x)\ketbra{x}{x}\otimes\rho_E^x\otimes\ketbra{r}{r}\right).
\end{equation}
Having applied the hash function $f\in F$, the state becomes
\begin{equation}
\rho_{SER}=\sum_{r\in R}\sum_{s\in
  S}\left(P_R(r)P_S(s)\ketbra{s}{s}\otimes\rho_E^s\otimes\ketbra{r}{r}\right),
\end{equation}
where $\rho_E^s=\sum_{x\in f^{-1}(s)}\rho_E^x$.  Ideally, the state of
the system in $\mathcal{H}_S$ would look uniform from Eve's point of
view, even if she was to learn $R$.  The variation from this ideal can
be expressed in terms of the trace distance,
$D(\rho_{SER},\rho_{U_S}\otimes\rho_{ER})$, and is bounded in the
following theorem \cite{Renner}.
\begin{theorem}
\label{Ren1}
If $f$ is chosen amongst a $universal_2$ set of hash functions, $F$,
using random string $R$, and is used to map $X$ to $S$ as described
above, then for $|S|=2^{\tau}$, we have
\begin{equation}
\label{sharp_min_ent}
D(\rho_{SER},\rho_{U_S}\otimes\rho_{ER})\leq\frac{1}{2}2^{-\frac{1}{2}\left(H_{\infty}(\rho_{XE}|E)-\tau\right)}.
\end{equation}
\end{theorem}

Hence, like in the classical case, Alice and Bob can exchange a random
string, $R$, publicly in order to compress their shared random string,
$X$, which is partly correlated with a quantum system held by Eve to a
string $S$ of length roughly equal to $H_{\infty}(\rho_{XE}|E)$, which
is essentially uncorrelated with Eve's system and with $R$.

A similar relation in terms of the smoothed version of the quantum
min-entropy (see Equation (\ref{smq})) provides a better bound on the
key length \cite{Renner}.  Specifically, Equation (\ref{sharp_min_ent})
in Theorem \ref{Ren1} is replaced by
\begin{equation}
\label{sm_min_ent}
D(\rho_{SER},\rho_{U_S}\otimes\rho_{ER})\leq\epsilon+\frac{1}{2}2^{-\frac{1}{2}\left(H_{\infty}^{\epsilon}(\rho_{XE}|E)-\tau\right)}.
\end{equation}

Other extractors more efficient in the length of random string
required are known (see, for example \cite{LRVW}, for ones that
require order $\log{n}$ bits to compress an $n$ bit string.)  However,
these extractors have not been proven to be secure against quantum
attacks, and hence we choose not to use them in this thesis.

\bigskip
Note that privacy amplification using $universal_2$ hash functions has
certain composability properties.  That the final string produced
looks uniform to Eve, means that even if all but one of the bits of
the string are revealed, the final bit remains uniformly distributed
from Eve's perspective.

\section{Types Of Security}
We outline here the various types of security to which we will
refer: 
\begin{enumerate}
\item {\bf Unconditional Security:}\qquad Here the security relies
  only on the laws of physics being correct and applies even against a
  cheater with unlimited computational power (see for example
  \cite{Rudolph,Mayers,Kent_relBC}).  A party can always cause the
  protocol to abort, but can never achieve any useful gain (i.e.\ can
  never discover any private information, or affect the outcome of the
  protocol).
\item {\bf Cheat-evident Security \cite{Kent_INItalk,CK1}:}\qquad The
  protocol is insecure in some way, but any useful cheating attack
  will be detected with certainty.
\item { \bf Cheat-sensitive Security
  \cite{HardyKent,Spekkens&Rudolph_CSWCF,Aharonov&2}:}\qquad The
  protocol has the property whereby any useful cheating attack by one
  party gives that party a non-zero probability of being detected.
\item {\bf Technological Security:}\qquad Also known as {\it
  computational} security in many contexts, although technological
  security subsumes computational security.  Assuming something about
  the technological (computational) power of the adversary, they have
  no useful cheating attack.  However, its security will cease if the
  technological power increases or when a slow algorithm has cracked
  the code.  Users of {\sc RSA}, for instance, are only offered temporary
  security: our best algorithms take several years to factor an
  appropriately-sized product of prime numbers, and if a quantum
  computer can be built, much less.  \comment{ For certain
  implementations of {\sc QKD}, a slightly stronger form of technological
  security is present.  Specifically, the eavesdropper needs to have
  the technological capabilities to render the {\sc QKD} insecure at the
  time of implementation of the protocol.  Once the implementation has
  taken place, the secret is safe forever.  For example, if a quantum
  computer was built in a year's time, a message sent this morning
  using {\sc RSA} would be insecure, while if instead an appropriate quantum
  cryptosystem had been used, the secret would remain safe
  indefinitely.  Of course, information-theoretically secure {\sc QKD}
  schemes exist in theory, and for these a secret would be safe if the
  scheme was implemented in the presence of an eavesdropper with a
  quantum computer.  However, practical {\sc QKD} schemes will, at least
  initially, not offer such a strong promise.}
\end{enumerate}

Let us briefly describe what we mean by {\it useful} cheating.  We do
not demand our protocols prevent any kind of deviation, rather we
require that all deviations are useless, in the sense that they do not
give the deviating party any private information, or allow that party
to influence the outcome of the task.  For instance, in any protocol,
either party can always declare abort at any stage.  We do not
consider this to be a problem, unless at the time of abortion, some
private information has been gleaned.

If we are happy with technological security, then much in the way of
secure multi-party computation has been accomplished.  Kilian has
shown that (at least classically) oblivious transfer (described in
Section \ref{RBC}) can be used to implement any two-party secure
computation \cite{Kilian}.  Since oblivious transfer protocols based
on computational assumptions exist (see for example
\cite{NaorPinkas}), we can generate technologically secure protocols
for two-party secure computations in the classical world.

Unconditional security is the holy grail of the field, and is the
strongest type of security we could hope for, although, as we point
out in Section \ref{thesetting}, there are always additional
assumptions involved.  In many situations, cheat-evident security will
suffice.  This will be the case when being caught cheating accrues a
large penalty.  Consider the case of a bank engaging in a protocol
with one of its clients.  If the client catches the bank cheating, the
resulting media scandal will certainly be detrimental for the bank,
while the bank who catches its client cheating can impose some large
fine.  If the penalties are high enough, cheat-sensitive security
could be sufficient to prevent a party from cheating.

In general, when discussing specific protocols, we will find that they
may have one or more security parameters, $\{N_1,\ldots,N_r\}$.  A
protocol is said to be {\it perfectly secure} if there exist fixed,
finite values of the $\{N_i\}$ for which the security conditions
relevant to the protocol hold. In practice, we will tolerate some
(sufficiently small) probability of failure.  We say that a protocol
is {\it secure} if the security conditions become increasingly close
to holding as the $\{N_i\}$ tend to infinity.  This means that, for
any non-zero failure probability sought, there exist a set of values
for the security parameters for which the protocol achieves this
failure rate.

\section{The Setting} \label{thesetting}
In order to do cryptography, one has to set up a precise world in
which actions take place.  Such a world provides the framework in
which one can make rigorous mathematical statements, and hence prove
results about security.  The actual security we can achieve in
practice depends on how closely the actual environment in which we
perform our protocol resembles our mathematically defined world.
Ideally the two would coincide.  In general though, we will introduce
assumptions in order to create a simple mathematical framework.

The type of worlds that concern us will be distinguished as either
{\it relativistic}, or {\it non-relativistic} (depending on whether we
want to rely on the impossibility of super-luminal signalling for
security), and either {\it quantum} or {\it classical} (depending on
whether the users can create and manipulate quantum systems or not).
Before discussing these distinctions, we give an overview of the set
of assumptions that we will apply within all of our fictitious worlds.

It is impossible to do cryptography without assumptions: the challenge
is to see what can be done assuming as little as possible.  The weaker
the terms of our assumptions, the more powerful the result.  Although
some assumptions are unrealistic in their literal form, they are often
sufficient for realistic purposes.  Take for example the following:
\begin{assumption}
\label{ass1}
Each party has complete trust in the security of their laboratory.
\end{assumption}
By this we mean that nothing within their laboratories can be observed
by outsiders.  Without this assumption, cryptography is pointless, and
yet no laboratory in the world will satisfy such a property.  Can
anyone really be sure that there isn't a microscopic camera flying
around their lab, reporting back data to their adversaries?  As a
matter of principle any laboratory must be coupled to the environment
in some way\footnote{Interactions with a laboratory unable to exchange
information with the outside world would be problematic!}, and this
opens up a channel through which private information could flow.
However, as a practical aside, most parties could set up a laboratory
for which they would be happy that Assumption \ref{ass1} holds to good
enough approximation.  In so doing, they are in essence making
technological assumptions about their adversaries.

No matter what the setting, we will always assume Assumption
\ref{ass1}.  Then, when we talk about, for example, unconditional
security, we implicitly mean unconditional security {\em given our
assumptions} or unconditional security {\em within our model}.  This
caveat does not allow us to turn technologically secure protocols into
unconditionally secure ones by making appropriate assumptions.  A
technologically secure protocol is always insecure from an
information-theoretic point of view.  For example, under the
assumption that factoring is hard, we can say that the {\sc RSA}
cryptosystem is {\it technologically} secure, while without this
assumption, it is insecure.

In the spirit of making the result as powerful as possible, we will
also make the following assumption:
\begin{assumption}
\label{ass2}
Each party trusts nothing outside their own laboratory.
\end{assumption}
In particular, this precludes the possibility of doing cryptography
using a trusted third party, or a source of trusted noise (a situation
in which many cryptographic tasks are known to be achievable
\cite{CrepeauKilian,Damgard&,Damgard&2}).

If a protocol is secure under this assumption, then it is secure {\em
even if our adversaries can control the rest of the universe}.  In
particular, we make no assumption about any other participants in the
protocol.  We allow for the possibility that they may have arbitrarily
good technology, and arbitrarily powerful quantum computers.  In
addition, even if all the other players collude in the most
appropriate way, the protocol must continue to protect any honest
participants.  We need not furnish such protection upon dishonest
parties.

We choose to perform our protocols within perfect environments, so
that all communications are noiseless, all instruments operate
perfectly, and additional parties make no attempt to interfere with
any communications (but the parties with which we are trying to
interact with might).  We sum this up in the following assumption:
\begin{assumption}
\label{ass3}
All communication channels and devices operate noiselessly.
\end{assumption}
It is very convenient to make this assumption in cryptographic
scenarios since it allows all errors that occur during the
implementation of a protocol to be attributed to attacks by another
party.  In the real world, it will be necessary to drop this
assumption, so proliferating the complications of otherwise much
simpler protocols.  This leads to a discussion of {\it reliability}.
We say that a protocol is {\it perfectly reliable} if for some fixed
finite values of the security parameters it has the property that if
both parties are honest, the protocol succeeds without aborting.  In
the presence of noise, for finite values of the security parameters,
there will always be some probability that an honest protocol aborts.
The best we can hope for in such a situation is a {\it reliable}
protocol, where, as the security parameters tend to infinity, the
protocol tends towards perfect reliability.  Given that we assume
Assumption \ref{ass3}, we will always look for perfectly reliable
protocols.

In the future, one might anticipate quantum technology to have become
as widespread as classical technology is today.  Local hardware
retailer might act as a supplies of basic components (unitary gates,
measurement devices etc.).  A cavalier supplier might sell faulty
goods.  A malicious supplier might sell devices that would give him or
her crucial information in a subsequent protocol.  The following
assumption rids us of such considerations
\begin{assumption}
\label{ass4}
Each party has complete knowledge of the operation of the devices they
use to implement a protocol.
\end{assumption}

Assumptions \ref{ass1}--\ref{ass4} will be implicitly assumed in the
protocols discussed in this dissertation, unless otherwise stated.  In
particular, in Chapter \ref{Ch4} we discuss a task where we drop
Assumption \ref{ass4}, and assume instead that all of the devices used
are sourced from a malicious supplier.  Whether a particular set of
assumptions are sufficiently accurate is ultimately a matter for the
protocol's user.

\section{Cryptographic Protocols}
\label{theories}
A protocol is a series of instructions such that if each party follows
the instructions precisely, a certain task is carried out.  The
protocol may permit certain parties to make inputs at various stages,
and may allow them to call on random strings in their possession to
make such inputs.  If the protocol is complete, each party should have
a specified response to cover any eventuality.

Each party in a protocol has a set of systems on which they interact.
Systems on which more than one party can interact form the channel,
which, in the case of more than two parties, may have distinct parts.
In general, the channel system may be intercepted by a malicious party
at any time.  One can always assume that the size of the channel
system is fixed throughout the protocol.  A protocol in which this is
not the case can be transformed into one with this property by first
enlarging the channel system by adding ancillary systems, then
replacing any operations where a system is added to the channel by
swap operations between the system to be added and an ancilla in the
channel.

\subsection{Non-Relativistic Protocols}
\label{non-rel_sec}
Non-relativistic protocols involve the exchange of classical or
quantum information between parties whose locations are completely
unconstrained.  In such protocols, there is a set order in which the
communications occur, and such communications may effectively be
assumed instantaneous.  No constraint is placed on the amount of time
each party has to enact a given step of a protocol, and hence the
surrounding spacetime in which the participants live is irrelevant.

Consider as an illustration a two party protocol between Alice and
Bob.  Suppose that the first communication in the protocol is from
Alice to Bob.  We denote Alice's Hilbert space by $\mathcal{H}_A$,
Bob's by $\mathcal{H}_B$, and the channel's by $\mathcal{H}_C$.  Any
two party protocol then has the following form.

\begin{protocol}\qquad

\begin{enumerate}
\item \label{st1} Alice creates a state of her choice in
$\mathcal{H}_A\otimes\mathcal{H}_C$ and Bob creates a state of his
choice in $\mathcal{H}_B$.  We can assume that these states are pure,
with each party enlarging their Hilbert space if required.
\item Alice sends the channel to Bob.
\item Bob performs a unitary of his choice on
  $\mathcal{H}_C\otimes\mathcal{H}_B$.
\item Bob sends the channel to Alice.
\item Alice performs a unitary of her choice on
  $\mathcal{H}_A\otimes\mathcal{H}_C$.
\item[$\vdots$]
\item[$N.$] At the end of the protocol, both parties measure
  certain parts of their spaces.
\end{enumerate}
\end{protocol}

Note the following.  It is sufficient for Alice and Bob to do
unitaries on all systems in their possession at each step of the
protocol.  All system enlargement can be performed when creating the
initial states in Step \ref{st1}, and all measurements can be kept at
the quantum level until the end of the protocol (see Section
\ref{keep_quantum}).  If we label the unitary operations $U_1$, $U_2$,
$\ldots$, then prior to measurement, the protocol has implemented the
unitary $\left((U_1\otimes\openone_B)(\openone_A\otimes
U_2)\ldots\right)$ on the initial state.  The measurement in Step
$N$ may be used both to check that the protocol took place
correctly, and also to determine a classical output.  This procedure
is illustrated in Figure \ref{fig:protocol}.

\begin{figure}
\begin{center}
\includegraphics[width=\textwidth]{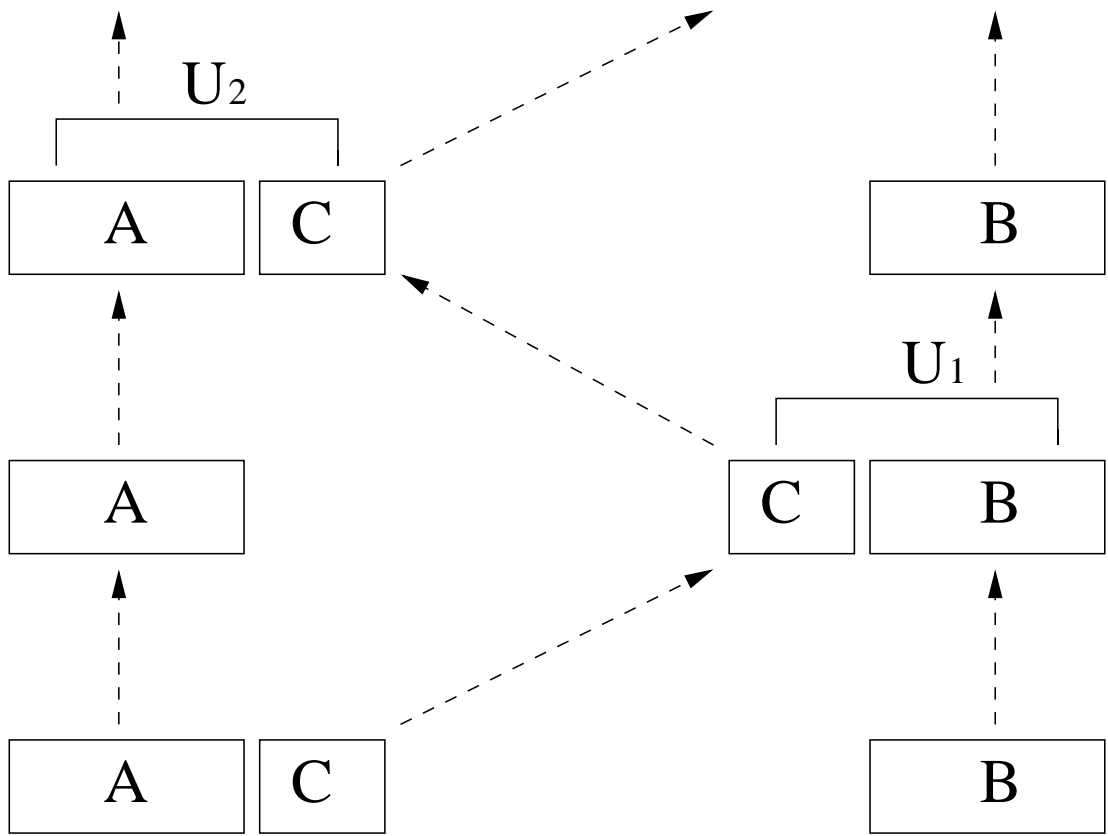}
\caption{Schematic of a non-relativistic protocol between two parties.
  A represents Alice's systems, B represents Bob's systems, and C is
  the channel.  Alice and Bob alternately perform unitaries as the
  protocol proceeds.}
\label{fig:protocol}
\end{center}
\end{figure}

A classical non-relativistic protocol is a special case in which all
states are replaced by classical data, unitary operations are replaced
by classical functions of such data, and we give each party private
randomness\footnote{In a quantum protocol, private randomness comes
for free since either party can create a state for which measurement
in the computational basis yields the desired probability
distribution.}.

\subsection{Relativistic Protocols}
\label{relprots}
In this dissertation, we will assume that relativistic protocols take
place in a Minkowski spacetime.  For practical purposes this is an
over-simplification.  In a more general spacetime the participants
could adopt the protocols we offer providing they are confident in
their knowledge of the structure of the surrounding spacetime and how
it changes during the protocol to sufficient precision.  A secure
protocol could be built along the lines of the ones we present
provided that bounds on the minimum light travel times between sets of
separated sites are known for the duration of the protocol.  Any
protocols carried out on Earth would certainly fit such a criteria.
To avoid a more elaborate discussion, detracting from the important
features of our protocols, we restrict to unalterable Minkowski
spacetimes.  For notational simplicity, we will also restrict our
discussion to the two-party case in the remainder of this section.

We use units in which the speed of light is unity and choose inertial
coordinates, so that the minimum possible time for a light signal to
go from one point in space to another is equal to their spatial
separation.  In a (two-party) relativistic protocol, Alice and Bob are
required to each set up laboratories within an agreed distance,
$\delta$, of two specified locations\footnote{This discussion
generalizes in an obvious way to cover protocols, which require Alice
and Bob to control three or more separate sites.}, $\underline{x}_1$
and $\underline{x}_2$.  Their separation is denoted
$\Delta=|\underline{x}_1-\underline{x}_2|\gg\delta$.
\nomenclature[gD]{$\Delta$}{Separation of sites in a relativistic
protocol} No restrictions are placed on the size and shape of the
laboratories, except that they do not overlap.

We refer to the laboratories in the vicinity of $\underline{x}_i$ as
$A_i$ and $B_i$, for $i = 1$ or $2$. We use the same labels for the
agents (sentient or otherwise) assumed to be occupying these
laboratories.  $A_1$ and $A_2$ operate with complete mutual trust and
have completely prearranged agreements on how to proceed such that we
identify them together simply as Alice; similarly $B_1$ and $B_2$ are
identified as Bob. \nomenclature[aA1A2]{$A_1$, $A_2$, \ldots}{The
agents of Alice in a relativistic protocol}
\nomenclature[aB1B2]{$B_1$, $B_2$, \ldots}{The agents of Bob in a
relativistic protocol} This setup is shown schematically in Figure
\ref{fig:rel_setup}.

\begin{figure}
\begin{center}
\includegraphics[width=\textwidth]{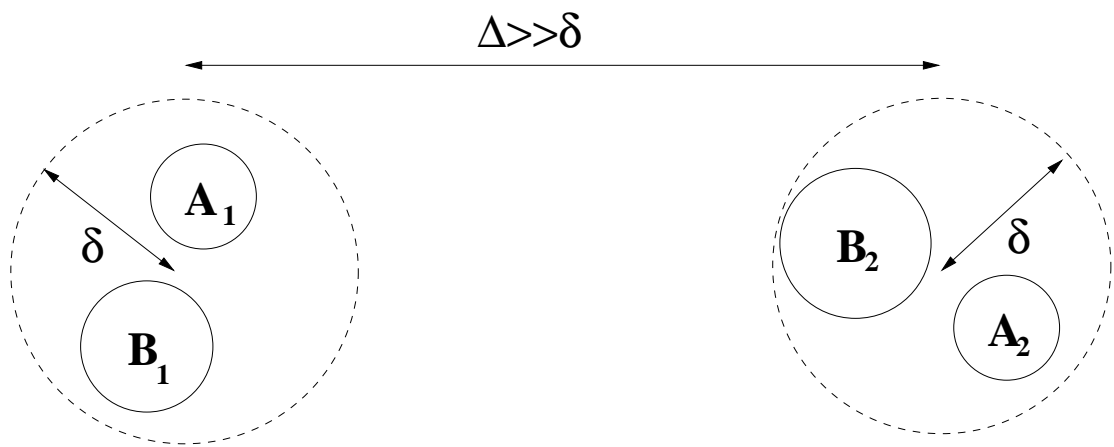}
\caption{Schematic of the setup for a relativistic protocol with two
  separated sites.}
\label{fig:rel_setup}
\end{center}
\end{figure}

To ensure in advance that their clocks are synchronized and that their
communication channels transmit at sufficiently near light speed, the
parties may check that test signals sent out from each of Bob's
laboratories receive a response within time $4 \delta$ from Alice's
neighbouring laboratory, and vice versa.  However, the parties need
not disclose the exact locations of their laboratories, or take it on
trust that the other has set up laboratories in the stipulated regions
(cf.\ Assumption \ref{ass2}). (A protocol which required such trust
would, of course, be fatally flawed.)  Each party can verify that the
other is not significantly deviating from the protocol by checking the
times at which signals from the other party arrive.  These arrival
times, together with the times of their own transmissions, can be used
to guarantee that particular specified pairs of signals, going from
Alice to Bob and from Bob to Alice, were generated independently. This
guarantee is all that is required for security.

We also assume that $A_1$ and $A_2$ either have, or can securely
generate, an indefinite string of random bits. This string is
independently generated and identically distributed, with probability
distribution defined by the protocol, and is denoted ${\bf
x}\equiv\{x_i\}$.  Similarly, $B_1$ and $B_2$ share a random string
${\bf y}\equiv\{y_i\}$. \nomenclature[ax]{{\bf x}}{Random string
shared by the agents of Alice in a relativistic protocol}
\nomenclature[ay]{{\bf y}}{Random string shared by the agents of Bob
in a relativistic protocol} These random strings will be used to make
all random choices as required by the protocol: as $A_1$ and $A_2$,
for instance, both possess the same string, ${\bf x}$, they know the
outcome of any random choices made during the protocol by the other.
We also assume the existence of secure authenticated pairwise channels
between the $A_i$ and between the $B_i$. \footnote{Note that this is
not an unreasonable assumption; these can easily be set up using the
familiar {\sc QKD} schemes, or simply by using the shared random
strings as one-time pads, and in suitable authentication procedures.}
These channels are not necessarily unjammable, but if an honest party
fails to receive the signals as required by a protocol, they abort.
Alternatively, one can think of Alice and Bob as occupying very long
laboratories, as depicted in Figure \ref{fig:rel_setup2}.

\begin{figure}
\begin{center}
\includegraphics[width=\textwidth]{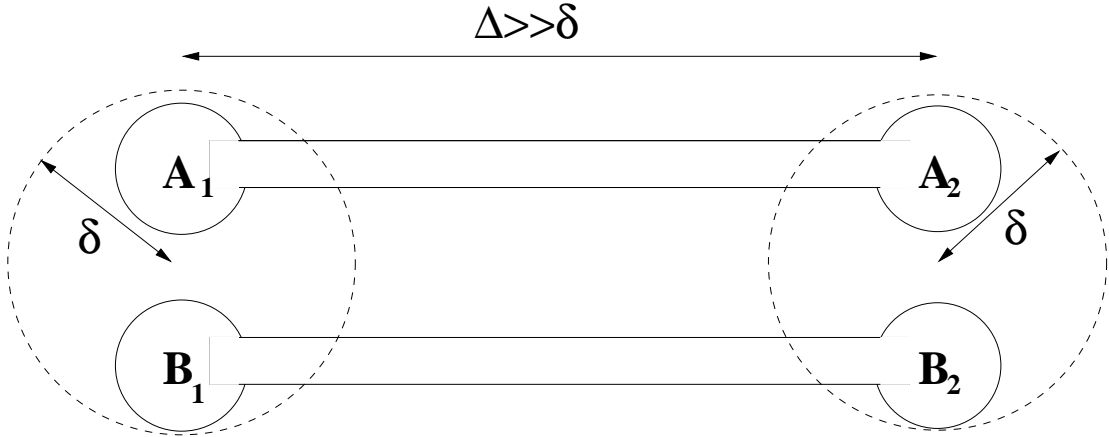}
\caption{Alternative setup for a relativistic protocol with two
  separated sites.}
\label{fig:rel_setup2}
\end{center}
\end{figure}

A relativistic protocol will be defined within this framework by a
prescribed schedule of exchanges of classical or quantum information
between the various agents.  In essence it involves two
non-relativistic protocols, one played out at each of the separated
locations.  These protocols have a limited ability to communicate
between one another.  This generates a constraint on the unitaries
that can be performed at various points in the protocol, since part of
Alice's Hilbert space may be in the secure channel between $A_1$ and
$A_2$, and hence temporarily inaccessible.

In a brief excursion to the real world, we note that the relativistic
setup we have described is not unrealistic.  $A_1$, $A_2$, $B_1$ and
$B_2$ need not be humans performing measurements by hand; rather they
can be machines performing millions of operations per second.  At a
separation of just $3{\rm m}$, one has around $10{\rm ns}$ to do
operations.  This, admittedly, is a little unrealistic for today's
technology, but at $3{\rm km}$, we have roughly $10{\rm \mu s}$ in
which to act.  Using an estimate of $10^8$ gates per second, we can
perform $10^3$ operations in this time.  We certainly do not need
planetary separations for such schemes.  There is also the matter of a
trade-off between large distance and low noise, especially when
considering quantum protocols, but because of Assumption \ref{ass3},
we will not be concerned by this.


\section{Cryptographic Primitives}
\label{RBC}
Three cryptographic primitives will be particularly relevant in this
thesis: Coin tossing, oblivious transfer ({\sc OT}) and bit commitment
({\sc BC}).  We give a brief outline of these tasks here.
\nomenclature[zOT]{{\sc OT}}{Oblivious Transfer}
\nomenclature[zBC]{{\sc BC}}{Bit Commitment}

Coin tossing protocols aim to generate a uniform random bit that is
shared by two parties in such a way that neither party can influence
the bit's value.  They will be discussed in detail in Chapter
\ref{Ch1}.

{\sc OT} comes in several flavours.  In this thesis, we use {\sc OT}
to describe the following functionality.  Alice sends a bit to Bob,
either $1$ or $0$.  Bob either learns Alice's bit, or he learns
nothing, each with probability $\frac{1}{2}$.  Alice does not learn
whether Bob received her bit or not.  It turns out that this task is
sufficient to allow any secure multi-party computation \cite{Kilian}.
Hence, {\sc OT} is in some sense the holy grail of the field.
However, it is known that {\sc OT} is impossible \cite{Rudolph}.  We
give a proof of this in Section \ref{OT_imposs}.


{\sc BC} is another important cryptographic primitive.  A {\sc BC} protocol
involves two steps.  In the first step, one party commits to a bit.
In the second, which occurs at some later time chosen by the
committer, this bit is revealed to a second party.  Before revelation,
the second party is oblivious to the value of the bit, while the first
is unable to alter its value.  One flavour of {\sc BC} can be used to build
a protocol for {\sc OT} \cite{Yao}.  A {\sc BC} of this type is impossible to
construct, even in a relativistic world (the Mayers-Lo-Chau argument
\cite{Mayers,LoChau} for non-relativistic protocols is easy to
extend).  Nevertheless, Kent has shown that a slightly different
flavour of {\sc BC} is possible in a classical, relativistic world
\cite{Kent_relBCshort,Kent_relBC}.  He further conjectures that this
protocol remains secure in a quantum world, against the most general
quantum attack, but presently this is unproven.

We will not go into the range of subtleties surrounding the various
types of {\sc BC} (the interested reader should refer to
\cite{Kent_relBC} for a longer discussion).  Here we simply point out
that Kent's {\sc BC} schemes require sustained communications in
order to maintain the commitment, and that they have the property of
{\em retractability}, that is the party making the commitment can get
their committed state returned if they later decide not to follow
through with the unveiling.  This latter feature is what scuppers the
use of relativistic bit commitment
({\sc RBC})\nomenclature[zRBC]{{\sc RBC}}{Relativistic Bit Commitment} schemes for
building Yao's {\sc OT} scheme \cite{Yao,Kent_relBC}.

We will use {\sc RBC} as a subprotocol in some of the schemes we later
discuss, so, for completeness, we outline a protocol for its
implementation here.  We choose the simplest of Kent's schemes, {\sc RBC}1,
in the case where Alice commits a bit to Bob:

\begin{protocol} ({\sc RBC}1)
\begin{enumerate}
\item $B_1$ sends to $A_1$ an ordered pair, $(n_{1,0},n_{1,1})$ of
  random non-equal integers.  These, along with all other integers
  used in the protocol, will be in the set $\{0,\ldots,N-1\}$, and all
  arithmetic performed is modulo $N=2^p$, for integer $p$.
\item To commit to bit $b$, $A_1$ returns $n_{1,b}+m_1$ to $B_1$.
\item To sustain the commitment, $A_2$ commits the binary form
  $a_{p-1}^1,a_{p-2}^1,\ldots,a_0^1$ of $m_1$ to $B_1$, by having
  $B_1$ send the random integer pairs
  $(n_{2,0},n_{2,1})$, $(n_{3,0},n_{3,1}),$ $\ldots,$ $(n_{p+1,0},n_{p+1,1}),$
  and returning the set
  $n_{2,a_{p-1}^1}+m_2,$ $n_{3,a_{p-2}^1}+m_3,$ $\ldots,$
  $n_{p+1,a_0^1}+m_{p+1}.$
\item This procedure then iterates, with $A_1$ committing the binary
  form of $m_2,\ldots,$ $m_{p+1}$ to $B_1$ in an analogous way.
\end{enumerate}

At some later time, Alice can unveil on either or both sides.  For
$A_1$ to unveil, she sends to $B_1$ the list of random numbers,
$\{m_i\}$, used by $A_2$ in her last set of commitments. ($A_1$ knows
these because they were generated using the shared random sting ${\bf
x}$.)  $B_1$ receives this list at such time that he can ensure they
were sent in a causally disconnected manner to the receipt of the
random pairs $\{(n_{i,0},n_{i,1})\}$ by $A_2$.  $B_1$ and $B_2$ can
then share all their data, and verify that it did correspond to a
valid commitment of either $0$ or $1$.
\end{protocol}

This protocol has the undesirable feature that it requires an
exponentially increasing rate of communication.  However, Kent has
also introduced a second protocol, {\sc RBC}2 which combines {\sc
RBC}1 with a scheme due to Rudich, in order to achieve {\sc RBC} with
a constant transmission rate.  The full details of this scheme can be
found in \cite{Kent_relBC}, and are not presented in this thesis.




\chapter{The Power Of The Theory -- Strong Coin Tossing}
\label{Ch1}

\begin{quote}
{\it ``A theory is acceptable to us only if it is beautiful''} --
Albert Einstein
\end{quote}

\section{Introduction}
Landauer's often quoted doctrine, {\it ``information is physical''}
succinctly expresses the fact that what can and cannot be done in
terms of information processing is fundamentally dictated by physics.
Information processing is performed by physical machines (abacuses,
computers, human beings, etc.), and the power of these limits the
information processing power.  In light of the above, it is not of
great surprise that new physical theories lead to changes in
information processing power.  Historically, though, more than 50
years elapsed between the development of quantum theory and the
realization that it offers an increase in information processing
power.  This delay can surely be attributed, at least in part, to the
failure of both physicists and information theorists to recognize the
physical nature of information.

In this chapter, we illustrate the r\^ole of the physical theory in
information processing power.  We consider theories that are either
quantum or classical, and are either relativistic or not.  The
relevance of the different theories for the construction of protocols
has been described in Section \ref{theories}.  Here we give specific
examples.  As a focus for our discussion we use one of the simplest
cryptographic tasks: strong coin tossing.  A classical
non-relativistic theory cannot realize this task to any extent.
Introducing quantum mechanics allows protocols with partial security,
while relativistic protocols can realize the task perfectly.

Informally, a coin tossing protocol seeks to allow two separated
parties to exchange information in such a way that they generate a
shared random bit.  The bit is random in the sense that (ideally) $0$
or $1$ occur with probability $\frac{1}{2}$ each, and neither party
can increase the probability of either outcome by any method.  In many
physical models, this ideal cannot be achieved.  In such cases one
weakens the requirements of the protocol.  It is demanded that if both
participants are honest, the outcome is $0$ or $1$ with probability
$\frac{1}{2}$ each.  A protocol is then given a figure of merit in
terms of the maximum cheating probability a dishonest party can
achieve against an honest party.  The quantity often used is the bias,
the deviation of the maximum cheating probability from $\frac{1}{2}$.
A strong coin tossing protocol seeks to protect an honest party from a
dishonest party whose direction of bias is unknown, while a weak coin
toss seeks to protect an honest party only against the dishonest party
biasing towards one particular outcome.  Commonly, coin tosses are of
the latter form (e.g.\ Alice and Bob, having recently divorced, want to
decide who keeps the car).  Strong coin tosses are relevant in
situations where there is knowledge asymmetry between the parties, so
that it is not clear to one which way the other wishes to bias (e.g.\
Alice knows whether the car works but Bob does not).

In the next sections, we give formal definitions of the relevant coin
tossing tasks before discussing how well they can be achieved in the
various physical models of interest.  Our contribution in this area is
Protocol \ref{colbeck_protocol}, for which no protocols are known with
a better bias.

\section{Definitions}
In a coin tossing protocol, two separated and mistrustful parties,
Alice and Bob, wish to generate a shared random bit.  We consider a
model in which they do not initially share any resources, but have
access to trusted laboratories containing trusted error-free apparatus
for creating and manipulating quantum states (cf.\ Assumptions
\ref{ass1}--\ref{ass4}). In general, a protocol for this task may be
defined to include one or more security parameters, which we denote
$N_1,\ldots,N_r$.

If both parties are honest, a coin tossing protocol guarantees that
they are returned the same outcome, $b\in\{0,1\}$ where outcome $b$
occurs with probability $\frac{1}{2}+\zeta_b(N_1,\ldots,N_r)$, or
``abort'' which occurs with probability $\zeta_2(N_1,\ldots,N_r)$, and
for each $j\in\{0,1,2\}$, $\zeta_j(N_1,\ldots,N_r)\rightarrow 0$ as
the $N_i\rightarrow\infty$.  The {\it bias} of the protocol towards
party $P\in\{A,B\}$ is denoted
$\epsilon_P=\max\left(\epsilon_P^0,\epsilon_P^1\right)$, where $P$ can
deviate from the protocol in such a way as to convince the other
(honest) party that the outcome is $b$ with probability at most
$\frac{1}{2}+\epsilon_{P}^{b}+\delta_P^b(N_1,\ldots,N_r)$, and the
$\delta_P^b(N_1,\ldots,N_r)\rightarrow 0$ as the
$N_i\rightarrow\infty$.  We make no requirements of the protocol in
the case where both parties cheat.

The {\it bias} of the protocol is defined to be
$\max(\epsilon_A,\epsilon_B)$.  A protocol is said to be {\it
balanced} if $\epsilon_A^{b}=\epsilon_B^{b}$, for $b=0$ and $b=1$.

We define the following types of coin tossing:
\begin{definition} 
{\bf (Ideal Coin Tossing)} A coin tossing protocol is ideal if it has
$\epsilon_A=\epsilon_B=0$, that is, no matter what one party does to
try to bias the outcome, their probability of successfully doing so is
strictly zero.  It is then said to be {\it perfectly secure} if for
some finite values of $N_1,\ldots,N_r$, the quantities
$\zeta_j(N_1,\ldots,N_r)$ and $\delta_P^b(N_1,\ldots,N_r)$ are
strictly zero, and otherwise is said to be {\it secure}.
\end{definition}

\begin{definition}
{\bf (Strong Coin Tossing)} A strong coin tossing
protocol is parameterized by a bias, $\gamma$.  The protocol has the
property that $\epsilon_P^b\leq\gamma$ for all $P\in\{A,B\}$ and
$b\in\{0,1\}$, with equality for at least one combination of $P$ and $b$.
\end{definition}

\begin{definition}
{\bf (Weak Coin Tossing)} A weak coin tossing protocol
is also parameterized by a bias, $\gamma$.  It has the property that
$\epsilon_A^0\leq\gamma$ and $\epsilon_B^1\leq\gamma$, with equality
in at least one of the two inequalities.
\end{definition}

\section{Where Lies The Cryptographic Power?}
\label{power}
Cryptography involves secrets.  One generally begins in a situation in
which each party holds private data, and ends in a situation in which
each party gains a specified and often highly restricted piece of
information on the inputs of the others.  Let us specialize to the
case of two party protocols.  Kilian \cite{Kilian} sums up the
difficulty of classical protocols for these on the grounds that at any
point in such a protocol, one party knows exactly what information
is available to the other, and vice-versa.  If this knowledge symmetry
can be broken (for example by assuming the existence of a black-box
performing {\sc OT}, or by using a trusted noisy channel
\cite{CrepeauKilian,Damgard&,Damgard&2}) then any secure multi-party
computation can be performed \cite{Kilian}.

Quantum mechanics also provides a way of generating knowledge
asymmetry.  For example, consider a protocol which involves Alice
choosing one of two non-orthogonal bases at random to encode each
bit.  She sends the quantum states which store the encodings to
Bob.  Bob, being unaware of Alice's bases, cannot reconstruct her bits
with certainty.  Likewise, if Bob measures each state he receives in
one of the two encoding bases chosen at random, then Alice cannot tell
exactly what Bob knows about her string.  Therefore, information
completeness is lost, and extra cryptographic power exists over
protocols involving only classical systems.  \comment{Suppose Alice
sends bits to Bob by encoding them in one of two non-orthogonal bases.
If Bob guesses the basis and measures, he will generate a string that
is partially correlated with Alice's.  If both parties follow this
procedure honestly, then Alice is unaware which of her bits Bob
received correctly, and likewise, Bob is unaware of a subset of
Alice's bits.}

The procedure described above acts like a noisy channel, but there is
a key cryptographic difference between the two.  The noise generated
by a noisy channel comes from an outside system, while that generated
by sending quantum states is inherent to the physics of the system.
From a cryptographic point of view, the former is equivalent to
assuming the existence of a trusted third party.  If either party
could tap into the system generating the noise, then security would be
compromised.  This is a by-product of the fundamental reversibility of
classical processes---if the process causing the noise was reversed,
the information would be recovered.  This is not the case for a
quantum mechanical measurement.  The process by which it is generated
is fundamentally irreversible, and hence such a security issue does
not arise\footnote{It is the possibility of delaying measurement that
prevents such a quantum system being used to build {\sc OT} as the
standard classical reductions \cite{CrepeauKilian,Damgard&,Damgard&2}
imply.  However, provided at least one party behaves honestly,
information completeness is lost.}.  This is not the only source of
cryptographic power generated by quantum theory.  Another comes from
the so-called monogamy of entanglement, which provides security in
Ekert's variant of the {\sc BB}84 protocol \cite{Ekert}, and also in
the protocols we discuss in Chapter \ref{Ch4}.

Relativistic protocols allow short-lived perfectly binding and
perfectly concealing commitments\footnote{From which longer-lived ones
can be constructed, as discussed previously.}.  When $A_1$ sends
classical information to $B_1$, she is perfectly committed to it
(since $B_1$ knows it).  However, from $B_2$'s point of view, this
commitment is perfectly concealing for the light travel time.
Anything $B_2$ sends to $A_2$ within this time he does in ignorance of
$A_1$'s message.  This is where the power lies in relativistic
cryptography.

\section{Coin Tossing}
\subsection{Classical Non-Relativistic Protocols}
Coin tossing is a two person game.  The ``moves'' of the game are the
communications of the parties.  In the classical and non-relativistic
case, coin tossing can be studied using well-established techniques of
game theory.  It can be phrased as a zero-sum game, meaning that the
payoffs for any outcome sum to zero. (We can assign $+1$ for a win,
$-1$ for a loss, and $0$ to abort for each party.  Thus if Alice wins,
she gets $+1$, while Bob gets $-1$, these having zero sum.  The exact
payoffs may not be precisely these, but this should not affect the
security of the computation.)

We present here a (sketch) proof of the impossibility of classical
coin tossing based on a result of game theory.  The result we need
refers to complete information games, which are those for which each
party knows all previous moves of all other parties prior to making
theirs.  The result states that all (finite) zero-sum
complete-information 2-person games are strictly determined, i.e., one
party following their optimal strategy can win against any strategy of
the other.  


Consider games in which there are no random moves.  After the last
move has been made, the game has a defined payout.  Let us suppose
that a positive payout favours Alice, and a negative one favours
Bob\footnote{Since the game is zero-sum, the payout to Alice is always
opposite that of Bob.}.  Suppose Bob makes the last move.  He will
choose his move so as to minimize the payout.  Assuming no degeneracy,
the last move is determined by this.  (If there is degeneracy, then
Bob can choose freely from amongst the degenerate moves.
Alternatively, one could construct a new game in which the degenerate
moves are combined.)  Since this move is determined, we can define a
game with one fewer moves in which the payouts are defined by what
results if Bob follows his optimal strategy.  This shorter game has
Alice making the final move, which she does so as to maximize the
payout.  Thus, if we assume Alice and Bob always make their best play
at every opportunity, this process iterates so that the entire game is
completely determined.  That is, one player always has a winning
strategy against any strategy of the other.  Since the winning
strategy works against any strategy of the other, it also works if the
other makes random choices at certain points in the protocol.  The
above argument is formally proven in Chapter 15 of
\cite{vonNeumann&Morg}.

A classical non-relativistic coin tossing protocol is such a game, and
hence one party can always win with certainty, i.e., the best
achievable bias is $\frac{1}{2}$.

Note that both non-relativistic quantum protocols and relativistic
protocols do not fit into this model.  In a quantum protocol, if one
party is allowed to choose their measurement basis, the other does not
know what information they received.  In a relativistic protocol,
timing constraints can be used to ensure that one party must make a
move without knowledge of those of the other party.  It is therefore
possible to construct protocols which are not information complete,
and hence the above argument does not go through.  We will demonstrate
this below by giving coin tossing protocols whose bias is less that
$\frac{1}{2}$.

\subsection{Quantum Non-Relativistic Protocols}
Such protocols, commonly abbreviated as quantum protocols, have been
widely studied in the literature.  That quantum coin tossing protocols
offer some advantage over classical ones was realized by Aharonov et
al.\ \cite{Aharonov&2}, who introduced a protocol achieving a bias of
$\frac{1}{2\sqrt{2}}$ \cite{Aharonov&2,Spekkens&Rudolph2}.  For strong
coin tossing, it has been shown by Kitaev that in any protocol, at
least one party can achieve a bias greater than
$\frac{1}{\sqrt{2}}-\frac{1}{2}$ \cite{Kitaev}. It is not known
whether this figure represents an achievable bias.  The best known
bias to date is $\frac{1}{4}$ \cite{Ambainis,Colbeck}.  This bias
is optimal for a large set of bit-commitment based protocols
\cite{Spekkens&Rudolph}.  For weak coin tossing, Kitaev's bound is
known not to apply and lower biases than
$\frac{1}{\sqrt{2}}-\frac{1}{2}$ have been achieved (see for example
\cite{Mochon2} for the best bias to date).  Moreover, Ambainis has
shown that a weak coin toss protocol with bias $\epsilon>0$ must have
a number of rounds that grows as $\Omega(\log\log\frac{1}{\epsilon})$
\cite{Ambainis}.

We present now the two protocols which achieve strong coin tossing
with bias $\frac{1}{4}$.  The first, due to Ambainis, is based on bit
commitment.\footnote{A bit commitment based coin tossing scheme has
one party commit a bit, after which the other announces another bit.
If the {\sc XOR} of the two bits is 0, the outcome is heads, if 1 it
is tails.}  The second protocol is our contribution.  It works by
trying to securely share entanglement before exploiting the quantum
correlations that result.

\comment{Our strong coin tossing protocols realise the following ideal:
\begin{if}(Strong coin tossing with bias $\frac{1}{4}$)

\begin{enumerate}
\item Each party can choose either to be honest, or to cheat.  If
  Alice chooses to cheat, she can decide the probability distribution
  $(p_0,p_1,p_{\text{abort}})$ as she pleases with the constraint that
  $p_0\leq\frac{3}{4}$ and $p_1\leq\frac{3}{4}$.
\item If both choose not to cheat, the ideal outputs a uniformly
  distributed random bit.
\item If Alice cheats, the ideal outputs $0$ with probability
  $\frac{3}{4}$, or else outputs ``abort''.
\item If Bob cheats, the ideal outputs $1$ with probability
  $\frac{3}{4}$, or $0$ with probability $\frac{1}{4}$.
}

We give a brief description of Ambainis' protocol below.  More
details, including the proof that it has a bias of $\frac{1}{4}$ can
be found in \cite{Ambainis}.
\begin{protocol}\qquad
\label{protocol1}

\noindent We define the states
\begin{eqnarray}
\ket{\phi_{b,x}}=\left\{\begin{array}{c}
\frac{1}{\sqrt{2}}(\ket{0}+\ket{1})\qquad b=0,\, x=0\\
\frac{1}{\sqrt{2}}(\ket{0}-\ket{1})\qquad b=0,\, x=1\\
\frac{1}{\sqrt{2}}(\ket{0}+\ket{2})\qquad b=1,\, x=0\\
\frac{1}{\sqrt{2}}(\ket{0}-\ket{2})\qquad b=1,\, x=1
\end{array}\right. .
\end{eqnarray}
The protocol then proceeds as follows:
\begin{enumerate}
\item \label{amb1} Alice picks two random bits $b\in\{0,1\}$ and
  $x\in\{0,1\}$, using a uniform distribution.  She creates the
  corresponding qutrit state $\ket{\phi_{b,x}}$ and sends it to Bob.
\item Bob picks a random bit, $b'\in\{0,1\}$ from a uniform
  distribution, and sends $b'$ to Alice.
\item Alice sends $b$ and $x$ to Bob, who then checks that the state
  he received in Step \ref{amb1} matches (by measuring it with respect
  to a basis consisting of $\ket{\phi_{b,x}}$ and two states
  orthogonal to it).  If the outcome of the measurement is not the one
  corresponding to $\ket{\phi_{b,x}}$, Bob aborts.
\item Otherwise, the result of the coin flip is $b\oplus b'$.
\end{enumerate}
\end{protocol}

This protocol is based on bit commitment.  Alice (imperfectly) commits
a bit, $b$, to Bob by encoding it using one of two non-orthogonal
pairs of states.  Bob then sends a bit $b'$ to Alice.  The outcome is
decided by the {\sc XOR} of $b$ and $b'$.  Many of the coin tossing
schemes considered in the literature are of this type.  The security
of such protocols is only as strong as the bit commitment on which
they are based.  Bounds on the possible biases achievable in bit
commitment schemes are well known \cite{Spekkens&Rudolph}.  However,
coin tossing is strictly weaker than bit commitment \cite{Kent_CTBC},
hence bounds on the achievability of bit commitment do not imply
similar ones for coin tossing.  It is therefore of interest to search
for schemes that do not rely on bit commitment.  We describe one such
protocol and give its complete security analysis below.  This protocol
has been published by us \cite{Colbeck}.

\begin{protocol}\qquad
\label{colbeck_protocol}
\begin{enumerate}
\item Alice creates $2$ copies of the state
   $\ket{\psi}=\frac{1}{\sqrt{2}}(\ket{00}+\ket{11})$ and sends the
   second qubit of each to Bob.
\item \label{Bch} Bob randomly selects one of the states to be used
  for the coin toss.  He informs Alice of his choice.
\item Alice and Bob measure their halves of the chosen state in the
  $\{\ket{0},\ket{1}\}$ basis to generate the result of the coin toss.
\item Alice sends her half of the other state to Bob who tests whether
  it is the state it should be by measuring the projection onto
  $\ket{\psi}$.  If his test fails, Bob aborts.
\end{enumerate}
\end{protocol}

\subsubsection{Alice's Bias}
Assume Bob is honest.  We will determine the maximum probability,
$p_A$, that Alice can achieve outcome 0 (an analogous result follows by
symmetry for the case that Alice wants to bias towards 1).  Alice's
most general strategy is as follows. She can create a state in an
arbitrarily large Hilbert space,
$\ket{\Psi}\in\mathcal{H}_A\otimes\mathcal{H}_{A_1}\otimes\mathcal{H}_{B_1}\otimes\mathcal{H}_{A_2}\otimes\mathcal{H}_{B_2}$,
where $\mathcal{H}_A$ represents the space of an ancillary system
Alice keeps, $\mathcal{H}_{B_1}$ and $\mathcal{H}_{B_2}$ are qubit
spaces sent to Bob in the first step of the protocol, and
$\mathcal{H}_{A_1}$ and $\mathcal{H}_{A_2}$ are qubit spaces, one of
which will be sent to Bob for verification.  On receiving Bob's choice
of state in Step \ref{Bch}, Alice can do one of two local operations
on the states in her possession, before sending Bob the relevant qubit
for verification.  Alice should choose her state and local operations
so as to maximize the probability that Bob obtains outcome $0$ and
does not detect her cheating.

Let us denote the state of the entire system by
\begin{equation}
\ket{\Psi}=\sum_{i=0}^1\sum_{j=0}^1a_{ij}\ket{\phi_{ij}}_{AA_1A_2}\ket{ij}_{B_1B_2}
\end{equation}
where $\{\ket{\phi_{ij}}_{AA_1A_2}\}_{i,j}$ are normalized states in
Alice's possession, and $\{a_{ij}\}_{i,j}$ are coefficients.  Suppose
Bob announces that he will use the first state for the coin toss.
There is nothing Alice can subsequently do to affect the probability
of Bob measuring 0 on the qubit in $\mathcal{H}_{B_1}$.  We can assume
that Bob makes the measurement on this qubit immediately on making his
choice.  Let us also assume that Alice discovers the outcome of this
measurement so that she knows the pure state of the entire system (we
could add a step in the protocol where Bob tells her, for
example\footnote{Such a step can only make it easier for Alice to
cheat, so security under this weakened protocol implies security under
the original one.}).  If Bob gets outcome $1$, then Alice cannot win.
On the other hand, if Bob gets outcome $0$, the state of the remaining
system becomes
\begin{equation}
\frac{a_{00}}{\sqrt{|a_{00}|^2+|a_{01}|^2}}\ket{\phi_{00}}_{AA_1A_2}\ket{0}_{B_2}+\frac{a_{01}}{\sqrt{|a_{00}|^2+|a_{01}|^2}}\ket{\phi_{01}}_{AA_1A_2}\ket{1}_{B_2},
\end{equation}
and Alice can win if she can pass Bob's test in the final step of the
protocol.  Since entanglement cannot be increased by local operations,
the system Alice sends to Bob in this case can be no more entangled
than this state.  Since measurements (on average) reduce entanglement,
Alice's best operation is a unitary on her systems.  Such an operation
is equivalent to a redefinition of $\{a_{ij}\}$ and
$\{\ket{\phi_{ij}}\}$, which Alice is free to choose at the start of
the protocol anyway.  Alice can do no better than by choosing the
coefficients, $\{a_{ij}\}$, to be real and positive.  The state which
best maximizes the overlap of the system in the $A_2B_2$ subspace with
$\ket{\psi}$ is then,
\begin{equation}
\frac{a_{00}}{\sqrt{a_{00}^2+a_{01}^2}}\ket{00}_{A_2B_2}+\frac{a_{01}}{\sqrt{a_{00}^2+a_{01}^2}}\ket{11}_{A_2B_2}.
\end{equation}
Alice therefore cannot fool Bob into thinking she was honest with
probability greater than
$\frac{(a_{00}+a_{01})^2}{2(a_{00}^2+a_{01}^2)}$.  Using a similar
argument for the case that Bob chooses the second state for the coin
toss shows that Alice's overall success probability is at most
$\frac{1}{4}\left(2a_{00}^2+2a_{00}a_{01}+2a_{00}a_{10}+a_{01}^2+a_{10}^2\right)$.
Maximizing this subject to the normalization condition gives a maximum
of $\frac{3}{4}$, hence we have the bound
$p_A\leq\frac{3}{4}$. Equality is achievable within the original
protocol (i.e., without the additional step we introduced) by having
Alice use the state
\begin{eqnarray}
\sqrt{\frac{2}{3}}\ket{0000}_{A_1B_1A_2B_2}+\frac{1}{\sqrt{6}}\left(\ket{0011}_{A_1B_1A_2B_2}+\ket{1100}_{A_1B_1A_2B_2}\right),
\end{eqnarray}
and simply sending $\mathcal{H}_{A_1}$ or $\mathcal{H}_{A_2}$ to Bob
in the final step, depending on Bob's choice.  

The protocol is cheat-sensitive towards Alice---{\em any} strategy
which increases her probability of obtaining one outcome gives her a
non-zero probability of being detected.

\subsubsection{Bob's Bias}
Assume Alice is honest.  We will determine the maximum probability, $p_B$,
that Bob can achieve the outcome 0.  The maximum probability
for outcome 1 follows by symmetry.  Bob seeks to take the qubits he
receives, perform some local operation on them, and then announce one
of them to be the coin-toss state such that the probability that Alice
measures 0 on her part of the state he announces is maximized.

Suppose that we have found the local operation maximizing Bob's
probability of convincing Alice that the outcome is 0.  Having
performed this operation and sent the announcement to Alice, the
outcome probabilities for Alice's subsequent measurement on the state
selected by Bob in the $\{\ket{0},\ket{1}\}$ basis are fixed.  Bob's
probability of winning depends only on this.  It is therefore
unaffected by anything Alice does to the other qubit, and, in
particular, is unaffected if Alice measures both of her qubits in the
$\{\ket{0},\ket{1}\}$ basis before looking at Bob's choice.  Such a
measurement commutes with Bob's local operation, so could be done by
Alice prior to Bob's operation without changing any outcome
probabilities.  If Alice does this measurement she gets outcome 1 on
both qubits with probability $\frac{1}{4}$.  In such a case, Bob
cannot convince Alice that the outcome is 0.  Therefore, we have
bounded Bob's maximum probability of winning via $p_B\leq\frac{3}{4}$.

To achieve equality, Bob can measure each qubit he receives in the
$\{\ket{0},\ket{1}\}$ basis, and if he gets one with outcome 0, choose
this state as the one to use for the coin toss.  There is no cheat
sensitivity towards Bob; he can use this strategy without fear of
being caught.

\subsubsection{Discussion}
In this section we have presented two non-relativistic quantum
protocols for strong coin tossing.  Each of which has bias
$\frac{1}{4}$.  The first, due to Ambainis, is based on bit
commitment.  The second is based on sharing entanglement.  In terms of
practicality, the key differences between the schemes are as follows.
Firstly, Ambainis' protocol requires manipulation and communication of
a single qutrit, while ours requires four qubits (two of which are
communicated).  Furthermore, there cannot be a bit-commitment based
scheme of this type\footnote{i.e.\ where all of the quantum systems
are supplied by Alice.} with a smaller dimensionality than Ambainis'
since bit-commitment based protocols using qubits cannot achieve bias
$\frac{1}{4}$ \cite{Spekkens&Rudolph}.  Secondly, Ambainis' protocol
does not require the storage of quantum systems.

The question of whether Kitaev's bound can be reached remains open.
That two protocols attempting to optimize the bias both have bias
$\frac{1}{4}$ is evidence that this might be the best possible.  One
would like to construct a proof of this.

\subsection{Relativistic Protocols}
\label{relCT}
Such protocols allow coin tossing with zero bias, due to the bit
commitment property they offer (cf.\ Section \ref{power}).

\begin{protocol}\qquad

\begin{enumerate}
\item At time $t_0$, $A_1$ sends a bit, $b\in\{0,1\}$, to $B_1$
  choosing $b$ from a uniform distribution.
\item $B_2$ simultaneously sends a bit, $b'$, to $A_2$.
\item $B_1$ checks that his received message arrived before time
  $t_0+D$, and likewise, so does $A_2$\footnote{Recall that we use
  units in which the speed of light is unity.}.  If this is not the
  case, they abort.
\item The disconnected agents of Alice communicate with one another,
  as do those of Bob.  Alice and Bob can then compute the coin toss
  outcome, $b\oplus b'$.
\end{enumerate}
\label{relCT_prot}
\end{protocol}

The impossibility of superluminal signalling prevents either party
cheating in such a protocol.

\section{Discussion}
In this chapter, we have shown how the physical world in which our
protocol operates has significant implications on its security, thus
highlighting the fact that what can and cannot be done in terms of
information processing tasks depends fundamentally on physics.  In a
non-relativistic, classical world, it is impossible to achieve
unconditional security for any two-party protocol, because such
protocols are information complete.  In a non-relativistic quantum
world, information completeness can be broken, as described in Section
\ref{power}.
\comment{For example, a protocol might involve Alice
encoding bits in one of two non-orthogonal bases, before sending the
quantum states to Bob. Bob, being unaware of Alice's bases, cannot
reconstruct her bits with certainty.  Likewise, if Bob measures each
state he receives in one of the two encoding bases chosen at random,
then Alice cannot tell exactly what Bob knows about her string.
Therefore, information completeness is lost, and extra cryptographic
power exists over protocols involving only classical systems.}  
This is sufficient to ensure partial security in coin tossing.
Relativity introduces the possibility of stronger security still.  The
impossibility of superluminal signalling means that information can be
completely concealed from one party, at least for the light travel
time.  This allows a zero-knowledge, finite-time commitment, which is
sufficient for coin tossing.

In the forthcoming chapters we discuss the extent to which quantum and
relativistic protocols can be used to achieve other cryptographic
tasks.  Chapter \ref{Ch2} will show that one additional task (variable
bias coin tossing) is possible, while in Chapter \ref{Ch3} a large set
of other tasks are shown to be impossible.



\chapter{Variable Bias Coin Tossing}
\label{Ch2}

\begin{quote}
{\it ``What does chance ever do for us?''} -- William Paley
\end{quote}

\section{Introduction}
In a future version of society, etiquette has become so important that
it is impinging on free will.  Declining an invitation from an
upstanding member of the community has become near impossible.  A new
social code has emerged to circumvent this, whereby the acceptance or
otherwise of all invitations are resolved via a variable bias coin
toss ({\sc VBCT}).  This task allows one party to secretly choose the bias
of the coin within some prescribed range.  If one wants to decline the
invitation, one biases so as to maximize the probability of
declination.  Then, on receiving the (hopefully) negative outcome, one
simply ascribes this to ill fortune.  This new social code therefore
restores some free will, at the expense that sometimes one has to
decline favourable invitations.
\bigskip
\nomenclature[zVBCT]{{\sc VBCT}}{Variable Bias Coin Tossing}

In this chapter, we consider protocols for the task of variable bias
coin tossing between two parties.  The results presented here have
been published by us in \cite{CK1}.  The aim of a {\sc VBCT} protocol is to
generate a shared random bit, as though by a biased coin whose bias is
secretly chosen by one of the parties to take some value within a
prescribed range.  This is the two-faced case of the more general task
of carrying out a variable bias $n$-faced die roll, in which one of
$n$ possible outcomes is randomly generated as though by a biased die,
whose bias (i.e. list of outcome probabilities) is secretly chosen by
one of the parties to take some value within a prescribed convex set.
Variable bias coin tossing and die rolling are themselves special
cases of secure $2$-party computations.  To understand their
significance, we first locate them within a general classification of
secure computation tasks.

\section{Secure Multi-Party Computation}
\label{sec_comp_intro}
A general secure classical computation involves $N$ parties, labelled
by $i$ in the range $1 \leq i \leq N$, who each have some input,
$x_i$, and wish to compute some (possibly non-deterministic) functions
of their inputs, with the $i$th party receiving as output $f_i(x_1 ,
\ldots , x_N )$.  We call this a classical computation because the
inputs and outputs are classical, although we allow such computations
to be implemented by protocols which involve the processing of quantum
states\footnote{Ones which do not, we call classical protocols: here
we are considering quantum relativistic protocols for classical
computations.}.  All of the computations we consider in this thesis
are classical in this sense (although most of the protocols we discuss
involve quantum information processing), and we will henceforth refer
to these as computations, with the term ``classical'' taken as
understood.  A perfectly secure computation guarantees, for each $i$,
each subset $J \subseteq \{ 1 , \ldots , N \}$, and each set of
possible inputs $x_i$ and $\{x_j\}_{j \in J}$, that if the parties $J$
do indeed input $\{x_j\}_{j \in J}$ and then collaborate, they can
gain no information about the input $x_i$ other than what is implied
by $\{x_j\}_{j \in J}$ and $\{f_j (x_1 , \ldots , x_N )\}_{j \in J}$.

Note that some tasks fall outside this model completely.  Bit
commitment, for example, requires that the output is at some time
fixed, but is not revealed until a later time.  Other computations
with this delay feature also fall outside the scope of our model.

We restrict attention here to two types of two-party computation: {\it
two-sided} computations in which the outputs prescribed for each party
are identical, and {\it one-sided} computations in which one party
gets no output.  We use the term {\it single function computations} to
cover both of these types, since, in both cases, only one function
need be evaluated.  We can classify single function computations by
the number of inputs (by which we mean the number of parties making an
input, as distinct from the size of the set of possible values of such
inputs), by whether they are deterministic or random, and by whether
one or two parties receive the output.

We are interested in protocols whose unconditional security is
guaranteed by the laws of physics.  In particular, as is standard in
these discussions, we do not allow any security arguments based on
technological or computational bounds: each party allows for the
possibility that the other may have arbitrarily good technology and
arbitrarily powerful quantum computers.  In addition, we assume that
Assumptions \ref{ass1}--\ref{ass4} (see Section \ref{thesetting})
hold.  Under such assumptions, the known results for secure
computations are summarized below.

{\bf Zero-input computations:} \qquad Secure protocols for zero-input
deterministic computations or zero-input random one-sided computations
can be trivially constructed, since the relevant computations can be
carried out by one or both parties separately.  The most general type
of zero-input two-sided random secure computation is a biased
$n$-faced secure die roll.  This can be implemented with unconditional
security by generalizing the relativistic protocol for a secure coin
toss given in Section \ref{relCT} as follows.  

\begin{protocol}\qquad

For an $n$-faced die with distribution $p_1,\ldots,p_n$,

\begin{enumerate}
\item $A_1$ creates a string, $X\in\{1,\ldots,n\}^N$, for which $Np_i$
members are $i$ for all $i\in\{1,\ldots,n\}$ ($N$ is such that $Np_i$
is an integer for all $i$, or, if the chosen probabilities are
irrational, we can get arbitrarily close to the correct distribution
by taking $N$ large).  The permutation of elements in the string is
chosen uniformly at random.  $A_1$ sends this string to
$B_1$. \footnote{For example, in an unbiased coin toss, $X$ is either
$(0,1)$ or $(1,0)$.  The second bit is redundant, hence the protocol
can be simplified to Protocol \ref{relCT_prot}, presented previously.}
\item $B_2$ simultaneously sends a random number,
$j\in\{1,\ldots,N\}$, to $A_2$.
\item $B_1$ checks that his received message arrived before time
$t_0+D$, and likewise, so does $A_2$.  If this is not the case, they
abort.
\item The disconnected agents of Alice communicate with one another,
as do those of Bob.  
\item Bob checks that the string he received from Alice has the
correct number of entries of each type.  If not, he aborts.  Otherwise
the outcome of the die roll is the $j$th member of $X$.
\end{enumerate}
\end{protocol}

\smallskip

{\bf One-input computations:} \qquad Secure protocols for
deterministic one-input computations are trivial; the party making the
input can always choose it to generate any desired output on the other
side and so might as well compute the function on their own and send
the output directly to the other party.

The non-deterministic case is of interest.  For one-sided computations,
where the output goes to the party that did not make the input, the
most general function is a one-sided variable bias $n$-faced die roll.
The input simply defines a probability distribution over the outputs.
In essence, one party chooses one from a collection of biased
$n$-faced dice to roll (the members of the collection are known to
both parties).  The output of the roll goes to one party only, who has
no other information about which die was chosen.

It is known that some computations of this type are impossible.
Oblivious transfer falls into this class, for instance, and is shown
to be impossible in Section \ref{OT_imposs}.  \footnote{To see that
{\sc OT} can be thought of as an example of a one-sided variable bias
$n$-faced die roll, consider the probability table, Table \ref{tabOT},
in Section \ref{OT_imposs}.  The computation can be though of as
having Alice pick one of two three sided die to roll (the three sides
being labelled $0$, $1$ and $?$).} In Chapter \ref{Ch3}, we discuss
other computations of this type, and show that they are impossible to
implement securely.

We call the two-sided case of a non-deterministic one-input function a
variable bias $n$-faced die roll.  This---and particularly the
two-faced case, a variable bias coin toss---is the subject of the
present chapter.  We will give a protocol that implements the task
with unconditional security for a limited range of biases, another
which permits any range of biases but achieves only cheat-evident
security, and two further protocols that allow any range of biases and
which we conjecture are unconditionally secure.  Such tasks are
impossible in non-relativistic cryptography.

\smallskip

{\bf Two-input computations:} \qquad Lo \cite{Lo} considered the task
of finding a secure nonrelativistic quantum protocol for a two-input,
deterministic, one-sided function. He showed that if the protocol
allows Alice to input $i$, Bob to input $j$, and Bob to receive
$f(i,j)$, while giving Alice no information on $j$, then Bob can also
obtain $f(i,j')$ for all $j'$.  For any cryptographically nontrivial
computation, there must be at least one $i$ for which knowing
$f(i,j')$ for all $j'$ gives Bob more information than knowing
$f(i,j)$ for just one value of $j$.  As this violates the definition
of security for a secure classical computation, it is impossible to
implement any cryptographically nontrivial computation securely.  

Lo's proof as stated applies to nonrelativistic protocols.  He showed
that there cannot exist a set of states
$\{\ket{\psi^{i,j}_{AB}}\}_{i,j}$, shared between Alice and Bob that
fulfil the requirements of such a computation.  In a relativistic
computation where all measurements are kept quantum until the end, the
final state must again be an ($i$, $j$)-dependent pure state
distributed between Alice and Bob.  Lo's impossibility result
therefore extends trivially to relativistic protocols\footnote{Rudolph
\cite{Rudolph} has defined the notion of a {\it consistent} task as
one for which there exist states shared between the parties, and local
operations which could satisfy the security demands.  Inconsistent
tasks are then impossible whether or not the protocol is relativistic,
hence Lo's proof also works in this scenario.  Consistent tasks are
not necessarily possible: they require a way to securely generate a
shared state of the correct form.  Whether this is achievable can
depend on whether a relativistic protocol is used or not.}.

Lo also noted that some secure two-input deterministic, two-sided
non-relativistic quantum computations are impossible, because they
imply the ability to do non-trivial secure two-input, deterministic
one-sided computations.  This argument also extends trivially to
relativistic protocols.

We will discuss further protocols in this class in detail in Chapter
\ref{Ch3}.

\bigskip

Table \ref{fns} summarizes the known results.

\begin{sidewaystable}
\begin{tabular}{|l|l|c|c|l|}
\hline
\multicolumn{2}{|c|}{Type of computation}&Securely Implementable&Comment\\
\hline
Zero-input&Deterministic &$ \checkmark $  &Trivial\\
        &Random       one-sided         &$ \checkmark $  &Trivial\\
        &Random       two-sided         &\checkmark   &Biased $n$-faced die roll\\
\hline
One-input &Deterministic &$\checkmark $  &Trivial\\
        &Random       one-sided         &(\ding{55})&One-sided
variable bias $n$-faced die roll\\
        &Random       two-sided         &$\checkmark^*$ &Variable bias
$n$-faced die roll\\
\hline
Two-input &Deterministic one-sided         &\ding{55}    &cf.\ Lo\\
        &Deterministic two-sided        &(\ding{55})  &cf.\ Lo\\
        &Random       one-sided        &? &see Chapter \ref{Ch3}\\
        &Random       two-sided        &? &see Chapter \ref{Ch3}\\
\hline
\end{tabular}
\caption{Functions computable securely in two-party computations using
  (potentially) both quantum and relativistic protocols, when
  unconditional security is sought.  $\checkmark$ indicates that all
  functions of this type are possible, \ding{55} indicates that all
  functions of this type are impossible, $\checkmark^*$ indicates that
  conjectures made later in this chapter imply that all functions of
  this type are possible, and (\ding{55}) indicates that some
  functions of this type are impossible. ? indicates an unknown result
  (to be discussed in Chapter \ref{Ch3}).  An updated version of this
  table, Table \ref{fns2}, is given at the end of Chapter \ref{Ch3}.}
\label{fns}
\end{sidewaystable}

\section{Variable Bias Coin Tossing}
\label{secvbct}
\subsection{Introduction}
We now specialize to the task of variable bias coin tossing ({\sc
VBCT}), the simplest case of a one-input, random, two-sided
computation.  We seek protocols whose security is guaranteed based on
the laws of physics.

The aim of a {\sc VBCT} protocol is to provide two mistrustful parties with
the outcome of a biased coin toss.  We label the possible outcomes by
$0$ and $1$ and define the {\it bias} to be the probability, $p_0$, of
outcome $0$.  The protocol should allow one party, by convention Bob,
to fix the bias to take any value within a pre-agreed range, $p_{\rm
min} \leq p_0 \leq p_{\rm max}$.  The protocol should guarantee to
both parties that the biased coin toss outcome is genuinely random, in
the sense that Bob's only way of influencing the outcome probabilities
is through choosing the bias, while Alice has no way of influencing
the outcome probabilities at all.  It should also guarantee to Bob
that Alice can obtain no information about his bias choice beyond what
she can infer from the coin toss outcome alone.

To illustrate the uses of {\sc VBCT}, consider a situation in which Bob may
or may not wish to accept Alice's invitation to a party, in a future
world in which social protocol decrees that his
decision\footnote{Naturally, a similar protocol, in which Alice
chooses the bias, governs the decision about whether or not an
invitation is issued.}  is determined by a variable bias coin toss in
which he chooses the bias within a prescribed range, let us say
$p_{\rm min} = \frac{1}{11} \leq p_0 \leq p_{\rm max} =
\frac{10}{11}$.  Alice, who is both self-confident and a Bayesian,
believes prior to the coin toss that the probability of Bob not
wishing to accept is $10^{-n}$, for some fairly large value of $n$.
If Bob does indeed wish to accept, he can choose $p_0 =
\frac{10}{11}$, ensuring a high probability of acceptance.  If he does
not, he can choose $p_0 = \frac{1}{11}$, ensuring a low probability of
acceptance.  If the invitation is declined, this social protocol
allows both parties to express regret, ascribing the outcome to bad
luck rather than to Bob's wishes.  Alice's posterior probability
estimate of Bob's not wishing to attend is approximately $10^{-n+1}$,
i.e., still close to zero.

For another illustration of the uses of {\sc VBCT}, suppose that Bob
has a large secret binary data set of size $N$.  For example, this
might be a binary encoding of a high resolution satellite image.  He
is willing to sell Alice a noisy image of the data set with a
specified level of random noise.  Alice is willing to purchase if
there is some way of guaranteeing, at least to within tolerable
bounds, that the noise is at the specified level and that it was
genuinely randomly generated.  In particular, she would like some
guarantee that constrains Bob so that he cannot selectively choose the
noise so as to obscure a significantly sized component of the data set
which he (but not necessarily she) knows to be especially interesting.
Let us suppose also that the full data set will eventually become
public, so that Alice will be able to check the noisy image against
it, and that she will be able to enforce suitably large penalties
against Bob if the noisy and true versions turn out not to be
appropriately related.  They may proceed by agreeing on parameters
$p_{\rm min}$ and $p_{\rm max} = 1 - p_{\rm min}$, and then running a
variable bias coin toss for each bit in the image, with Bob choosing
$p_0 = p_{\rm min}$ if the bit is $1$ and $p_0 = p_{\rm max}$ if the
bit is $0$.  Following this protocol honestly provides Alice with the
required randomly generated noisy image.  On the other hand, if Bob
deviates significantly from these choices for more than $O( \sqrt{N}
)$ of the bits, Alice will almost certainly be able to unmask his
cheating once she acquires the full data set.

\subsection{Definitions}
A {\sc VBCT} protocol is defined by a prescribed series of classical or
quantum communications between two parties, Alice and Bob.  If the
protocol is relativistic, it may also require that the parties each
occupy two or more appropriately located sites and may stipulate which
sites each communication should be made from and to.  The protocol's
definition includes bias parameters $p_{\rm min}$ and $p_{\rm max}$,
with $p_{\rm min} < p_{\rm max}$, and may also include one or more
security parameters $N_1 , \ldots , N_r$.  It accepts a one bit input
from one party, Bob, and must result in both parties receiving the
same output, one of the three possibilities $0$, $1$ or ``abort''.
The output ``abort'' can arise only if at least one of the parties
refuses to complete the protocol honestly.

We follow the convention that Bob can fix $p$ to be $p_{\rm min}$ or
$p_{\rm max}$ by choosing inputs $1$ or $0$ respectively (so that an
input of bit value $b$ maximizes the probability of output $b$).  He
can thus fix $p$ anywhere in the range $p_{\rm min} \leq p \leq p_{\rm
max}$ by choosing the input randomly with an appropriate weighting.
Since any {\sc VBCT} protocol gives Bob this freedom, we do not require a
perfectly secure protocol to exclude other strategies which have the
same result: i.e., a perfectly secure protocol may allow any strategy
of Bob's which causes $p_0$ to lie in the given range, so long as no
other security condition is violated. \footnote{Similar statements
hold, with appropriate epsilonics, for secure protocols: see below.}
However, if Bob in honest, he chooses either $p=p_{\rm min}$ or
$p=p_{\rm max}$.  This motivates the following security definitions.

We say the protocol is {\it secure} if the following conditions hold
when at least one party honestly follows the protocol.  Let $p_0$ be
the probability of the output being $0$, and $p_1$ be the probability
of the output being $1$.  Then, regardless of the strategy that a
dishonest party may follow during the protocol, we have $ p_0 \leq p +
\epsilon (N_1 , \ldots , N_r )$ and $ p_1 \leq (1-p) + \epsilon (N_1 ,
\ldots , N_r )$, where $p_{\rm min} \leq p \leq p_{\rm max}$ and the
protocol allows Bob to determine $p$ to take any value in this range.
Alice has probability less than $\zeta(N_1,\ldots,N_r)$ of obtaining
more than $\delta(N_1 , \ldots , N_r)$ bits of information that are
not implied by the outcome.  In addition, if Bob honestly follows the
protocol and legitimately aborts before the coin toss outcome is
known\footnote{We take this to be the point at which both parties have
enough information (possibly distributed between their remote agents)
to determine the outcome.}, then Alice has probability less than
$\zeta(N_1 , \ldots , N_r)$ of obtaining more than $\delta(N_1 ,
\ldots , N_r )$ bits of information about Bob's input.

(We should comment here on a technical detail that will be relevant to
some of the protocols we later consider.  It turns out, in some of our
protocols, to be possible and useful for Bob to make supplementary
security tests even after both parties have received information which
would determine the coin toss outcome.  The protocols are secure
whether or not these supplementary tests are made, in the sense that
the security criteria hold as the security parameters tend to
infinity.  However, the supplementary tests increase the level of
security for any fixed finite value of the security parameters.

We need slightly modified definitions to cover this case, since the
output of the protocol is defined to be ``abort'' if Bob aborts after
carrying out supplementary security tests.  If Bob honestly follows a
protocol with supplementary tests, and legitimately aborts after the
coin toss outcome is determined, then we require that Alice should
have probability less than $\zeta(N_1 , \ldots , N_r)$ of obtaining
more than $\delta (N_1 , \ldots , N_r )$ extra bits of
information---i.e., beyond what is implied by the coin toss outcome.

Note that introducing supplementary security tests may allow Alice to
follow the protocol honestly until she obtains the coin toss outcome,
and then deliberately fail the supplementary tests in order to cause
the protocol to abort.  However, this may not give her useful extra
scope for cheating.  In a {\sc VBCT} protocol in which the coin toss
outcome has some real world consequence, for instance, Alice can
always follow the protocol honestly and then refuse to abide by the
consequence dictated by its outcome: for example, she can decide not
to invite Bob to her party, even if the variable bias coin toss
suggests that she should.  This unavoidable possibility has the same
effect as her causing the protocol to abort after the coin toss
outcome is determined.)

In all the above cases, we require $\delta (N_1 , \ldots , N_r )
\rightarrow 0$, $\epsilon (N_1 , \ldots , N_r ) \rightarrow 0$ and
$\zeta(N_1,\ldots,N_r)\rightarrow 0$ as the $N_i \rightarrow \infty$.
We say the protocol is {\it perfectly secure} for some fixed values
$N_1 , \ldots , N_r$ if the above conditions hold with $\epsilon(N_1 ,
\ldots , N_r ) = \delta (N_1 , \ldots , N_r ) = \zeta(N_1 , \ldots ,
N_r ) = 0$.

Suppose now that one party is honest and the other party fixes their
strategy (which may be probabilistic and may depend on data received
during the protocol) before the protocol commences, and suppose that
the probability of the protocol aborting, given this strategy, is less
than $\epsilon'$.  Since the only possible outcomes are $0$, $1$ and
``abort'', it follows from the above conditions that, if Bob inputs
$1$, we have $p_{\rm min} - \epsilon (N_1 , \ldots , N_r ) - \epsilon'
< p_0 \leq p_{\rm min} + \epsilon (N_1 , \ldots , N_r )$ and $(1 -
p_{\rm min}) - \epsilon (N_1 , \ldots , N_r ) - \epsilon' < p_1 \leq
(1 - p_{\rm min}) + \epsilon (N_1 , \ldots , N_r )$.  Similarly, if
Bob inputs $0$, we have $p_{\rm max} - \epsilon (N_1 , \ldots , N_r )
- \epsilon' < p_0 \leq p_{\rm max} + \epsilon (N_1 , \ldots , N_r )$
and $(1 - p_{\rm max}) - \epsilon (N_1 , \ldots , N_r ) - \epsilon' <
p_1 \leq (1 - p_{\rm max}) + \epsilon (N_1 , \ldots , N_r )$.  In
other words, unless a dishonest party is willing to accept a
significant risk of the protocol aborting, they cannot cause the
outcome probabilities for $0$ or $1$ to be significantly outside the
allowed range.  Moreover, no aborting strategy can increase the
probability of $0$ or $1$ beyond the allowed maximum.

For an {\it unconditionally secure} {\sc VBCT} protocol, the above
conditions hold assuming only that the laws of physics are correct. In
a {\it cheat-evidently secure} protocol, if any of the above
conditions fail, then the non-cheating party is guaranteed to detect
this, again assuming only the validity of the laws of physics.

\section{VBCT Protocols}

\subsection{Protocol VBCT1}

We consider first a simple relativistic quantum protocol, which
implements {\sc VBCT} with unconditional security, for a limited range of
biases.  The protocol requires each party to have agents located at
three appropriately separated sites.
\bigskip\newline
{\bf Protocol {\sc VBCT}1}

\begin{enumerate}

\item \label{poisson} $B_1$, $B_2$ and $B_3$ agree on a random number
  $n$ chosen from a Poisson distribution with large mean (or other
  suitable distribution).

\item $A_1$ sends a sequence of qubits $\{\ket{\phi_i}\}$ to $B_1$,
where each $\ket{\phi_i} \in \{ \ket{ \psi_0},\ket{\psi_1} \}$ is
chosen independently with probability half each, using the random
string ${\bf x}$. The states $\ket{ \psi_0}$ and $\ket{\psi_1}$ are
agreed between Alice and Bob prior to the protocol, and the qubits are
sent at regular intervals according to a previously agreed schedule,
so that all the agents involved can coordinate their transmissions.

\item $B_1$ receives each qubit and stores it.

\item \label{seq} $A_2$ tells $B_2$ the sequence of states $\{
\ket{\phi_i} \}$ sent, choosing the timings so that $A_1$'s quantum
communication of the qubit $\ket{\phi_i}$ is spacelike separated from
$A_2$'s classical communication of its identity.  $B_2$ relays these
communications to $B_1$.

\item \label{announce} $B_3$ announces to $A_3$ that the $n$th state
will be used for the coin toss.  This announcement is made at a point
spacelike separated from the $n$th rounds of communication between
$A_1$ and $B_1$ and $A_2$ and $B_2$.  $A_3$ reports the value of $n$
to $A_1$ and $A_2$.

\item \label{Bmeasure} $B_1$ performs the measurement on
$\ket{\phi_n}$ that optimally distinguishes $\ket{\psi_0}$ from
$\ket{\psi_1}$, and then reveals $n$ to $A_1$, along with a bit $b$.
If his measurement is indicative of the state being $\ket{\psi_{b'}}$,
then Bob should select $b=b'$ if he wants outcome 0, or else select
$b=\bar{b'}$.  Let Alice's random choice for the $n$th state be
$\ket{\psi_a}$: recall that $A_2$ reported the value of $a$ to $B_2$
in Step \ref{seq}.

\item \label{Bcheck1} Some time later, on receipt of the sequence sent
by $B_2$ in Step \ref{seq}, $B_1$ measures his remaining stored states
to verify that they were correctly described by $A_2$.  If any error
occurs, he aborts.

\item \label{receive} $A_1$ receives from $A_3$ the value of $n$ sent
by $B_3$, confirming that $B_1$ was committed to guess the $n$th
state, and $B_1$ receives from $B_2$ the value of $a$ sent by $A_2$.
The outcome of the coin toss is $c= a \oplus b$.
\end{enumerate}

It will be seen that this protocol is a variant of the familiar
relativistic protocol for ordinary coin tossing.  As in that protocol,
Alice and Bob simultaneously exchange random bits.  However, Alice's
bit is here encoded in non-orthogonal qubits, which means that Bob can
obtain some information about its value.  Bob uses this information to
affect the bias of the coin toss.

We use the bit $w$ to represent Bob's wishes, with $w=0$ representing
Bob trying to produce the outcome $0$ by guessing correctly, and $w=1$
representing him trying to produce the outcome $1$ by guessing
wrongly. Security requires that
\begin{equation}
\label{criterion}
p(w|a,b,c) \approx p(w|c),  
\end{equation}
i.e.\ the bits $a$ and $b$ convey no information about Bob's wishes.
Perfect security requires equality in the above equation.

\subsubsection{Bob's Strategy}
The choice of $n$ need not be fixed by Bob at the start of the
protocol: for example, it could be decided during the protocol by
using an entangled state shared by the $B_i$.  However, we may assume
$B_3$ sends a classical choice of $n$ to $A_3$ ($A_3$ will measure any
quantum state he sends immediately in the computational basis, and
hence we may assume, for the purposes of security analysis, that $B_3$
carries out this measurement).  $B_3$'s announcement of $n$ is
causally disconnected from the sending of the $n$th state to $B_1$ and
of its identity to $B_2$.  Therefore, no matter how it is selected, it
does not depend on the value of the $n$th state.  While it could be
generated in such a way as to depend on some information about the
sequence of states previously received, these states are uncorrelated
with the $n$th state if Alice follows the protocol.  Such a strategy
thus confers no advantage, and we may assume, for the purposes of
security analysis, that the choice of $n$ is generated by an
algorithm independent of the previous sequence of states.  We may also
assume that $n$ is generated in such a way that $B_1$ and $B_2$ can
obtain the value of $n$ announced by $B_3$ with certainty: if not,
their task is only made harder.  In summary, for the purposes of
security analysis, we may assume that $B_3$ announces a classical
value of $n$, pre-agreed with $B_1$ and $B_2$ at the beginning of the
protocol.

$B_1$ is then committed to making a guess of the value of the $n$th
state: if he fails to do so then Alice knows Bob has cheated.  $B_1$'s
best strategy is thus to perform some measurement on the $n$th state
and use the outcome to make his guess.  We define
$\ket{\psi_0}=\cos{\frac{\theta}{2}}\ket{0}+\sin{\frac{\theta}{2}}\ket{1}$
and
$\ket{\psi_1}=\cos{\frac{\theta}{2}}\ket{0}-\sin{\frac{\theta}{2}}\ket{1}$,
where $0\leq\theta\leq\frac{\pi}{2}$.  Let the projections defining
the optimal measurement be $P_0$ and $P_1$.  We say that the outcome
corresponding to $P_0$ is ``outcome 0'', and similarly for the outcome
corresponding to $P_1$.  Without loss of generality, we can take
outcome 0 to correspond to the most likely state Alice sent being
$\ket{\psi_0}$ and similarly outcome 1 to correspond to
$\ket{\psi_1}$. Bob's probability of guessing correctly is then given
by,
\begin{equation}
p_B = \frac{1}{2}\left(\bra{\psi_0}P_0\ket{\psi_0}+\bra{\psi_1}P_1\ket{\psi_1}\right) \, .
\end{equation}
This is maximized for $P_0$ and $P_1$ corresponding to measurements in
the $\ket{\pm}$ basis, where
$\ket{\pm}=\frac{1}{\sqrt{2}}(\ket{0}\pm\ket{1})$.  The maximum value
is,
\begin{equation}
\label{p_Bmax}
p_B^{\rm{max}}=\frac{1}{2}\left( 1+\sin{\theta}\right).
\end{equation}

It is easy to see that the security criterion (\ref{criterion}) is
always satisfied.  Furthermore, the outcome, $c$ can be used by either
party to simulate the intermediates produced in the protocol (i.e.,
$a$, $b$, and the set of quantum states), making it clear that no
information is gained, other than that implied by the outcome (the
r\^ole of simulatability in security will be discussed further in
Section \ref{sec_defs}).  The minimum probability of Bob guessing
correctly is always $1-p_B^{\rm{max}}$, which he can attain by following
the same strategy but associating outcome $b'$ with a guess of
$\bar{b'}$.  The possible range of biases are those between $p_{\rm
min}=\frac{1}{2} \left( 1-\sin{\theta} \right)$ and $p_{\rm
max}=\frac{1}{2} \left( 1+\sin{\theta} \right)$.  The protocol thus
implements {\sc VBCT} for all values of $p_{\rm min}$ and $p_{\rm max}$ with
$p_{\rm min} + p_{\rm max} =1$ (and no others).

\subsubsection{Security Against Alice}

Security against Alice is ensured by the fact that $B_1$ tests $A_2$'s
statements about the identity of the states sent to $B_1$.

We seek to show that if Alice attempts to alter the probability of
$B_1$ measuring $0$ or $1$ with his measurement in Step
\ref{Bmeasure}, then in the limit of large $n$, either the probability
of her being detected tends to $1$, or her probability of successfully
altering the probability tends to zero.  Note that it may be useful
for Alice to alter the probabilities in either direction: if she
increases the probability that $B_1$ guesses correctly, she learns
more information about Bob's bias than she should; if she decreases
it, she limits Bob's ability to affect the bias.

We need to show that if, on the $i$th round, $B_1$ receives state
$\rho_i$, for which the probability of outcome $0$ differs from those
dictated by the protocol, then the probability of $B_1$ not detecting
Alice cheating on this state is strictly less than $1$.

$B_1$'s projections are onto $\{\ket{+},\ket{-}\}$ for the $n$th
state.  Alice's cheating strategy must ensure that for some subset of
the states she sends to $B_1$, there is a different probability of his
measurement giving outcome $0$.  Suppose that $\rho_i$ satisfies
\begin{eqnarray}
\label{overlap}
\bra{+}\rho_i\ket{+}&=p_{\rm max}+\delta_1\\
&=p_{\rm min}+\delta_2\, ,
\end{eqnarray}
where $\delta_1,\delta_2\neq 0$.  Then, if $B_1$ was to instead test
Alice's honesty, the state which maximizes the probability of Alice
passing the test, among those satisfying (\ref{overlap}), is
\begin{equation}
(p_{\rm max}+\delta_1)^{\frac{1}{2}}\ket{+}+(1-p_{\rm max}-\delta_1)^{\frac{1}{2}}\ket{-},
\end{equation}
and she should declare this state to be whichever of $\ket{\phi_0}$ or
$\ket{\phi_1}$ maximizes the probability of passing Bob's test.  We
have
\begin{equation}
\left(( p_{\rm max} ( p_{\rm max} + \delta_1 ))^{\frac{1}{2}} + ( (
1-p_{\rm max}) (1-p_{\rm max} - \delta_1 ))^{\frac{1}{2}}\right)^2
\leq 1-\delta_1^2\, ,
\end{equation}
and a similar equation with $p_{\rm min}$ replacing $p_{\rm max}$ and
$\delta_2$ replacing $\delta_1$.  Hence the probability of passing
Bob's test is at most $1 - \delta^2$, where $\delta = \min (|\delta_1|
, |\delta_2| )$.  In order to affect $B_1$'s measurement probabilities
with significant chance of success, there must be a significant
fraction of states satisfying (\ref{overlap}).  If a fraction $\gamma$
of states satisfy (\ref{overlap}) with $\min (|\delta_1| , |\delta_2|
) \geq \delta$ for some fixed $\delta >0$, then this cheating strategy
succeeds with probability at most $\gamma(1-\delta^2)^{\gamma n}$.
Hence, for any $\delta$, $\gamma$, the probability of this technique
being successful for Alice can be made arbitrarily close to $0$ if Bob
chooses the mean of the Poisson distribution used in Step
\ref{poisson} (and hence the expected value of $n$) to be sufficiently
large.

Note that, as this argument applies state by state to the $\rho_i$, it
covers every possible strategy of Alice's: in particular, the argument
holds whether or not the sequence of qubits she sends is entangled.

We hence conclude that the protocol is asymptotically secure against
Alice.

\subsection{Protocol VBCT2}

We now present a relativistic quantum {\sc VBCT} protocol which allows any
range of biases, but achieves only cheat-evident security rather than
unconditional security.
\bigskip\newline
{\bf Protocol {\sc VBCT}2}

\begin{enumerate}
\item \label{first step} $B_1$ creates $N$ states, each being either
$\ket{\psi_0}=\alpha_0\ket{00}+\beta_0\ket{11}$ or
$\ket{\psi_1}=\alpha_1\ket{00}+\beta_1\ket{11}$, with
$\{\alpha_0,\alpha_1,\beta_0,\beta_1\}\in\mathbb{R}^+$, $\alpha_0^2 >
\alpha_1^2$, and $\alpha_i^2+\beta_i^2=1$.  The states are chosen with
probability half each.  In the unlikely event that all the states are
the same, $B_1$ rejects this batch and starts again.  $B_1$ uses the
shared random string ${\bf y}$ to make his random choices, so that
$B_1$ and $B_2$ both know the identity of the $i$th state.  $B_1$
sends the second qubit of each state to $A_1$.  The values of
$\alpha_0,\beta_0,\alpha_1$ and $\beta_1$ are known to both Alice and
Bob.  We define the bias of the state $\ket{\psi_i}$ to be
$\alpha_i^2$, and write $p_{\rm min} = \alpha_1^2$ and $p_{\rm max} =
\alpha_0^2$.

\item \label{Alice decides} Alice decides whether to test Bob's
honesty ($z=1$), or to trust him ($z=0$).  She selects $z=0$ with
probability $2^{-M}$. $A_1$ and $A_2$ simultaneously inform $B_1$ and
$B_2$ of $z$, $A_2$'s communication being spacelike separated from the
creation of the states by $B_1$ in Step \ref{first step}.

\item \begin{enumerate} \item If $z=1$, $B_1$ sends all of his qubits
and their identities to $A_1$, while $B_2$ sends the identities to
$A_2$.  $A_1$ can then verify that they are as claimed and if so, the
protocol returns to Step \ref{first step}.  If not, she aborts the
protocol.
  
\item \label{Bcheck} If $z=0$, $B_1$ randomly chooses a state to use
for the coin toss from among those with the bias he wants.  $B_2$
simultaneously informs $A_2$ of $B_1$'s choice. \end{enumerate}

\item $A_1$ and $B_1$ measure their halves of the chosen state in the
$\{\ket{0},\ket{1}\}$ basis, and this defines the outcome of the coin
toss.  \renewcommand{\labelenumi}{(\arabic{enumi}.}

\item \label{extratest} As an optional supplementary post coin toss
security test, $B_1$ may ask $A_1$ to send all her remaining qubits
back to him, except for her half of the state selected for the coin
toss. He can then perform projective measurements on these states to
check that they correspond to those originally sent.)

\end{enumerate}

An intuitive argument for security of this protocol is as follows.  On
the one hand, as $M\rightarrow\infty$, the protocol is secure against
Bob since, in this limit, he always has to convince Alice that he
supplied the right states which he can only do if he has been honest.
But also, in the limit $N\rightarrow\infty$, we expect the protocol to
be secure against Alice, since in this limit, she cannot gain any more
information about the bias Bob selected than can be gained by
performing the honest measurement.

The protocol can only provide cheat-evident security rather than
unconditional security, since there are useful cheating strategies
open to Alice, albeit ones which will certainly be detected.  One such
strategy is for $A_1$ to claim that $z=0$ on some state, while $A_2$
claims that $z=1$.  This allows Alice to determine Bob's desired bias,
since $B_1$ will tell $A_1$ the state to use, and $B_2$ will tell
$A_2$ its identity.  However, this cheating attack will be exposed
once $B_1$ and $B_2$ communicate.

(Technically, Alice has another possible attack: she can follow the
protocol honestly until she learns the outcome, and then intentionally
try to fail Bob's tests in Step \ref{extratest} by altering her halves
of the remaining states in some way.  By so doing, she can cause the
protocol to abort after the coin toss outcome is determined.  However,
as discussed in Section \ref{secvbct}, this gives her no advantage.)

\subsubsection{Security Against Alice}
Assume Bob does not deviate from the protocol.  $A_2$ must announce
the value of $z$ without any information about the current batch of
states sent to $A_1$ by $B_1$.  Alice therefore cannot affect the
bias: once a given batch is accepted, she cannot affect $B_1$'s
measurement probabilities on any state he chooses for the coin toss.
While Alice's choices of $z$ need not be classical bits determined
before the protocol and shared by the $A_i$, we may assume, for the
purposes of security analysis, that they are, by the same argument
used in analyzing Bob's choice of $n$ in {\sc VBCT}1.

Once Bob has announced the state he wishes to use for the coin toss,
though, Alice can perform any measurement on the states in her
possession in order to gain information about Bob's chosen bias.  It
would be sufficient to show that any such attack that provides
significant information is almost certain to be detected by Bob's
tests in Step \ref{Bcheck}; if so, the existence of such attacks would
not compromise the cheat-evident security of the protocol.  In fact, a
stronger result holds: Alice cannot gain significant information by
such attacks.  From her perspective, if Bob selects a $\ket{\psi_0}$
state for the coin toss, the (un-normalized) mixed state of the
remaining $(N-1)$ qubits is,
\begin{equation}
\tilde{\sigma}_0 \equiv \sum_{m=0}^{N-2} \,
\sum_{\genfrac{}{}{0pt}{}{i_1 , \ldots , i_{N-1} \in \{ 0,
1\}}{\sum_{j=1}^{N-1} i_j = (N-1-m)}} \rho_{i_1} \otimes \rho_{i_2}
\otimes \cdots \otimes\rho_{i_{N-1}} \, ,
\end{equation}
while if Bob selects a $\ket{\psi_1}$ state for the coin toss, the
(un-normalized) mixed state of the remaining $(N-1)$ qubits is
\begin{equation}
\tilde{\sigma}_1 \equiv \sum_{m=1}^{N-1} \, \sum_{\genfrac{}{}{0pt}{}{i_1 , \ldots ,
i_{N-1} \in \{ 0, 1\}}{\sum_{j=1}^{N-1} i_j = (N-1-m)}}
\rho_{i_1} \otimes \rho_{i_2} \otimes \cdots \otimes \rho_{i_{N-1}} \, ,
\end{equation}
where
\begin{equation*}
\rho_i = \text{tr}_B ( \ketbra{\psi_i}{\psi_i} ) \qquad {\rm for~}i=0,1 \,
.
\end{equation*}
We will use $\sigma_0$ and $\sigma_1$ to denote the normalized
versions of $\tilde{\sigma}_0$ and $\tilde{\sigma}_1$ respectively.

We have
\begin{equation}
D(\rho_0\otimes\sigma_0,\rho_1\otimes\sigma_1)\leq
D(\rho_0\otimes\sigma_0,\rho_1\otimes\sigma_0)+
D(\rho_1\otimes\sigma_0,\rho_1\otimes\sigma_1)
\end{equation}
As $N \rightarrow \infty$, we have $D(\sigma_0 , \sigma_1) \rightarrow
0$ and so $D(\rho_0\otimes\sigma_0,\rho_1\otimes\sigma_1) \rightarrow
D(\rho_0 , \rho_1)$.  Since the maximum probability of distinguishing
two states is a function only of the trace distance (see Appendix
\ref{AppA}), the maximum probability of distinguishing
$\rho_0\otimes\sigma_0$ from $\rho_1\otimes\sigma_1$ tends, as $N
\rightarrow \infty$, to the maximum probability of distinguishing
$\rho_0$ from $\rho_1$.  The measurement that attains this maximum is
that dictated by the protocol.  
We hence conclude that, in the limit of large $N$, the excess
information Alice can gain by using any cheating strategy tends to
zero.

\subsubsection{Security Against Bob}
We now consider Bob's cheating possibilities, assuming that Alice does
not deviate from the protocol.  To cheat, Bob must achieve a bias
outside the range permitted.  Let us suppose he wants to ensure that
the outcome probability of $0$ satisfies $p_0 \geq p_{\rm max} +
\delta$, for some $\delta > 0$ (the case $p_1 \geq 1-p_{\rm min} +
\delta$ can be treated similarly), and let us suppose this can be
achieved with probability $\delta' > 0$.

For this to be the case, there must be some cheating strategy
(possibly including measurements) which, with probability $\delta'$,
allows $B_2$ to identify a choice of $i$ from the relevant batch of
$N$ qubits such that the state $\rho_i$ of $A_1$'s $i$th qubit then
satisfies
\begin{equation}
\label{zer}
\bra{0} \rho_i \ket{0} \geq p_{\rm max} + \delta .
\end{equation} 

If $A_1$'s $i$th qubit does indeed have this property, and she chooses
to test Bob's honesty on the relevant batch, the probability of the
$i$th qubit passing the test is at most $1 - \delta^2$.  To see this,
note that if (\ref{zer}) holds, the probability of passing the test is
maximized if the $i$th state is
\begin{equation} 
( p_{\rm max} + \delta)^{\frac{1}{2}} \ket{00} + ( 1 - p_{\rm max} -
\delta )^{\frac{1}{2}} \ket{11} \, ,
\end{equation}
and $B_1$ declares that the $i$th state is $\ket{\psi_0}$.  The
probability is then
\begin{equation}
\left(( p_{\rm max} ( p_{\rm max} + \delta ))^{\frac{1}{2}} + ( (
1-p_{\rm max}) (1-p_{\rm max} -\delta ))^{\frac{1}{2}}\right)^2 \leq
1-\delta^2\, .
\end{equation}

However, the probability of $A_1$'s measurement outcomes is
independent of $B_2$'s actions.  Hence this bound applies whether or
not $B_2$ actually implements a cheating strategy on the relevant
batch.  Thus there must be a probability of at least $\delta'
\delta^2$ of at least one member of the batch failing $A_1$'s tests.
Hence, for any given $\delta, \delta' > 0$, the probability that one
of the $\approx 2^M$ batches for which $z=1$ fails $A_1$'s tests can
be made arbitrarily close to $1$ by taking $M$ sufficiently large.

\subsection{Protocol VBCT3}
Protocol {\sc VBCT}2 can be improved by using bit commitment subprotocols to
keep Bob's choice of state secret until he is able to compare the
values of $z$ announced by $A_1$ and $A_2$.  This eliminates the
cheat-evident attack discussed in the last section, and defines a
protocol which we conjecture is unconditionally secure.  We use the
relativistic bit commitment protocol {\sc RBC}2 that is defined and reviewed
in \cite{Kent_relBC}.
\bigskip\newline
{\bf Protocol {\sc VBCT}3}

\begin{enumerate}
\item \label{1ststepz} $B_1$ creates $N$ states, each being either
$\ket{\psi_0}=\alpha_0\ket{00}+\beta_0\ket{11}$ or
$\ket{\psi_1}=\alpha_1\ket{00}+\beta_1\ket{11}$, with
$\{\alpha_0,\alpha_1,\beta_0,\beta_1\}\in\mathbb{R}^+$, and
$\alpha_i^2+\beta_i^2=1$.  The states are chosen with probability half
each.  $B_1$ and $B_2$ both know the identity of the $i$th state,
since $B_1$ uses the shared random string ${\bf y}$ to make his random
choices.  $B_1$ sends the second qubit of each state to $A_1$.  The
values of $\alpha_0,\beta_0,\alpha_1$ and $\beta_1$ are known to both
Alice and Bob.

\item Alice decides whether to test Bob's honesty, which she codes by
choosing the bit value $z=1$, or to trust him, coded by $z=0$.  She
selects $z=0$ with probability $2^{-M}$. $A_1$ and $A_2$
simultaneously inform $B_1$ and $B_2$ of the choice of $z$.

\item \label{Bsendsz} $B_1$ and $B_2$ broadcast the value of $z$ they
received to one another.

\item If $B_1$ received $z=1$ from $A_1$, he sends the first qubit of
each state to $A_1$, along with a classical bit identifying the state
as $\ket{\psi_0}$ or $\ket{\psi_1}$.  If $B_2$ received $z=1$ from
$A_2$, he sends $A_2$ a classical bit identifying the state as
$\ket{\psi_0}$ or $\ket{\psi_1}$.  These communications are sent
quickly enough that Alice is guaranteed that each of the $B_i$ sent
their transmission before knowing the value of $z$ sent to the other.
$A_2$ broadcasts the classical data to $A_1$ who tests that the
quantum states are those claimed in the classical communications by
carrying out the appropriate projective measurements.  If not, she
aborts.  If so, the protocol restarts at Step \ref{1ststepz}: $B_1$
creates a new set of $N$ states and proceeds as above.

\item If $z=0$, $A_2$ waits for time $\frac{D}{2}$ in the stationary
reference frame of $B_2$ before starting a series of relativistic bit
commitment subprotocols of type {\sc RBC}2 by sending the appropriate
communication (a list of suitably chosen random integers) to $B_2$.
$B_2$ verifies the delay interval was indeed $\frac{D}{2}$, to within
some tolerance.

\item $B_2$ continues the {\sc RBC}2 subprotocols by sending $A_2$
communications which commit Bob to the value of $i$ that defines the
state to use for the coin toss.

\item $B_1$ and $B_2$ then wait a further time $\frac{D}{2}$, by which
point they have received the signals sent in Step \ref{Bsendsz}.  They
then check that the $z$ values they received from the $A_i$ are the
same. If not, they abort the protocol.

\item $B_1$ and $B_2$ send communications to $A_1$ and $A_2$ which
unveil the value of $i$ to which they were committed, and hence reveal
the state chosen for the coin toss.  If the unveiling is invalid,
Alice aborts.

\item $A_1$ and $B_1$ measure their halves of the $i$th state in the
$\{\ket{0},\ket{1}\}$ basis to define the outcome of the coin toss.
\renewcommand{\labelenumi}{(\arabic{enumi}.}

\item \label{Bchecktwo} As an optional supplementary post coin toss
security test, $B_1$ asks $A_1$ to return her qubits from all states
other than the $i$th.  He then tests that the returned states are
those originally sent, by carrying out appropriate projective
measurements.  If the tests fail, he aborts the protocol.)
\end{enumerate}

\subsubsection{Security Against Alice}
In this modification of Protocol {\sc VBCT}2, there is no longer any
advantage to Alice in cheating by arranging that one of the $A_i$
sends $z=0$ and the other $z=1$.  Such an attack will be detected with
certainty, as is the case with Protocol {\sc VBCT}2.  Moreover, since Bob's
chosen value of $i$ is encrypted by a bit commitment, which is only
unveiled once the $B_i$ have checked that the values of $z$ they
received are identical, Alice gains no information about Bob's chosen
bias from the attack.  The bit commitment subprotocol {\sc RBC}2 is
unconditionally secure against Alice \cite{Kent_relBC}, since the
communications she receive are, from her perspective, uniformly
distributed random strings.

(As in the case of {\sc VBCT}2, technically speaking, Alice has another
possible attack: she can follow the protocol honestly up to Step
\ref{Bchecktwo} and then, once she learns Bob's chosen state,
intentionally try to fail Bob's tests by altering her halves of the
remaining states in some way.  By so doing, she can cause the protocol
to abort after the coin toss outcome is known.  Again, though, this
gives her no advantage.)

The protocol therefore presents Alice with no useful cheating attack.

\subsubsection{Security Against Bob}
Intuitively, one might expect the proof that {\sc VBCT}2 is secure against
Bob to carry over to a proof that {\sc VBCT}3 is similarly secure, for the
following reasons.  First, the only difference between the two
protocols is that Bob makes a commitment to the value of $i$ rather
than announcing it immediately.  Second, when the bit commitment
protocol {\sc RBC}2 is used, as here, just for a single round of
communications, it is provably unconditionally secure against general
(classical or quantum) attacks by Bob.

To make this argument rigorous, one would need to show that {\sc RBC}2 and
the other elements of {\sc VBCT}3 are securely {\it composable} in an
appropriate sense: i.e., that Bob has no collective quantum attack
which allows him to generate and manipulate collectively the data used
in the various steps of {\sc VBCT}3 in such a way as to cheat.  We
conjecture that this is indeed the case, but have no proof.

\subsection{Protocol VBCT4}
Classical communications and information processing are generally less
costly than their quantum counterparts, so much so that, in some
circumstances, it is reasonable to treat classical resources as
essentially cost free compared to quantum resources.  It is thus
interesting to note the existence of a classical relativistic protocol
for {\sc VBCT}, which is unconditionally secure against classical attacks,
and which we conjecture is unconditionally secure against quantum
attacks.  The protocol requires Alice and Bob each to have two
appropriately located agents, $A_1$, $A_2$ and $B_1$, $B_2$.
\bigskip\newline
{\bf Protocol {\sc VBCT}4}

\begin{enumerate}
\item Bob generates a $2M\times N$ matrix of bits such that each row
contains either $\alpha_0^2 N$ zero entries or $\alpha_1^2 N$ zero
entries, these being positioned randomly throughout the row.  The rows
are arranged in pairs, so that, for $m$ from $0$ to $(M-1)$, either
the $2m$th row contains $\alpha_0^2 N$ entries and the $(2m+1)$th
contains $\alpha_1^2 N$, or vice versa.  This choice is made randomly,
equiprobably and independently for each pair.  The matrix is known to
both $B_1$ and $B_2$ but kept secret from Alice.

\item Bob then commits each element of the matrix separately to Alice
using the classically secure relativistic bit commitment subprotocol
{\sc RBC}2 \cite{Kent_relBC}, initiated by communications between $A_2$ and
$B_2$.

\item $A_1$ then picks $M-1$ pairs at random.  She asks $B_1$ to
unveil Bob's commitment for all of the bits in these pairs of rows.

\item The {\sc RBC}2 commitments for the remaining bits are sustained while
$A_1$ and $A_2$ communicate to verify that each unveiling corresponds
to a valid commitment to either 0 or 1.  Alice also checks that each
unveiled pair contains one row with $\alpha_0^2 N$ zeros and one with
$\alpha_1^2 N$ zeros.  If Bob fails either set of tests, Alice aborts.

\item If Bob passes all of Alice's tests, $B_1$ picks the remaining
row corresponding to the bias he desires, and $A_2$ simultaneously
picks a random column.  They inform $A_1$ and $B_2$ respectively, thus
identifying a single matrix element belonging to the intersection.

\item Bob then unveils this bit, which is used as the outcome of the
coin toss.  The remaining commitments are never unveiled.
\end{enumerate}

\subsubsection{Security} 
The above protocol shows that, classically, bit commitment can be used
as a subprotocol to achieve {\sc VBCT}.  The proof that {\sc RBC}2 is
unconditionally secure against classical attacks \cite{Kent_relBC} can
be extended to show that Protocol {\sc VBCT}4 is similarly secure.  {\sc RBC}2 is
conjectured, but not proven, to be secure against general quantum
attacks.  We conjecture, but have no proof, that the same is true of
Protocol {\sc VBCT}4.

\section{Summary}
We have defined the task of variable bias coin tossing, illustrated
its use with a couple of applications, and presented four {\sc VBCT}
protocols.  The first, {\sc VBCT}1, allows {\sc VBCT} for a limited
range of biases, and is unconditionally secure against general quantum
attacks.  The second protocol, {\sc VBCT}2, is defined for any range
of biases and guarantees cheat-evident security against general
quantum attacks.  The third, {\sc VBCT}3, extends the second by using
a relativistic bit commitment subprotocol, and we conjecture that it
is unconditionally secure against general quantum attacks.

The fourth protocol, {\sc VBCT}4, is classical, and is based on multiple
uses of a classical relativistic bit commitment scheme which is proven
secure against classical attacks.  It can be shown to be
unconditionally secure against classical attacks.  The relevant
relativistic bit commitment scheme is conjectured secure against
quantum attacks, and we conjecture that this is also true of Protocol
{\sc VBCT}4.

Variable bias coin tossing is a simple example of a random one-input
two-sided secure computation.  The most general such computation is
what we have termed a variable bias $n$-faced die roll.  In this case,
there is a finite range of $n$ outputs, with each of Bob's inputs
leading to a different probability distribution over these outputs.
In other words, Bob is effectively allowed to choose one of a fixed
set of biased $n$-faced dice to generate the output, while Alice is
guaranteed that Bob's chosen die is restricted to the agreed set.

The protocols {\sc VBCT}2, {\sc VBCT}3 and {\sc VBCT}4 can easily be generalized to
protocols defining variable bias $n$-faced die rolls.  Thus, to adapt
protocols {\sc VBCT}2 and {\sc VBCT}3 to variable bias die rolling, we require Bob
to choose a series of states from the set $\{\ket{\psi_i}=
\sum_{j=0}^{n-1} \alpha^j_i \ket{jj} \}_{i=1}^r$, where $r$ is the
number of dice in the allowed set and where $(\alpha^j_i)^2$ defines
the probability of outcome $j$ for the $i$th dice (we take
$\{\alpha^j_i\}$ to be real and positive).  The protocols then proceed
similarly to those given above, defining protocols which we conjecture
to be cheat-evidently secure and unconditionally secure respectively.

To adapt Protocol {\sc VBCT}4, we require that the matrix rows contain
appropriate proportions of entries corresponding to the various
possible die roll outcomes.  We conjecture that this protocol is
unconditionally secure.

As we noted earlier, variable bias $n$-sided die rolling is the most
general one-input random two-sided two party single function
computation.  Our conjectures, if proven, would thus imply that all
such computations can be implemented with unconditional security.

\chapter{Secure Two-Party Classical Computation}
\label{Ch3}

\begin{quote}
{\it ``An essential element of freedom is the right to privacy, a
  right that cannot be expected to stand against an unremitting
  technological attack.''} -- Whitfield Diffie
\end{quote}

\section{Introduction}
Two wealthy and powerful businessmen wish to know who is the richest.
They are highly secretive about their bank balances, and do not wish
to disclose more information than that necessarily implied by the
outcome.  Does there exist a sequence of exchanges of (quantum)
information that implements this task?  This is an example of a secure
two-party computation.  In this chapter, we consider a range of such
computations and ask whether they can be implemented with
unconditional security.
\bigskip

A general introduction to secure two-party computation has been given
in Section \ref{sec_comp_intro}.  In this chapter, we continue to
focus on {\it single function computations}.  We will drop the
qualifier {\it single function}---all functions in this chapter can be
assumed to take this form unless otherwise stated.  The main focus is
on two-input functions for which the two-sided deterministic and the
one-sided and two-sided non-deterministic cases will each be discussed
separately (see Section \ref{sec_comp_intro} for definitions).  We
present a cheating attack that renders a large subset of these
functions insecure on the grounds that this attack allows one party to
gain more information about the other's input than is implied by the
outcome of the computation.

\section{Security Definitions In Secure Multi-Party Computation}
\label{sec_defs}
Phrasing security definitions for secure multi-party computations
requires some care.  It is not sufficient (but is necessary), for
example, to demand that the amount of information divulged to a
dishonest party in an implementation of a protocol be less than that
implied by the honest outcome, since the type of information may also
be important.  \comment{As an example, consider a protocol for which
an honest implementation gives Bob any three bits of a four bit string
input by Alice.  If Bob could deviate from the protocol in order to
ascertain the single bit corresponding to the parity of Alice's
string, we would consider the protocol insecure.}  In this section, we
discuss security definitions which sufficiently restrict both the
amount and type of information.  In essence, the idea is that a
protocol is secure if any information one party can get by deviating
from the protocol could have been derived from their output.

It may also be advantageous for one party to deviate from the protocol
in order to influence its outcome, in effect changing the computation
being performed.  A secure protocol must also protect against this
possibility.  Furthermore, we would like a security definition which
guarantees that when the protocol is used as a component of a larger
protocol, it remains secure.  The task of proving security of the
larger protocol can then be reduced to that of its sub-protocols,
together with an argument that the composition of such protocols
performs the desired task.

The universal security framework of Canetti \cite{Canetti}, and the
reactive simulatability framework of Backes, Pfitzmann and Waidner
\cite{PW,BPW} try to capture this idea in a classical context and
have recently been extended and used in quantum scenarios
\cite{Ben-OrMayers,Unruh,CGS}.  Following \cite{HMU}, we define the
following types of security.
\begin{definition} {\bf (Stand-alone security)}
For a proposed protocol, one gives an ideal behaviour.  One then
demands that for every attack against a real execution of the
protocol, there is an equivalent attack against the ideal, in the
following sense.  Suppose we have a black box implementing the ideal.
Then, for any attack on the real protocol, there must exist a {\it
simulator} which, when used in conjunction with the ideal protocol can
generate exactly the same view\footnote{The {\it view} of one party is
their complete set of quantum states and classical values.} as present
after the attack on the real execution.  If some part of the view is
probabilistic, the simulator must be able to generate a view whose
joint distribution with the computation's input is identical to that
of the real protocol.  Furthermore, there must exist a simulator that
can, in conjunction with a black box implementing the ideal, generate
any intermediate states present in the real execution if both parties
are honest.
\end{definition}
\begin{definition}{\bf (Universally composable security)}
The requirements of stand-alone security hold when the protocol is
used in any environment (i.e., as a subprotocol of any larger
protocol).
\end{definition}
The additional requirement for universal composability allows us to
replace the protocol by its ideal in any security analyses, and is
hence highly desirable.  However, such a requirement is rarely
achievable, and often one has to make do with stand-alone security.
The difficulty of satisfying universally composable security
definitions is highlighted in Section \ref{simulator_role}.

In order to prove security under either the stand-alone or universally
composable definitions, one needs to produce a suitable description
for the behaviour of an ideal protocol.  Such descriptions are often
given by invoking a trusted third party
({\sc TTP}\nomenclature[zTTP]{{\sc TTP}}{Trusted Third Party}).  While such
behaviours are called ``ideal'', they may not be ideal in the sense of
being the ultimate demands we might impose upon a protocol.  Such
demands often have to be weakened in order to find a set that are
feasible.  We give two ideals that might be used for computations
involving any number of parties, before specializing to the two-party
case.  We begin with ideals relevant to classical protocols.  Ideal
Behaviour \ref{ib1} represents a true ideal\footnote{The ideals we
give are phrased for general computations, but can easily be
specialized to the single-function case.}.


\begin{ib} \qquad
\label{ib1}
\begin{enumerate}
\item The {\sc TTP} obtains all of the data from all of the parties.
\item It extracts their correct input\footnote{For instance, in a
  computation, where one is supposed to input their bank balance, the
  correct input is the actual balance: an unscrupulous user may lie
  about their input.} from this data and performs the computation.
\item The {\sc TTP} returns to each party their individual outputs.
\end{enumerate}
\end{ib}
It is clear that such a behaviour places unduly strong requirements on
a protocol such that it could never be mimicked by a protocol in the
real world.  A party cannot even lie about their input in such a
model!  Instead, the following (weakened) model has been suggested
\cite{Goldreich} in order to capture some attacks that are impossible
to avoid.

\begin{ib}\qquad
\begin{enumerate}
\item The dishonest parties share their original inputs and decide on
  replaced inputs which they send to the {\sc TTP}.  The honest parties send
  their inputs.
\item The {\sc TTP} uses the inputs to determine the corresponding outputs,
  and sends them to the relevant parties.
\item The dishonest parties may collect their outputs of the {\sc TTP} and
  compute some function dependent on these and their initial inputs.
\end{enumerate}
\label{ib2}
\end{ib}
Let us emphasize two important points.  Firstly, cheating in a
protocol that satisfies the requirements of Ideal Behaviour \ref{ib2}
is only possible by make a replaced input.  The dishonest parties are
not allowed to coerce the {\sc TTP} into generating a different
functionality.  Secondly, in a real implementation of such a protocol,
each party will receive more than just their output.  In a secure
protocol, any additional data received must be of no use.  This is
captured in the security definition by the simulator.

For two-party protocols, it is known that such a behaviour cannot be
realized, and hence Ideal Behaviour \ref{ib2} is often only applied
for the case of honest majority \cite{Goldreich}.  The reason is that
one has to take into account each party's ability to abort within the
ideal behaviour.  In a real protocol, either party may abort, and, in
particular, they may do so at such a point where they have a knowledge
advantage over the other (except in the case of single output
computations, where one of the parties should never gain any
information).  This attack falls outside the scope of Ideal Behaviour
\ref{ib2}.
Goldreich \cite{Goldreich} introduces Ideal Behaviour \ref{ib3} which
tries to allow for this:

\begin{ib}\qquad

\label{ib3}
\begin{enumerate}
\item Each party sends its input to the {\sc TTP}.  A dishonest party
  may replace their input or send no input (abort).
\item {\sc TTP} determines the corresponding outputs and sends the first
  output to the first party.
\item The first party may (if dishonest) tell the {\sc TTP} to abort,
  otherwise it tells it to proceed.
\item If told to proceed, the {\sc TTP} gives the second output to the
  second party.  Otherwise it does nothing.
\end{enumerate}
\end{ib}

It is known that, assuming the existence of enhanced trapdoor
functions, protocols for any secure two-party computation can be
constructed that emulate Ideal Behaviour \ref{ib3} with computational
security \cite{Goldreich}.  When unconditional security is sought,
this ideal behaviour is suitable for a single-round protocol, or one
in which no information is given away until the last step (in which
case early abort is equivalent, in terms of the information gain of
both parties, to not going through with the protocol).  However, this
ideal behaviour neglects the possibility that either party may abort
{\it at any time}.  One could imagine protocols in which information
is built up gradually by each party in each round of communication, in
such a way that one party can only have a small amount more than the
other at any given time \cite{Lo}.  One might then invoke an instance
of Ideal Behaviour \ref{ib3} for each round of the protocol.  This
seems unduly cumbersome to build into a definition of a secure
computation.  An ideal whereby abort is allowed at any step is
desirable.  

We introduce the following ideal behaviour in order to capture this
(specializing now to the two-party case):

\newpage
\begin{ib}\qquad

\label{ib4}
\begin{enumerate}
\item Each party sends its input to the {\sc TTP}, along with a
  number, $a$, representing an additional function to compute at the
  end (if desired). (If both parties submit numbers, the lowest is
  taken.  Additionally, Alice can only submit even numbers, and Bob
  odd ones.)
\item The {\sc TTP} performs the correct function based on the inputs
  supplied to generate the correct outputs, $k_A$ and $k_B$.
\item The {\sc TTP} applies a further function to each of the outputs before
  sending $(a, f_a(k_A))$ to Alice and $(a, g_a(k_B))$ to Bob.
\end{enumerate}
\end{ib}

The additional function to be computed represents the output that
would be generated by a protocol which is aborted after step $a$.  The
behaviour has been phrased above in order to emphasize that the output
generated by early aborting gives no extra information and no other
type of information than that generated by following the protocol
honestly, in the sense that the correct final output can be used to
simulate any of the intermediate ones.

When extending such definitions to quantum protocols, there are a
number of additional considerations.  The ideal behaviour in a quantum
protocol may in many cases be weaker than its classical counterpart.
This comes about because:
\begin{enumerate}
\item A real protocol cannot mimic a {\sc TTP} that does measurements,
  since in the real implementation of a protocol, it is always
  possible to keep all measurements at the quantum level until the
  end. \footnote{Even though honest parties can be trusted to make
  measurements as the protocol progresses, it is equivalent when
  performing a security analysis to assume that they kept their
  measurements quantum until the end of the protocol, and hence we can
  restrict to protocols for which this is the case.}
\item A real protocol cannot perform classical certification of the
  inputs (i.e., cannot abort when a superposition is input instead of
  a single member of the computational basis) \cite{Kent_certif}.
\end{enumerate}
A classical protocol is able to circumvent such issues by implicitly
making the (technological) assumption that all parties can only
manipulate classical data.

Consider the following ideal behaviour for a quantum protocol
implementing a computation:

\begin{ib}\qquad

\label{ib5}
\begin{enumerate}
\item Both parties send their inputs to the {\sc TTP}.  If dishonest, the
  inputs may be quantum (i.e., superpositions) rather than members of
  an orthogonal basis.
\item The {\sc TTP} does a unitary operation on the inputs.  For
  example, in a two-sided deterministic computation the unitary might
  be defined by
  $U_f\ket{i}\ket{j}\ket{0}\ket{0}=\ket{i}\ket{j}\ket{k}\ket{k}$,
  where $f$ indexes the function being computed, $i$ is Alice's input,
  $j$ is Bob's input and $k=f(i,j)$ is the corresponding
  output\footnote{There are possible variants of the chosen unitary
  operation (see Section \ref{comp_model}).}.
\item The {\sc TTP} returns the first and third Hilbert spaces to Alice, and
  the second and fourth ones to Bob.
\end{enumerate}
\end{ib}
This is in fact stronger than we can achieve because it does not allow
for early abort.  Following arguments we presented in the classical
case, we should modify the steps to allow Alice to choose whether Bob
gets his output, and make further modifications to account for early
aborts, in the spirit of Ideal Behaviour \ref{ib4}.

Under Ideal Behaviour \ref{ib5}, cheating is restricted to making a
dishonest input, and to making an alternative measurement on the
output.  We will show that such cheating is enough to break any
reasonable requirements one might make for a large class of secure
classical computations.  Points 1 and 2 above ensure that one cannot
weaken the ideal behaviour such that this attack fails.  Hence quantum
protocols for these classes of secure classical computation do not
exist.

One special case is that of a one-input computation.  In the two-party
case, such a computation must be both random and two-sided (otherwise
it is trivial).  In Chapter \ref{Ch2}, we conjectured that such
computations (variable bias $n$-faced die rolls) are possible with
unconditional security.  Our definitions there were such that (if we
ignore the supplementary tests, which asymptotically were not
necessary for security) there are no useful ways to abort, and so the
behaviour realized is that of Ideal Behaviour \ref{ib2}.  (In Ideal
Behaviour \ref{ib2}, either party can force the outcome to be abort by
refusing to make an input to the {\sc TTP} in the first place.)

Our protocols for variable bias $n$-faced die rolling were
relativistic.  Exploiting relativistic signalling constraints does not
affect the type of behaviour realizable, in either quantum or
classical protocols.  Rather, using a relativistic protocol affects
the range of computations possible within each model.  For instance,
we cannot mimic a {\sc TTP} that performs coin tossing in a non-relativistic
world, but can in a relativistic one.  This is distinct from the types
of behaviour in which we embed the {\sc TTP}.

\subsection{The R\^ole Of The Simulator}
\label{simulator_role}
Let us demonstrate the importance of the simulator for universally
composable security definitions.  For this we will use the task of
extending coin tosses \cite{HMU}.  In such a task, Alice and Bob are
given access to a finite source of coin tosses, guaranteed to be
independent and uniformly distributed.  Their goal is to exchange
information and use this source in order to generate a shared random
string longer than that which is available from the source alone.

The protocol takes place in a classical environment in which Alice and
Bob are given access to the device supplying coin tosses.  This device
operates according to the following ideal.
\begin{idf}\qquad
\label{if1}
\begin{enumerate}
\item The {\sc TTP} waits until it has been initialized by both parties,
  after which, it generates a random string, $R$.
\item If Alice is dishonest, she can choose when the {\sc TTP} gives $R$
  to Bob, otherwise, the {\sc TTP} does so immediately.
\item Similarly, if Bob is dishonest, he can choose when the {\sc TTP} gives
  $R$ to Alice, otherwise, the {\sc TTP} does so immediately.
\end{enumerate}
\end{idf}
The following classical non-relativistic protocol is employed to
generate a shared random string longer than $R$, using a single call
of the above ideal at the appropriate time.  
\begin{protocol}\qquad
\label{ext_CT}
\begin{enumerate}
\item Alice sends a random string, $a$, to Bob.
\item \label{Bob_sends} Bob receives Alice's string and sends a random
string, $b$, to Alice.  Strings $a$ and $b$ have the same length.
\item Alice and Bob supply initiation signals to a device (device 1)
  that supplies perfect coin tosses.
\item Device 1 supplies $R$ to Alice and/or Bob in accordance with
  Ideal Functionality \ref{if1}, i.e., depending on whether either
  party is dishonest.
\item Alice and Bob use $R$ to perform privacy amplification on the
  concatenated string, $(a,b)$.  This generates a final string, $s$,
  that is (virtually) uniform and independent of $R$, $a$ and $b$.
  The final output of the protocol is the concatenation, $(R,s)$.
\end{enumerate}
\end{protocol}
Security of this protocol is discussed in \cite{HMU}.  It relies on
the fact that $R$ is not known to Alice and Bob until after they have
exchanged strings, and then follows from results on privacy
amplification (see Section \ref{priv_amp}).

\begin{figure}
\begin{center}
\includegraphics[width=\textwidth]{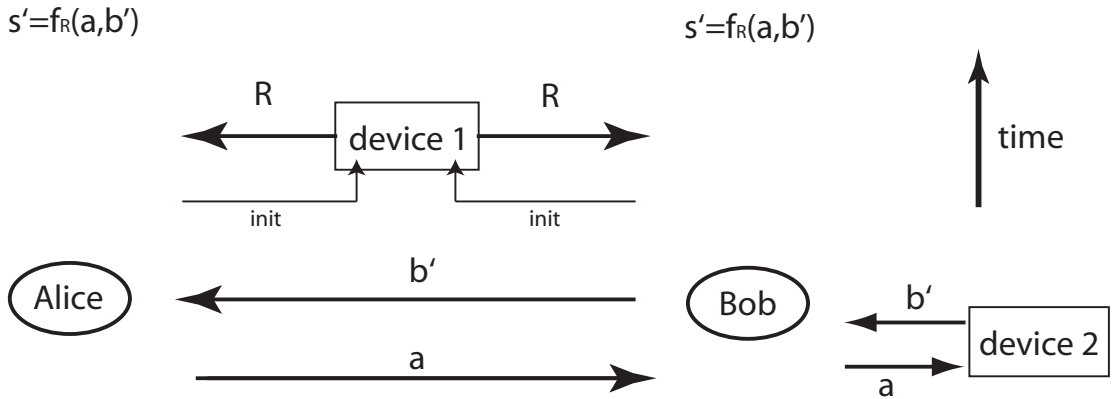}
\caption{The sequence of exchanges between Alice and Bob in the
  protocol for extending coin tosses (Protocol \ref{ext_CT}), where
  Bob interacts with a second device to choose his string.  Device 1
  is the supplier of perfect coin tosses, in the form of string $R$.
  In the original form of the protocol, device 2 is not used, and Bob
  sends a random string, $b$, of his own choosing to Alice.}
\label{fig:simulator2}
\end{center}
\end{figure}

We will show that this protocol is not sufficient to realize a
modification of Ideal Functionality \ref{if1}, where $R$ is replaced
by $(R,s)$.  This follows because there exists an interaction with a
system in the environment that Bob can follow in the real protocol,
but cannot implement in the ideal.

Consider an instance of the real protocol, and suppose Bob has access
to an additional device (device 2) with which he interacts only once.
He inputs $a$ into this device and it returns $b'$ to him, with $b'$
being a function of $a$ which he does not know.  Bob sends $b'$ to
Alice in place of $b$, after which the protocol proceeds with both
parties behaving honestly.  If he follows this strategy, the final
string $s'$ is distributed uniformly, regardless of the function
applied by the extra device.  Given an implementation of the ideal
protocol, which outputs $s_I$, it is easy for Bob to simulate $a$ and
$b$.  However, if Bob simulates $a$, and then inputs this into device
2, the string $b'$ he is returned will not necessarily be compatible
with the string returned by the protocol (i.e.\ $(R,f_R(a,b'))$ may
not equal $s_I$\footnote{More generally, the joint distributions over
all variables are different in the two cases.}).  It is impossible to
correctly simulate $b'$ and hence the protocol does not satisfy
universally composable security requirements.  The entire procedure is
shown in Figure \ref{fig:simulator2}.

While a cheating strategy of this kind is unlikely to present a
problem in any future application, it is possible that more
significant attacks exist.  The universally composable security
definition relieves us of such worries---if such a security definition
is satisfied, then one can replace all instances of the protocol with
the ideal without affecting security.

Unfortunately, it is rare that universally composable security can be
realised.  The type of attack given in this section is detrimental in
many contexts.  Protocols in which one party must respond to
information received by the other are particularly vulnerable in this
way.  One exception is the case of a classical protocol to give a
zero-knowledge proof for the graph non-isomorphism problem \cite{GMW},
which we discuss in Appendix \ref{AppD}.  The reason that this
protocol escapes the aforementioned attack is that one party (the
prover) always has the freedom to deterministically choose the output
of the protocol.

Relativistic protocols can provide a way to avoid this type of attack.
In a non-relativistic situation in which one party can pass
information they receive through an external device before responding,
it may be possible to instead use a relativistic protocol in which the
response is supplied by a distant agent of that party\footnote{This
will only work for protocols in which the response is supposed to be
independent of the received information.}.  As an example, suppose we
demanded that Step \ref{Bob_sends} of Protocol \ref{ext_CT} occurs at
spacelike separation to the point where Bob receives $a$ (which can be
done by having $B_2$ send it to $A_2$ in a relativistic protocol).
The attack involving the second device cannot then be implemented in
the real world and we do not need to provide a simulator for it in the
ideal case.

\comment{Nevertheless, there are many cases in which one is forced to make do
with stand-alone security definitions.  They are equivalent to
universally composable behaviour under the assumption that no
computations take place in parallel with the protocol being performed.
That is, while one protocol is being implemented, no party can send
any data to another protocol.}

\subsection{Computational Model}
\label{comp_model}
We will use a black box model for secure computation.  A black box is
a hypothetical device that satisfies a certain set of ideal
functionality requirements.  It features an authentication system
(e.g., an unalterable label) so that each party can be sure of the
function it computes.  We will give a security requirement, and show
that even if black boxes satisfying Ideal Behaviour \ref{ib5} were to
exist, this requirement cannot, in general, be satisfied.

We now comment on the possible forms of unitary operation that could
implement a particular computation.  In a two-sided, non-deterministic
computation, one seeks the functionality given by $U_f$, defined by
\begin{equation}
U_f\ket{i}_A\ket{j}_B\ket{0}\ket{0}=\ket{i}_A\ket{j}_B\sum_k\alpha^k_{i,j}\ket{kk}_{AB}.
\end{equation}

In practice, a computation might generate additional states, and one
should consider instead $U'_f$ defined by
\begin{equation}
U'_f\ket{i}_A\ket{j}_B\ket{0}\ket{0}\ket{0}=\ket{i}_A\ket{j}_B\sum_k\alpha^k_{i,j}\ket{kk}_{AB}\ket{\psi_{i,j}^k}_{AB},
\end{equation}
where the final Hilbert space corresponds to an ancillary system the
black box uses for the computation (and has arbitrary dimension).  In
the protocol mimicking such a box, this final state must be
distributed between Alice and Bob in some way, such that the part that
goes to Bob contains no information on Alice's input, and vice versa.

If this second unitary operation replaces that given in Ideal
Behaviour \ref{ib5}, then again each party has two ways of
cheating---inputing a superposition of honest states or using a
different measurement on the output.
We now show that under such attacks, insecurity of functions under
$U_f$ implies insecurity under $U'_f$, and hence we consider only the
former.

Consider the case where Alice makes a superposed input,
$\sum_ia_i\ket{i}$, rather than a single member of the computational
basis.  Then, at the end of the protocol, her reduced density matrix
takes either the form 
\begin{equation}
\label{sigj}
\sigma_j=\sum_{i,i',k}a_ia_{i'}^*\alpha_{i,j}^k(\alpha_{i',j}^k)^*\ketbra{i}{i'}\otimes\ketbra{k}{k},
\end{equation}
or
\begin{equation}
\sigma'_j=\sum_{i,i',k}a_ia_{i'}^*\alpha_{i,j}^k(\alpha_{i',j}^k)^*\ketbra{i}{i'}\otimes\ketbra{k}{k}\otimes{\text{tr}}_B\ketbra{\psi_{i,j}^k}{\psi_{i',j}^k},
\end{equation}
the first case applying to $U_f$, and the second to $U'_f$.

Alice is then to make a measurement on her state in order to
distinguish between the different possible inputs Bob could have made,
as best she could.  We will show that there exists a trace-preserving
quantum operation that Alice can use to convert $\sigma'_j$ to
$\sigma_j$ for all $j$.  Therefore Alice's ability to distinguish
between $\{\sigma'_j\}_j$ is at least as good as her ability to
distinguish between $\{\sigma_j\}_j$.

In order that the protocol functions correctly when both Alice and Bob
are honest, we require
$\text{tr}_B\ketbra{\psi_{i,j}^k}{\psi_{i,j}^k}\equiv\rho^{i,k}$ to be
independent of $j$ (otherwise Alice can gain more information on Bob's
input than that implied by $k$ by a suitable measurement on her part
of this state).  By expressing $\rho^{i,k}$ in its diagonal basis,
$\rho^{i,k}=\sum_m\lambda^{i,k}_mU_A^{i,k}\ketbra{m}{m}_A(U_A^{i,k})^{\dagger}$,
we have
\begin{equation}
\ket{\psi_{i,j}^k}=\sum_m\sqrt{\lambda_m^{i,k}}U_A^{i,k}\ket{m}_A\otimes U_B^{i,j,k}\ket{m}_B,
\end{equation}
where $\{\ket{m}_A\}_m$ form an orthogonal basis set on Alice's system
and likewise $\{\ket{m}_B\}_m$ is an orthogonal basis for Bob's system.
Bob then holds
\begin{equation}
\text{tr}_A\ketbra{\psi_{i,j}^k}{\psi_{i,j}^k}=\sum_m\lambda_m^{i,k}U_B^{i,j,k}\ketbra{m}{m}_B(U_B^{i,j,k})^{\dagger}.
\end{equation}
This must be independent of $i$, hence $\lambda_m^{i,k}$ and
$U_B^{i,j,k}$ must be independent of $i$.  Thus
\begin{equation}
\ket{\psi_{i,j}^k}=\sum_m\sqrt{\lambda_m^k}(U_A^{i,k}\otimes
U_B^{j,k})\ket{m}_A\ket{m}_B.
\end{equation}
It hence follows that there is a unitary on Alice's system converting
$\ket{\psi_{i_1,j}^k}$ to $\ket{\psi_{i_2,j}^k}$ for all $i_1$, $i_2$,
and that furthermore, this unitary is independent of $j$.
Likewise, there is a unitary on Bob's system converting
$\ket{\psi_{i,j_1}^k}$ to $\ket{\psi_{i,j_2}^k}$ for all $j_1$, $j_2$,
with this unitary being independent of $i$.

Returning now to the case where Alice makes a superposed input.  The
final state of the entire system can be written
\begin{equation}
\sum_{i,k}a_i\alpha_{i,j}^k\ket{i}_A\ket{j}_B\ket{k}_A\ket{k}_B(U_A^{i,k}\ket{m}_A)(U_B^{j,k}\ket{m}_B).
\end{equation}
Alice can then apply the unitary
\begin{equation}
V=\sum_{i,k}\ketbra{i}{i}_A\otimes\openone_B\otimes\ketbra{k}{k}_A\otimes\openone_B\otimes
(U_A^{i,k})^{\dagger}\otimes\openone_B
\end{equation}
to her systems leaving the state as
\begin{equation}
\sum_{i,k}a_i\alpha_{i,j}^k\ket{i}_A\ket{j}_B\ket{k}_A\ket{k}_B\sum_m\sqrt{\lambda_m^k}\ket{m}_A(U_B^{j,k}\ket{m}_B).
\end{equation}
Alice is thus in possession of density matrix
\begin{equation}
\sum_{i,i',k}a_ia_{i'}^*\alpha_{i,j}^k(\alpha_{i',j}^k)^*\ketbra{i}{i'}\otimes\ketbra{k}{k}\otimes\rho_A^k,
\end{equation}
where $\rho_A^k=\sum_m\lambda_m^k\ketbra{m}{m}_A$.  Hence, on tracing
out the final system, we are left with $\sigma_j$ as defined by
(\ref{sigj}).

We have hence shown that there is a ($j$-independent) trace-preserving
quantum operation Alice can perform which converts $\sigma'_j$ to
$\sigma_j$ for all $j$.  Hence Alice's ability to distinguish between
Bob's inputs after computations of the type $U'_f$ is at least as good
as her ability to distinguish Bob's inputs after computations of the
type $U_f$, and so, under the type of attack we consider, insecurity of
computations specified by $U_f$ implies insecurity of those specified
by $U'_f$.  We will therefore consider only type $U_f$ in our
analysis.  An analogous argument follows for the one-sided case, and
likewise for the deterministic cases (which are special cases of the
non-deterministic ones).

\bigskip
In this chapter we will show that the following security condition can
be broken for a large class of computation.
\begin{SC}
Consider the case where Bob is honest.  A secure computation is one
for which there is no input, together with a measurement on the
corresponding output that gives Alice a better probability of guessing
Bob's input than she would have gained by following the protocol
honestly and making her most informative input.  This condition must
hold for all forms of prior information Alice holds on Bob's input.
\end{SC}

\section{Deterministic Functions}
\label{det_fns}
We first focus on the deterministic case\footnote{We refer the reader
to Section \ref{sec_comp_intro} for descriptions of the various types
of function we consider.}.  Lo showed that two-input deterministic
one-sided computations are impossible to compute securely \cite{Lo},
hence only two-sided deterministic functions remain\footnote{Lo did
not consider relativistic cryptography, but his results apply to this
case as well (see the discussion in Section \ref{sec_comp_intro}).}.
Suppose now that the outcome of such a protocol leads to some
real-world consequence.  In the dating problem \cite{GottesmanLo}, for
example, one requires a secure computation of $k=i\times j$, where
$i,j\in\{0,1\}$.  If the computation returns $k=1$, then the protocol
dictates that Alice and Bob go on a date.  This additional real-world
consequence is impossible to enforce, although naturally, both Alice
and Bob have some incentive not to stand the other up, since this
results in a loss of the other's trust.  A cost function could be
introduced to quantify this.  Because suitable cost assignments must
be assessed case by case, it is difficult to develop general results.
To eliminate such an issue, we restrict to the case where the sole
purpose of the computation is to learn something about the input of
the other party.  No subsequent action of either party based on this
information will be specified.

We say that a function is {\it potentially concealing} if there is no
input by Alice which will reveal Bob's input with certainty, and
vice-versa.  If the aim of the computation is only to learn something
about the input of the other party, and if Bob's data is truly
private, he will not enter a secure computation with Alice if she can
learn his input with certainty.  We hence only consider potentially
concealing functions in what follows.  In addition, we will ignore
{\it degenerate} functions in which two different inputs are
indistinguishable in terms of the outcomes they afford.  If the sole
purpose of the computation is to learn something about the other
party's input, then, rather than compute a degenerate function, Alice
and Bob could instead compute the simpler function formed by combining
the degenerate inputs of the original.

An alternative way of thinking about such functions is that they
correspond to those in which there is cost for ignoring the real world
consequence implied by the computation.  At the other extreme, one
could invoke the presence of an enforcer who would compel each party
to go ahead with the computation's specified action.  This would have
no effect on security for a given function (a cheating attack that
works without an enforcer also works with one) but introduces a larger
set of functions that one might wish to compute.  There exist
functions within this larger set for which the attack we present does
not work.

We specify functions by giving a matrix of outcomes.  For convenience,
the outputs of the function are labelled with consecutive integers
starting with 0.  We consider functions that satisfy the following
conditions:
\begin{enumerate}
\item \label{cond1} (Potentially concealing requirement) Each row and
  each column must contain at least two elements that are the same.
\item \label{cond2} (Non degeneracy requirement) No two rows or
columns should be the same.
\end{enumerate}
For instance, if $i,j\in\{0,1,2\}$ (which we term a $3\times 3$
function), the function $f(i,j)=1-\delta_{ij}$ is represented by
\begin{center}
\begin{tabular}{cc|ccc|}
\multirow{2}{*}{$f(i,j)$}&&\multicolumn{3}{|c|}{$i$}\\
&\;&\;0\;&\;1\;&\;2\;\\
\hline
\multirow{3}{*}{$j$}&0\;&$0$&$1$&$1$\\
&1\;&$1$&$0$&$1$\\
&2\;&$1$&$1$&$0$\\
\hline
\end{tabular}
.
\end{center}
\smallskip This function is potentially concealing, and non-degenerate.

We consider the case of $3\times 3$ functions.  We first give a
non-constructive proof that Alice can always cheat, and then an
explicit cheating strategy.

Let us assume that we have a black box that can implement the
protocol, i.e., one that performs the following operation:
\begin{eqnarray}
U_f\ket{i}_A\ket{j}_B\ket{0}\ket{0}=\ket{i}_A\ket{j}_B\ket{f(i,j)}_A\ket{f(i,j)}_B.
\end{eqnarray}
The states $\{\ket{i}_A\}$ are mutually orthogonal, as are the members
of the sets $\{\ket{j}_B\}$, $\{\ket{f(i,j)}_A\}$ and
$\{\ket{f(i,j)}_B\}$.  This ensures that Alice and Bob always obtain
the correct output if both have been honest.  The existence of such a
black box would allow Alice to cheat in the following way.  She can
first input a superposition, $\sum_{i=0}^2a_i\ket{i}_A$ in place
of $\ket{i}_A$.  Her output from the box is one of
$\rho_0,\rho_1,\rho_2$, the subscript corresponding to Bob's input,
$j$, where (using the shorthand ${\rm tr}_B(\ket{\Psi})\equiv{\rm
tr}_B(\ketbra{\Psi}{\Psi})$)
\begin{eqnarray}
\rho_j\equiv{\rm tr}_B\left(U_f\sum_{i=0}^2a_i\ket{i}_A\ket{j}_B\ket{0}_A\ket{0}_B\right).
\end{eqnarray}
Alice can then attempt to distinguish between these using any
measurement of her choice.

The main result of this section is the following theorem.
\begin{theorem}
Consider the computation of a $3\times 3$ deterministic function
satisfying conditions \ref{cond1} and \ref{cond2}.  For each function
of this type, there exists a set of co-efficients, $\{a_i\}$,
such that when Alice inputs $\sum_{i=0}^2a_i\ket{i}_A$ into the
protocol, there exists a measurement that gives her a better
probability of distinguishing the three possible ($j$ dependent)
output states than that given by her best honest strategy.
\label{thm1}
\end{theorem}

\begin{table}
\begin{center}
\begin{tabular}{cc|ccc|}
\multirow{2}{*}{$f(i,j)$}&&\multicolumn{3}{|c|}{$i$}\\
&\;&\;0\;&\;1\;&\;2\;\\
\hline
\multirow{3}{*}{$j$}&0\;&\;$0$\;&\;$a$\;&\;$.$\;\\
&1\;&$0$&$b$&$.$\\
&2\;&$1$&$b$&$.$\\
\hline
\end{tabular}
\caption{This function can be taken as the most general $3\times 3$
  function satisfying conditions 1 and 2, where $a\neq b$, and $a=0$
  or $b=0$ or $b=1$.  The dots represent unspecified (and not
  necessarily identical) entries consistent with the conditions.}
\label{tab1}
\end{center}
\end{table}
\begin{proof}
We will rely on the following lemma.
\begin{lemma} 
All $3\times 3$ functions satisfying conditions 1 and 2 can be put in
the form of the function in Table \ref{tab1}.
\end{lemma}
\begin{proof}
The essential properties of any function are unchanged under
permutations of rows or columns (which correspond to relabelling of
inputs), and under relabelling of outputs.  In order that the function
is potentially concealing, there can be at most one column whose
elements are identical.  By relabelling the columns if necessary, we
can ensure that this corresponds to $i=2$.  Relabelling the outputs
and rows, if necessary, the column corresponding to $i=0$ has entries
$(f(0,0),f(0,1),f(0,2))=(0,0,1)$.  The column corresponding to $i=1$
then must have entries $(a,a,b)$ or $(a,b,b)$, with $a\neq b$.  In the
case $(a,a,b)$, the $i=2$ column must have the form $(c,d,d)$, for
$c\neq d$, in which case we can permute the $i=1$ and $i=2$ columns to
recover the form $(a,b,b)$ for the $i=1$ column.  Relabellings always
put such cases into forms with $a=0$ or $b=0$ or $b=1$.
\end{proof}

Suppose Alice inputs $\frac{1}{\sqrt{2}}\left(\ket{0}+\ket{1}\right)$
into a function of the form given in Table \ref{tab1}.  After tracing
out Bob's systems, Alice holds one of
\begin{eqnarray}
\rho_0&=&\frac{1}{2}\left(\ketbra{00}{00}+\delta_{a,0}\left(\ketbra{00}{10}+\ketbra{10}{00}\right)+\ketbra{1a}{1a}\right)\\
\rho_1&=&\frac{1}{2}\left(\ketbra{00}{00}+\delta_{b,0}\left(\ketbra{00}{10}+\ketbra{10}{00}\right)+\ketbra{1b}{1b}\right)\\
\rho_2&=&\frac{1}{2}\left(\ketbra{01}{01}+\delta_{b,1}\left(\ketbra{01}{11}+\ketbra{11}{01}\right)+\ketbra{1b}{1b}\right).
\end{eqnarray}
Measurement using the set $\{E_{i,k}=\ketbra{ik}{ik}\}$ in effect
reverts to an honest strategy.  The probability of correctly guessing
Bob's input using this set is the same as that for Alice's best honest
strategy.  These operators can be combined to form just three
operators, $\{E_{j'}\}$ such that a result corresponding to $E_{j'}$
means that Alice's best guess of Bob's input is $j'$.  Then
\begin{eqnarray}
\nonumber E_0&=&\alpha_1\ketbra{00}{00}+\delta_{a,0}\ketbra{10}{10}+\delta_{a,1}\ketbra{11}{11}+\delta_{a,2}\ketbra{12}{12}+\\\label{E0}&&\delta_{a,3}\ketbra{13}{13}\\
E_1&=&(1-\alpha_1)\ketbra{00}{00}+\alpha_2\delta_{b,0}\ketbra{10}{10}+\alpha_3\delta_{b,1}\ketbra{11}{11}+\nonumber\\\label{E1}&&\alpha_4\delta_{b,2}\ketbra{12}{12}+\alpha_5\delta_{b,3}\ketbra{13}{13}\\
\label{E2}E_2&=&\openone-E_0-E_1,
\end{eqnarray}
where the $\{\alpha_l\}$ are arbitrary parameters, $0\leq\alpha_l\leq
1$, and do not affect the success probability.  We will show that such
a measurement is not optimal to distinguish between the corresponding
$\{\rho_j\}$.  This follows from Theorem \ref{optimal_cond}.

Equations (\ref{con1}) and (\ref{con2}) imply respectively,
\begin{eqnarray}
\left(\alpha_1=0\quad{\rm or }\quad\alpha_2=0\quad{\rm or
}\quad b\neq 0\right)\quad{\rm and}\quad \left(\alpha_1=1\quad{\rm or }\quad
a\neq 0\right)\quad{\rm and}\nonumber\\\label{first} \quad \left(\alpha_1=1\quad{\rm or
}\quad\alpha_2=1\quad{\rm or }\quad b\neq 0\right)\quad{\rm and}\quad
\left(\alpha_3=0\quad{\rm or }\quad b\neq 1\right),
\end{eqnarray}
and,
\begin{eqnarray}
\nonumber&\left(\alpha_1=0\quad{\rm or }\quad(b\neq
0\quad{\rm and}\quad a\neq 0)\right)\quad{\rm and}\quad
\left(b=1\quad{\rm or }\quad\alpha_3\geq\frac{1}{4}\right)\quad{\rm
and}\\\nonumber&\quad 
\left(a=1\quad{\rm or }\quad\alpha_3=1\quad{\rm or }\quad b\neq
1\right)\quad{\rm and}\quad\left(b=0\quad{\rm or }\quad\alpha_2(1-\alpha_1)\geq\frac{1}{4}\right)\quad{\rm and}\\
&\left(\alpha_1=1\quad{\rm or }\quad
b\neq 0\quad{\rm or }\quad\alpha_2=0\right).
\end{eqnarray}
In addition, because the function is in the form given in Table
\ref{tab1}, we also have
\begin{eqnarray}
\label{last}\left(a=0\quad{\rm or }\quad b=0\quad{\rm or }\quad
b=1\right)\quad{\rm and}\quad a\neq b.
\end{eqnarray}
The system of equations (\ref{first}--\ref{last}) cannot be satisfied
for any values of $a,b,\{\alpha_k\}$.  Hence, the measurement
operators (\ref{E0}--\ref{E2}) are not optimal for discriminating
between Bob's inputs, so Alice always has a cheating strategy.
\end{proof}

Our proof of Theorem \ref{thm1} is non-constructive---we have shown
that cheating is possible, but not explicitly how it can be done.
Except in special cases (e.g., where the states $\{\rho_j\}$ are
symmetric), no procedure for finding the optimal {\sc POVM} to
distinguish between states is known \cite{Chefles,JRF}.  Nevertheless,
we have found a construction based on the square root measurement
\cite{HJW,HW} that, while not being optimal, gives a higher
probability of successfully guessing Bob's input than any honest
strategy.

The strategy applies to the states, $\sigma_j$, formed when Alice
inputs the state
$\frac{1}{\sqrt{3}}\left(\ket{0}+\ket{1}+\ket{2}\right)$ into the
computation.  The set of operators are those corresponding to the
square root measurement, defined by
\begin{eqnarray}
E_{j'}=\left(\sum_j\sigma_j\right)^{-\frac{1}{2}}\sigma_{j'}\left(\sum_j\sigma_j\right)^{-\frac{1}{2}}.
\end{eqnarray}
One can verify, case by case, that this strategy affords Alice a
better guessing probability over Bob's input than any honest one for
all functions of the form of Table \ref{tab1}.  The Mathematica script
which we have used to confirm this is available on the world wide web
\cite{mathematica_script}.

\section{Non-Deterministic Functions}
\subsection{Two-Sided Case}
\begin{table}
\begin{center}
\begin{tabular}{cc|cc|}
\multirow{2}{*}{$p(0|i,j)$}&&\multicolumn{2}{|c|}{$i$}\\
&&\;0\;&\;1\;\\
\hline
\multirow{2}{*}{$j$}&0\;&\;$p_{00}$\;&\;$p_{10}$\;\\
&1\;&\;$p_{01}$\;&\;$p_{11}$\;\\
\hline
\end{tabular}
\caption{The entries in the table give the probabilities of output 0
  given inputs $i,j$.  For example, if both parties input 0, then the
  output of the function is 0 with probability $p_{00}$, and 1 with
  probability $1-p_{00}$.}
\label{function1}
\end{center}
\end{table}

Initially, we specialize to the case $i,j,k\in\{0,1\}$.  We specify
such functions via a matrix of probabilities whose meaning is given in
Table \ref{function1}.  For the two-sided case, the relevant black box
implements the unitary, $U$, given by
\begin{eqnarray}
\label{U}U\ket{i}_A\ket{j}_B\ket{0}\ket{0}=\ket{i}_A\ket{j}_B\left(\sqrt{p_{ij}}\ket{00}_{AB}+\sqrt{1-p_{ij}}\ket{11}_{AB}\right).
\end{eqnarray}
Suppose that Alice has prior information about Bob's input such that,
from her perspective, he will input $0$ with probability $\eta_0$, and
$1$ with probability $\eta_1=1-\eta_0$.  The probability of correctly
guessing Bob's input using the best honest strategy is
\begin{equation}
\label{ph}
p_h=\max_i\left(\max_j\left(p_{ij}\eta_j\right)+\max_j\left((1-p_{ij})\eta_j\right)\right).
\end{equation}
Denote Alice's final state by $\rho_j$, where $j$ is Bob's input.  The
optimal strategy to distinguish $\rho_0$ and $\rho_1$ is successful
with probability
\begin{eqnarray}
\label{pc}
p_c=\frac{1}{2}\left(1+{\rm tr}\left|\eta_0\rho_0-\eta_1\rho_1\right|\right)
\end{eqnarray}
(cf.\ Theorem \ref{optimal_cond}).
\begin{theorem}
Let Alice input $\frac{1}{\sqrt{2}}\left(\ket{0}+\ket{1}\right)$ and
Bob input $j$ into the computation given in (\ref{U}).  Let Alice
implement the optimal measurement to distinguish the corresponding
$\rho_0$ and $\rho_1$ and call the probability of a correct guess
using this measurement $p_c$. Then, for all
$\{p_{00},p_{01},p_{10},p_{11}\}$, there exists a value of $\eta_0$ such
that $p_c>p_h$, unless, 
\begin{enumerate}
\item $p_{00}=p_{10}$ and $p_{01}=p_{11}$, or 
\item $p_{00}=p_{01}$ and $p_{10}=p_{11}$. 
\end{enumerate}
\label{thm3}
\end{theorem}
The two exceptional cases correspond to functions for which only one
party can make a meaningful input.  We hence conclude that all
genuinely two-input functions of this type are impossible to compute
securely.
\begin{proof}
Take $\eta_0=1-\epsilon$. For sufficiently small $\epsilon>0$, (\ref{ph})
implies $p_h=\eta_0$.
We then seek $p_c$.  The eigenvalues of
$\eta_0\rho_0-\eta_1\rho_1$ are
\begin{eqnarray}
\lambda_{\pm}&=&\frac{1}{4}\left(a(\{p_{i,j}\})\pm\sqrt{a^2(\{p_{i,j}\})+b(\{p_{i,j}\})}\right)\\
\mu_{\pm}&=&\frac{1}{4}\left(a(\{\overline{p_{i,j}}\})\pm\sqrt{a^2(\{\overline{p_{i,j}}\})+b(\{\overline{p_{i,j}}\})}\right),
\end{eqnarray}
where $a(\{p_{i,j}\})=(p_{00}+p_{10})\eta_0-(p_{01}+p_{11})\eta_1$,
$b(\{p_{i,j}\})=4(\sqrt{p_{01}p_{10}}-\sqrt{p_{00}p_{11}})^2\eta_0\eta_1$,
and $\overline{p_{ij}}\equiv 1-p_{ij}$.

For $\epsilon$ sufficiently small, we have $a\gg b>0$.  Using
$\sqrt{1+x}\leq 1+\frac{x}{2}$, we find,
$\lambda_+\geq\frac{1}{4}(2a(\{p_{i,j}\})+\frac{b(\{p_{i,j}\})}{2a(\{p_{i,j}\})}),$ 
$\lambda_-\leq -\frac{b(\{p_{i,j}\})}{8a(\{p_{i,j}\})},$ 
$\mu_+\geq\frac{1}{4}(2a(\{\overline{p_{i,j}}\})+\frac{b(\{\overline{p_{i,j}}\})}{2a(\{\overline{p_{i,j}}\})}),$
and $\mu_-\leq
-\frac{b(\{\overline{p_{i,j}}\})}{8a(\{\overline{p_{i,j}}\})}$, with
equality iff $b(\{p_{i,j}\})=0$ and $b(\{\overline{p_{i,j}}\})=0$.  We
hence have
$\frac{1}{2}\left(1+\text{tr}|\eta_0\rho_0-\eta_1\rho_1|\right)\geq \eta_0$
and so $p_c\geq p_h$, with equality iff $p_{00}=p_{10}$ and
$p_{01}=p_{11}$, or $p_{00}=p_{01}$ and $p_{10}=p_{11}$.
\end{proof}
The explicit construction for the optimal cheating measurement is
given in Appendix \ref{AppA}.

\subsection{One-Sided Case}
For one-sided computations of non-deterministic functions, Alice can
cheat without inputing a superposed state.  In this case, the black
box performs the unitary
\begin{eqnarray}
U\ket{i}_A\ket{j}_B\ket{0}=\ket{i}_A\ket{j}_B\left(\sqrt{p_{ij}}\ket{0}_A+\sqrt{1-p_{ij}}\ket{1}_A\right),
\end{eqnarray}
where the last qubit goes to Alice at the end of the protocol.  The
following theorem shows that such computations cannot be securely
implemented.
\begin{theorem}
Having made an honest input to the black box above, Alice's optimum
procedure to correctly guess Bob's input is not given by a
measurement in the $\{\ket{0},\ket{1}\}$ basis, except if
$\{p_{ij}\}_{ij}\in\{0,1\}$ for all $i$, $j$.
\label{thm4}
\end{theorem}
\begin{proof} 
From (\ref{con1}) of Theorem \ref{optimal_cond}, if Alice inputs
$i=1$, the measurement operators $\{\ketbra{0}{0},\ketbra{1}{1}\}$ are
optimal only if
\begin{equation}
\eta_0\sqrt{p_{10}(1-p_{10})}=(1-\eta_0)\sqrt{p_{11}(1-p_{11})}.
\end{equation}
For this to hold for all $\eta_0$, we require that either $p_{11}=0$ or
$p_{11}=1$, and either $p_{10}=0$ or $p_{10}=1$.  Similarly, if Alice
inputs $i=0$, we require either $p_{01}=0$ or $p_{01}=1$, and either
$p_{00}=0$ or $p_{00}=1$, in order that the specified measurement
operators are optimal.
\end{proof}
These exceptions correspond to functions that are deterministic, so do
not properly fall into the class presently being discussed.  Many are
essentially single-input, hence trivial, and all such exceptions are
either degenerate or not potentially concealing (see Section
\ref{det_fns}).

Our theorem also has the following consequence.
\begin{corollary}
One-sided variable bias coin tossing (see Chapter \ref{Ch2}) is
impossible.
\end{corollary}
\begin{proof}
A one-sided variable bias coin toss is the special case where both
$p_{00}=p_{10}$ and $p_{01}=p_{11}$.  These cases are not exceptions
of Theorem \ref{thm4}, and hence are impossible.
\end{proof}

\subsection{Example:  The Impossibility Of OT}
\label{OT_imposs}
Here we show explicitly how to attack a black box that performs {\sc
OT} when used honestly.  This is a second proof of its impossibility
in a stand-alone manner (the first being Rudolph's
\cite{Rudolph}). \footnote{Impossibility had previously been
argued on the grounds that {\sc OT} implies {\sc BC} and hence is
impossible because {\sc BC} is.  However, while this argument rules
out the possibility of a composable {\sc OT} protocol, a stand-alone
one is not excluded.}

The probability table for this task is given in Table \ref{tabOT}.
\begin{table}
\begin{center}
\begin{tabular}{cc|cc|}
\multirow{2}{*}{$p(k|i)$}&&\multicolumn{2}{|c|}{$i$}\\
&&\;0\;&\;1\;\\
\hline
\multirow{3}{*}{$k$}
&0\;&\;$\frac{1}{2}$\;&\;0\;\\
&1\;&\;0\;&\;$\frac{1}{2}$\;\\
&?\;&\;$\frac{1}{2}$\;&\;$\frac{1}{2}$\;\\
\hline
\end{tabular}
\caption{Probability table for oblivious transfer.}
\label{tabOT}
\end{center}
\end{table}

In an honest implementation of {\sc OT}, Bob is able to guess Alice's input
with probability $\frac{3}{4}$.  However, the final states after using
the ideal black box are of the form
$\ket{\psi_b}=\frac{1}{\sqrt{2}}\left(\ket{b}+\ket{?}\right)$, where
$\ket{0}$, $\ket{1}$ and $\ket{?}$ are mutually orthogonal.  These are
optimally distinguished using the {\sc POVM} $(E_0,\openone-E_0)$, where
\begin{equation}
E_0=\frac{1}{6}\left(\begin{array}{ccc}
2+\sqrt{3}&-1&1+\sqrt{3}\\
-1&2-\sqrt{3}&1-\sqrt{3}\\
1+\sqrt{3}&1-\sqrt{3}&2\\
\end{array}\right).
\end{equation}
This {\sc POVM} allows Bob to guess Alice's bit with probability
$\frac{1}{2}\left(1+\frac{\sqrt{3}}{2}\right)$, which is significantly
greater than $\frac{3}{4}$.

\section{Discussion} 
We have introduced a black box model of computation, and have given a
necessary condition for security.  Even if such black boxes were to
exist as prescribed by the model, one party can always break the
security condition.  Specifically, by inputing a superposed state
rather than a classical one, and performing an appropriate measurement
on the outcome state, one party can always gain more information on
the input of the other than that gained using an honest strategy.  In
the case of deterministic functions, this attack has only been shown
to work if the function is non-degenerate and potentially concealing.
In the case where the sole purpose of the function is to learn
something about the other party's input, this class of function is
the most natural to consider.

Our theorems deal only with the simplest cases of the relevant
functions.  However, the results can be extended to more general
functions as described below.

{\bf Larger input alphabets:} A deterministic function is impossible
to compute securely if it possesses a $3\times 3$ submatrix which is
potentially concealing and satisfies the degeneracy requirement.  This
follows because Alice's prior might be such that she can reduce Bob to
three possible values of $j$.  This argument does not rule out the
possibility of all larger functions, since some exist that are
potentially concealing without possessing a potentially concealing
$3\times 3$ subfunction.  Nevertheless, we conjecture that all
potentially concealing functions have a cheating attack which involves
inputing a superposition and then optimally measuring the outcome.

In the non-deterministic case, all functions with more possibilities
for $i$ and $j$ values possess $2\times 2$ submatrices that are ruled
out by the attacks presented, or reduce to functions that are
one-input.  Therefore, no two-party non-deterministic computations can
satisfy our security condition.

{\bf Larger output alphabets:} In the non-deterministic case, we
considered only binary outputs.  We conjecture that the attacks we
have presented work more generally on functions with a larger range of
possible outputs.


\bigskip
We have not proven that the aforementioned attacks work for any
function within the classes analysed, although we conjecture this to
be the case.  Furthermore, for any given computation, one can use the
methods presented in this chapter to verify its vulnerability under
such attacks.

Our results imply that there is no way to define an ideal suitable to
realise secure classical computations in a quantum relativistic
framework.  Hence, without making additional assumptions, or invoking
the presence of a trusted third party, secure classical computation is
impossible to realise using the usual notions of security.  The
quantum relativistic world, while offering more cryptographic power
than the classical world, as exemplified in Chapter \ref{Ch1}, still
does not permit a range of computational tasks.  Table \ref{fns2}
summarizes the known results for unconditionally secure two-party
computation.

\begin{sidewaystable}
\begin{tabular}{|l|l|c|c|l|}
\hline
\multicolumn{2}{|c|}{Type of computation}&Securely Implementable&Comment\\
\hline
Zero-input&Deterministic &$ \checkmark $  &Trivial\\
        &Random       one-sided         &$ \checkmark $  &Trivial\\
        &Random       two-sided         &\checkmark   &Biased $n$-faced die roll\\
\hline
One-input &Deterministic &$\checkmark $  &Trivial\\
        &Random       one-sided         &\ding{55}$^*$&One-sided
variable bias $n$-faced die roll\\
        &Random       two-sided         &$\checkmark^*$ &Variable bias
$n$-faced die roll\\
\hline
Two-input &Deterministic one-sided         &\ding{55}    &cf. Lo\\
        &Deterministic two-sided        &\ding{55}$^*$  &this chapter\\
        &Random       one-sided        &\ding{55}$^*$  &this chapter\\
        &Random       two-sided        &\ding{55}$^*$  &this chapter\\
\hline
\end{tabular}
\caption{Functions computable securely in two-party computations using
  (potentially) both quantum and relativistic protocols, when
  unconditional security is sought.  $\checkmark$ indicates that all
  functions of this type are possible, \ding{55} indicates that all
  functions of this type are impossible, $\checkmark^*$ indicates that
  a wide range of functions of this type are possible and conjectures
  made in Chapter \ref{Ch2} imply that all functions of this type are
  possible, and \ding{55}$^*$ indicates that a wide range of functions
  of this type are impossible and conjectures made in this Chapter
  imply that all functions of this type are impossible. This is the
  version of Table \ref{fns} updated in light of our work.}
\label{fns2}
\end{sidewaystable}

One reasonable form of additional assumption is that the storage power
of an adversary is bounded\comment{\footnote{In fact, it is believed
that there is a physical principle that bounds the information
capacity of a region in terms of its surface area (see
\cite{Bekenstein} for a recent review), although this bound is so
large that it is difficult to imagine it ever being of use.}}.  The
so-called bounded storage model has been used in both classical and
quantum settings.  This model evades our no-go results because
limiting the quantum storage power of an adversary forces them to make
measurements.  This collapses our unitary model of computation.  In
the classical bounded storage model, the adversary's memory size can
be at most quadratic in the memory size of the honest parties in order
to form secure protocols \cite{Cachin&,Ding&}.  However, if quantum
protocols are considered, and an adversary's quantum memory is
limited, a much wider separation is possible.  Protocols exist for
which the honest participants need no quantum memory, while the
adversary needs to store half of the qubits transmitted in the
protocol in order to cheat \cite{Damgard&3}.

In the recent literature, there have been investigations into the
cryptographic power afforded by theories that go beyond quantum
mechanics.  Such theories are often constrained to be non-signalling.
Popescu and Rohrlich investigated violations of the {\sc CHSH}
inequality (see Section \ref{Bell}) in non-signalling theories
\cite{PR}.  Such theories are able to obtain the maximum algebraic
value of the {\sc CHSH} quantity, 4.  The hypothetical device that
achieves such a violation has subsequently been called a non-local
box.  Devices of this kind would allow substantial reductions in the
communication complexity of distributed computing tasks \cite{vD} and
have been shown to allow any two-party secure computation
\cite{BCUWW}.  One might conclude that there is a further gap in
cryptographic power between non-signalling theories and quantum ones.
However, we argue that this is not justified for two reasons.
Firstly, in non-local box cryptography, one gives such boxes for free
to parties which need them.  Secondly, no procedure for doing joint,
or even alternative single measurements is prescribed to a non-local
box setting.  To make a fair comparison between non-local box
cryptography and standard quantum cryptography, one should consider a
quantum scenario in which separated parties are given shared singlets
for free, and also constrain them to make one of two measurements on
each state they hold.  Alternatively, one could find a new theory in
which non-local boxes emerge as features.  In the absence of such a
theory, one should be cautious about making comparisons.

Recently, it has been shown that any non-signalling box whose
correlations are non-separable is sufficient for bit commitment
\cite{WW}.  This includes the case where the correlations are quantum,
or indeed weaker.  Since quantum (non-relativistic) bit commitment is
impossible, even given access to shared {\sc EPR} pairs, the
additional cryptographic power cannot be attributed to the presence of
correlations above those that are possible using quantum mechanics
alone.  It remains an open question whether the same is true for {\sc
OT}.

\comment{We consider
it an open question as to whether nature provides us with a
cryptographically weaker physical theory than that which could be
afforded by the most general non-signalling theory.}

We further remark that the cheating strategy we present for the
non-deterministic case does not work for all assignments of Alice's
prior over Bob's inputs---there exist functions and values of the
prior for which it is impossible to cheat using the attack we have
presented.  This continues to be the case when we allow Alice to
choose any input state, including ones entangled with some space that
she keeps).  As a concrete example, consider the set
$(p_{00},p_{01},p_{10},p_{11})=(\frac{47}{150},\frac{103}{150},\frac{8}{9},\frac{5}{9})$,
with $\eta_0=\frac{1}{2}$ in the two sided version.  Hence, in practice,
there could be situations in which Bob would be happy to perform such
a computation, for example, if he was sure Alice had no prior
information over his inputs.



\chapter{Private Randomness Expansion Under Relaxed Cryptographic Assumptions}
\label{Ch4}

\begin{quote}
{\it ``The generation of random numbers is too important to be left to
  chance.''} -- Robert R. Coveyou
\end{quote}

\section{Introduction}
As a casino owner, Alice has a vested interest in random number
generation.  Her slot machines use pseudo-random numbers which she is
eager to do away with.  Alice has a sound command of quantum physics,
and realises a way to produce guaranteed randomness.  However, running
a casino is not easy, and Alice has neither the time nor resources to
construct the necessary quantum machinery herself.  Instead, her local
merchant, the shady Dr Snoop, offers to supply the necessary parts.
Naturally Alice is suspicious, and would like some way of ensuring
that Snoop's equipment really is providing her with a source of
private random bits.

\comment{Like all members of society, Alice has some important secrets.  From
time to time, she has need to send one of these to Bob.  Alice
possesses a sound command of physics, and understands the intricacies
of quantum key distribution but has neither the time nor resources to
construct the quantum machinery necessary to implement it.  Instead,
her local merchant, the shady Dr Snoop, offers to supply the necessary
parts.  The naturally cautious Alice wants to test the equipment
supplied in order to gain confidence in it before using it to send her
secret.}


\bigskip
Random numbers are important in a wide range of applications.  In
some, for example statistical sampling or computer simulations,
pseudo-randomness may be sufficient.  Psuedo-random sources satisfy
many tests for randomness, but are in fact deterministically generated
from a much shorter seed.  In applications such as gambling or
cryptography, this may be detrimental.  Since quantum measurements are
the only physical processes we know of that are random, it is natural
to construct random number generators based on these.  Devices which
generate randomness through quantum measurement have recently hit the
marketplace, but what guarantee does the consumer have that these
perform as claimed?  In this chapter, we investigate protocols that
guarantee private random number generation even when all the devices
used in the process come from an untrusted source.  This corresponds
to relaxing Assumption \ref{ass4} (see Section \ref{thesetting}), that
each party has complete knowledge of the operation of the devices they
use to implement a protocol.  We use the task of expanding a random
string, that is, using a given random string in some procedure in
order to generate a longer one\footnote{Note that this task involves
only one party trying to expand a random string in contrast to the
task of extending coin tosses discussed in Section
\ref{simulator_role}, where both Alice and Bob must generate the same
shared expansion.}, to illustrate that some cryptographic tasks are
possible even when this assumption is dropped.

Expansion of randomness comes in two flavours.  In the weakest form,
one simply wants to guarantee that the lengthened string really is
random and could not have been influenced by any outside source.  If
one also requires that no information on the lengthened string be
accessible to another party, then a stronger protocol is needed.  The
latter task, we refer to as {\em private} randomness expansion, and is
clearly sufficient for the former\footnote{A string formed by
measuring individual halves of singlets in some fixed basis is random,
but not secret, since the holder of the other half can discover the
random data.}.  The possession of guaranteed randomness is useful in
many contexts.  In a gambling scenario, for instance, several players
may learn the outcome of a random event (e.g., the spin of a roulette
wheel) but would be at a great advantage if they could influence
it. The {\sc BB}84 {\sc QKD} scheme on the other hand requires a
private random string to choose the bases to use.  Private randomness
expansion will be the focus of this chapter.

We give a protocol that uses an initial private random string,
together with devices supplied by an adversary, to expand this initial
string.  Our protocol is such that any specified amount of additional
randomness can be generated using a sufficiently long initial string.
Further, we give a second protocol which allows a (sufficiently long)
initial string to be expanded by any amount.  The length of initial
string required depends on the desired tolerance for successful
cheating by Snoop.  This second protocol has the undesirable feature
of requiring a large set of sites that cannot communicate with one
another.  Our protocols are not optimized for efficiency, and at
present do not have full security proofs.

\subsection{The Setting}
Let us now iterate the practical significance of dropping Assumption
\ref{ass4}.  Randomness expansion is a single party protocol.  We assume
that all quantum devices that the user, Alice, will use to perform the
protocol were sourced by Snoop\footnote{Since Alice herself is a
classical information processing device, it is unreasonable to ask
that Snoop created all classical devices.}.  Snoop will supply devices
that he claims function exactly as Alice prescribes\footnote{We assume
that Snoop can construct any device consistent with the laws of
physics, and that Alice does not ask for impossible devices.}.  The
devices cannot send communications outside of Alice's lab unless she
allows them to (cf.\ Assumption \ref{ass1}), and Alice can, if
necessary, prevent them from communicating with each another.

To become confident that the devices have not been tampered with,
Alice will perform some test on them.  In keeping with Kerckhoff's
principle \cite{Kerckhoff}, we assume that Snoop knows completely the
details of such tests.  If all of Alice's devices come from Snoop,
there is an immediate no-go result.  We idealize Alice's procedure for
testing the devices as a sequence of operations generating a member of
a set of possible outcomes.  Certain outcomes result in her rejecting
the devices, while others lead to their acceptance.
\begin{theorem}
If Alice follows a deterministic procedure, and sources all of her
devices from Snoop, then she cannot distinguish the case where Snoop's
devices implement the procedure as intended from the case where his
devices make pre-determined classical outputs.
\end{theorem}
\begin{proof}
There exists a set of classical data that Alice will accept as a
passing of her test.  Snoop need simply provide devices that output
this set of data as required by Alice's procedure.
\end{proof}

To circumvent this no-go result, we give Alice an initial private
random string.  By using this string, she can ensure that Snoop does
not know every detail of her test procedure.  As we shall see, this
string is enough to constrain Snoop such that Alice can generate
random bits.  Since she needs an initial source of bits, this task is
randomness expansion.

\subsection{Using Non-Local Correlations}
\label{GHZ}
We have shown that without the use of an initial random string, Alice
cannot perform randomness expansion.  However, it is also the case
that without exploiting the non-local features of quantum mechanics,
she cannot either.  This is a corollary to the following theorem.
\begin{theorem}
If Alice sources all of her devices from Snoop and follows a local
procedure, then she cannot distinguish the case where Snoop's
devices implement the procedure as intended from the case where his
devices make classically generated outputs\footnote{A classically
generated output is one formed from the input without use of quantum
states or measurements.}.
\end{theorem}
\begin{proof}
If all the processes occur locally, we can reduce any setup to the
following.  Snoop supplies a device into which Alice inputs her random
string, before it produces an output.  Snoop's cheating strategy in
this case is simply to program his device with a correct output for
each of Alice's possible inputs.
\end{proof}
It then follows that since Snoop's devices can offer a one-to-one
correspondence between Alice's input and their output, the amount of
private randomness in Alice's possession remains constant.

Alice's tests need to exploit non-local effects in order to be of use.
To see that these evade the no-go results above, consider two
spatially separated devices, both inside Alice's laboratory.  Alice
inputs part of her random string into each device, and demands that
the devices produce fast outcomes (i.e., within the light travel time
between the two devices).  Thus, the second device must follow a
procedure that is independent of the random string input to the first,
and vice-versa.  If Alice is to test for non-classical correlations
between the outcomes, then Snoop's potential to cheat is constrained.
States which produce non-classical correlations possess some intrinsic
randomness, and so, by verifying that Snoop's devices are
producing such states, Alice can be sure that she derives genuine
randomness from them.

The non-local nature of quantum mechanics is often exemplified using
the {\sc CHSH} test, as described in Section \ref{Bell}.  However, the {\sc CHSH}
test is not well suited for our purposes because it is based on
statistics generated over lots of runs.  In a finite run, it is
impossible to say for certain whether the value achieved was due to
malicious behaviour or simply bad luck.  Such a property is an
inconvenience, but not fatal since, by increasing the number of runs,
we can ensure that the probability of Snoop passing the test if he has
deviated from it is arbitrarily small.  The {\sc GHZ} test, discussed below,
has a much neater failure detection than this.

Consider instead the following test, which we call a {\sc GHZ} test
\cite{GHZ}, after its inventors Greenberger, Horne and Zeilinger.
\nomenclature[zGHZ]{{\sc GHZ}}{Greenberger, Horne and Zeilinger, the
collaboration behind the so-called {\sc GHZ} Bell test, and after whom
the {\sc GHZ} state,
$\frac{1}{\sqrt{2}}\left(\ket{000}-\ket{111}\right)$, was named}
Alice asks for three devices, each of which has two settings (which we
label $P$ and $Q$ following Section \ref{Bell}) and can output either
$1$ or $-1$.  Alice is to consider the four quantities $P_1P_2P_3$,
$P_1Q_2Q_3$, $Q_1P_2Q_3$ and $Q_1Q_2P_3$.  She demands that the first
of these is always $-1$, while the remaining three are $+1$.  That
these cannot be satisfied by a classical assignment can be seen as
follows.  Consider the product of the four quantities, which according
to Alice's demands must be $-1$.  However, the algebraic expression is
$P_1^2P_2^2P_3^2Q_1^2Q_2^2Q_3^2$, which for a classical assignment
must be positive.  This is a contradiction, and so no classical
assignment exists.  If, instead, the $\{P_i\}$ and $\{Q_i\}$ are
formed by the outcomes of measurements acting on an entangled quantum
state, then such demands can always be met.  In Appendix \ref{AppC},
we describe the complete set of operators and states that achieve
this.  In essence, all such operators behave like Pauli $\sigma_x$ and
$\sigma_y$ operators
\nomenclature[g\sigma_x]{$\sigma_x,\sigma_y,\sigma_z$}{The Pauli spin
operators} and the state behaves like a {\sc GHZ} state, that is, the
state $\frac{1}{\sqrt{2}}\left(\ket{000}-\ket{111}\right)$, up to
local unitary invariance.

The {\sc GHZ} test does not rely on the statistical properties of several
runs.  Rather, the outcome of each run is specified.  If any run
contradicts the specified value, then one can be sure that the state
and operators are not the ones claimed.  This is highly useful
cryptographically, as it allows Alice to be certain when she has
detected interference.  {\sc GHZ} tests have a further advantage over {\sc CHSH}
in that they offer a higher rate of increase in randomness (we discuss
this further in Section \ref{effic}).

\section{Private Randomness Expansion}
\subsection{The Privacy Of A Random String}
Let us first consider the case where each party has classical
information.  Alice has string $x\in X$.  Snoop has
$z\in Z$, partially correlated with $x$, these strings having been
drawn from a distribution $P_{XZ}$.  Alice's string is private if
$I(X:Z=z)$ is negligible, i.e., string $X$ is essentially uniformly
distributed from Snoop's point of view.

Now consider the case where Snoop's information on Alice's string is
quantum.  In general, it is not enough to demand that for any
measurement Snoop performs, his resulting string $z$ is such that
$I(X:Z=z)$ is negligible.  Such a definition does not ensure that
Alice's string can be used in {\it any} further application.  The
reason is that Snoop need not measure his system to form a classical
string, but can instead keep hold of his quantum system.  He may be
then able to acquire knowledge which constitutes cheating in the
further application (see also the discussion in Section
\ref{sec_defs}).  For example, as a result of parameters revealed in
the further application, Snoop might be able to identify a suitable
measurement on his original quantum system that renders the further
application insecure.  (See \cite{BHLMO, RennerKoenig, Renner} for a
further discussion, and an explicit example.) Thus, a key with the
property that for all measurements by Snoop $I(X:Z=z)$ is negligible,
cannot be treated in the same way as a private random key.

Instead, as discussed in Chapter \ref{Ch3}, security definitions are
defined with reference to the properties of a suitable ideal.  In an
ideal protocol, the final state is of the form
$\frac{1}{|A_n|}\sum_{i\in
A_n}\ketbra{i}{i}\otimes\rho_Z\equiv\sigma_I$, where $\rho_Z$ is
Snoop's final system and is independent of $i$, $A_n$ represents the
set of strings of length $n$, and $|A_n|=2^n$ is the size of set
$A_n$.  In a real implementation, the final state has the form
$\sum_{i\in A_n}P_I(i)\ketbra{i}{i} \otimes\rho_Z^i\equiv\sigma_R$.  A
useful security definition is that $D(\sigma_I,\sigma_R)\leq\epsilon$,
which implies that the two situations can be distinguished with
probability at most $\epsilon$ (see Section \ref{tracedist}).
Moreover, since the trace distance is non-increasing under quantum
operations \cite{Nielsen&Chuang}, this condition must persist when the
string is used in any application, and hence the string satisfies a
stand-alone security definition (see Section \ref{sec_defs}).  Since
the protocol is non-interactive, and takes place entirely within
Alice's laboratory, it is clear that universally composable security
is also realized.

In many applications, the string produced may not satisfy a security
requirement of this kind without first undergoing privacy
amplification.  
In Section \ref{priv_amp}, we discussed privacy amplification in a
three party scenario, in which Alice and Bob seek to generate a shared
random string on which Eve's information is negligible.  Alice and Bob
are required to communicate during the amplification stage, and thus
leak information about the amplification to Eve.  Private randomness
expansion, on the other hand, is a two party game.  No information
need be leaked in amplification since there is no second honest party
needing to perform the same procedure.  For instance, if $universal_2$
hashing is used, the adversary never gains any knowledge about the
hash function.  The randomness used to choose it remains private and
hence acts catalytically.

\subsection{Definitions}
Let us denote Alice's initial private uniform random string by $x\in
X$.  This string has length $n$ bits.  Alice expands $x$, generating
the additional string $s\in S$.  A protocol for private randomness
expansion using devices supplied by Snoop is $\epsilon$-secure if, for
any strategy followed by Snoop whereby he holds Hilbert space
$\mathcal{H}_Z$, we have
\begin{equation}
D(\rho_{SZ},\rho_{U_S}\otimes\rho_Z)\leq\epsilon,
\end{equation}
where $\rho_{U_S}$ denotes the maximally mixed state in
$\mathcal{H}_S$.

A protocol for the weaker task of randomness expansion is
$\epsilon$-secure if
\begin{equation}
D(\rho_S,\rho_{U_S})\leq\epsilon.
\end{equation}
Then, no restriction is placed on how much information the state in
$\mathcal{H}_Z$ provides on $S$.  For instance, it could be entangled
in such a way that Snoop can always find $S$.  What is important is
that Snoop cannot influence $S$ in any way, except with probability
$\epsilon$.

Like in previous chapters, using $N_1,\ldots,N_r$ as security
parameters, we say that a protocol is secure if $\epsilon\rightarrow
0$ as the $N_i\rightarrow\infty$, and that a protocol is perfectly
secure if $\epsilon=0$ for some fixed finite values of the $N_i$.

\subsection{Finite Expansion}
\label{finite_exp}
We now give a protocol which allows a private random string to be
expanded.  Before undergoing the protocol, Alice asks Snoop for three
devices, each of which has two settings (inputs), ($P_i$ and $Q_i$ for
the $i$th device) and can make two possible outputs, $+1$ or $-1$.
These devices cannot communicate with agents outside of Alice's
laboratory (cf.\ Assumption \ref{ass1}), nor with one another.  Alice
asks that whenever these devices are used to measure one of the four
{\sc GHZ} quantities ($P_1P_2P_3$, $P_1Q_2Q_3$, $Q_1P_2Q_3$ and
$Q_1Q_2P_3$), they return the outcomes specified in Section \ref{GHZ}
(i.e., $-1$, $+1$, $+1$ and $+1$ respectively).  \footnote{In
practice, Alice might ask for devices that measure either $\sigma_x$
or $\sigma_z$, and for a further device that creates {\sc GHZ}
states. Of course, she will not be able to distinguish this scenario
from one satisfying the test but using a different set of states and
operators, hence we have kept the description as general as possible.}
We call these three devices taken together a device triple.  Alice
uses her device triple in the following procedure.

\begin{protocol}\qquad
\label{protA}

\begin{enumerate}
\item Alice chooses security parameter $\epsilon$, to give a
  sufficiently small probability of Snoop successfully cheating.  She
  divides her string $x$ into two strings $x_1$ and $R$, of equal
  length.
\item \label{pst1} Alice uses 2 bits of $x_1$ to choose one of the four
  tests, via the assignment in Table \ref{GHZtab}.
\item She performs the corresponding test, by having each of three
  agents make inputs to their boxes and receive their outputs such
  that light could not have travelled between any pair of boxes
  between input and output\footnote{Alternatively, Alice can avoid the
  need for large separations if she can ensure no communication
  between devices after the protocol begins, e.g.\ by putting each
  device in its own separate laboratory.}.
\item \label{pst3} If she receives the wrong product of outputs, she
  aborts, otherwise she turns her output into a bit string using the
  assignments given in Table \ref{GHZtab}.  In this way, Alice builds
  a random string $x'\in X'$.
\item Alice repeats steps \ref{pst1}--\ref{pst3} until she has
  depleted $x_1$.
\item \label{PA} Alice bounds
  $H_{\infty}^{\epsilon}(\rho_{X'X_1Z}|X_1Z)$.  Here,
\begin{equation}
\rho_{X'X_1Z}=\sum_{x',x_1}P_{X'X_1}(x',x_1)\ketbra{x'x_1}{x'x_1}\otimes\rho_Z^{x'x_1},
\end{equation}
  and $\mathcal{H}_Z$ is the Hilbert space held by Snoop.  She then
  performs privacy amplification using a $universal_2$ hash function,
  where the random string $R$ is used to choose the hash function.
  (Note that $R$ has the same length as $x'$ \cite{CW,WC}.)  If
  Alice's final string, $s$, has length $\tau$, then Equation
  (\ref{sm_min_ent}) implies that $s$ can be distinguished from
  uniform with probability at most
  $\epsilon+\frac{1}{2}2^{-\frac{1}{2}\left(H_{\infty}^{\epsilon}(\rho_{X'X_1Z}|X_1Z)-\tau\right)}$. \footnote{Since
  it is only quantum devices that are supplied by Snoop, and hashing
  is a classical procedure, there is no security issue associated with
  this step.}
\end{enumerate}
\end{protocol}

\begin{table}
\begin{center}
\begin{tabular}{|c|cccc|}
\hline
bit sequence&00&01&10&11\\
\hline
input&$P_1P_2P_3$&$P_1Q_2Q_3$&$Q_1P_2Q_3$&$Q_1Q_2P_3$\\
\hline
\end{tabular}

\bigskip
(a)

\bigskip
\begin{tabular}{|c|cccc|}
\hline
\multirow{2}{*}{output}&$---$&$-++$&$+-+$&$++-$\\
&$+++$&$+--$&$-+-$&$--+$\\
\hline
assignment&00&01&10&11\\
\hline
\end{tabular}

\bigskip
(b)
\caption{Assignment table for (a) choosing the inputs to the three
devices based on two random bits, and (b) assigning the outputs
generated from the three devices to form two new random bits.  
}
\label{GHZtab}
\end{center}
\end{table}

\begin{figure}
\begin{center}
\includegraphics[width=\textwidth]{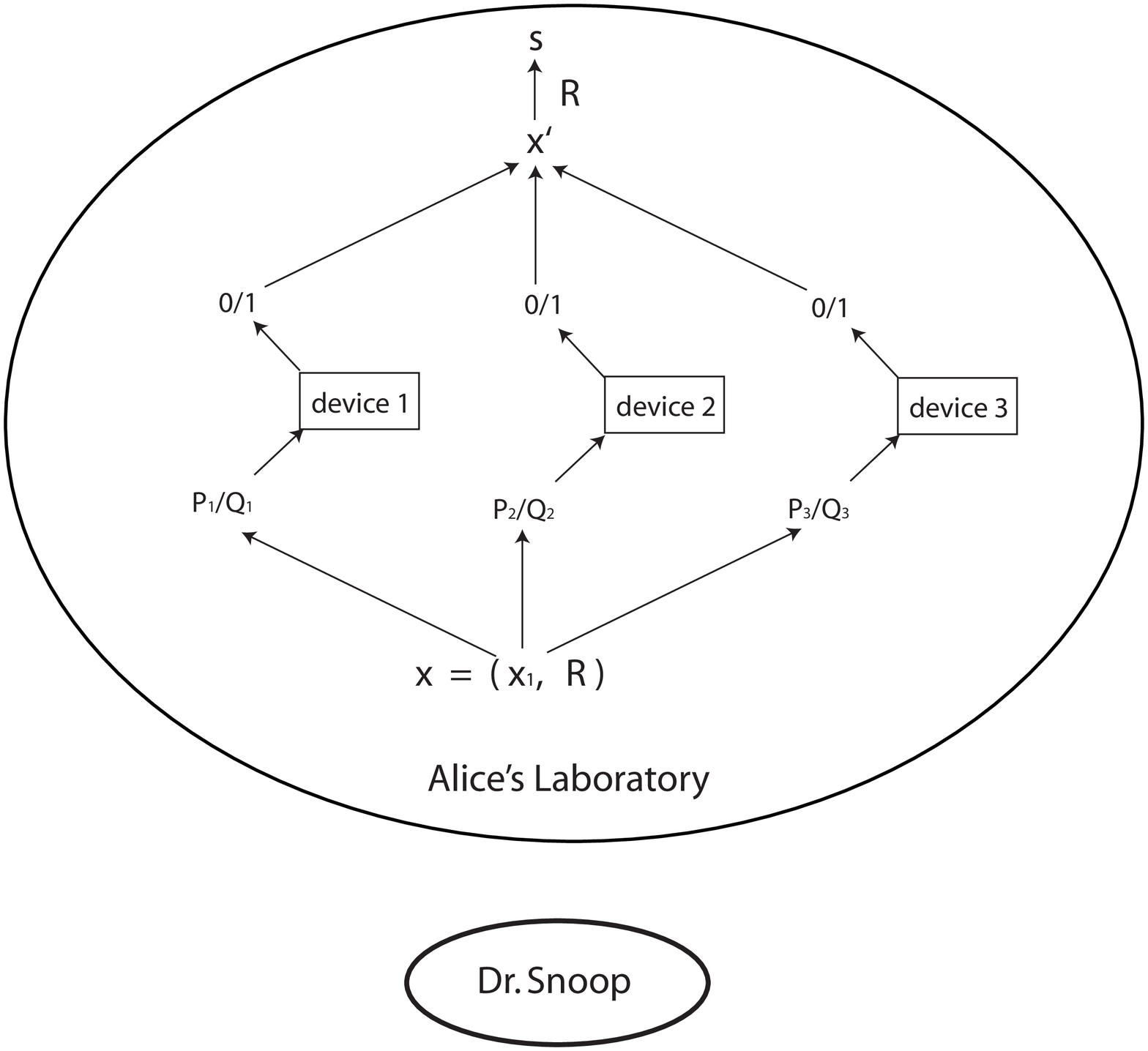}
\caption{Diagram of the steps in Protocol \ref{protA}.  Together
  devices 1--3 form a device triple.}
\label{fig:Rand_exp1}
\end{center}
\end{figure}

This protocol is illustrated in Figure \ref{fig:Rand_exp1}.  Note that
Alice bounds the quantity $H_{\infty}^{\epsilon}(\rho_{X'X_1Z}|X_1Z)$,
rather than $H_{\infty}^{\epsilon}(\rho_{X'Z}|Z)$.  This ensures that
if Snoop discovers $x_1$ {\it after} the protocol has taken place
(e.g., perhaps it is used as part of some further application), the
string $s$ remains secure.  This is important for the composability of
Protocol \ref{protA}.  Following the discussion of privacy
amplification in Section \ref{priv_amp}, the same is true if $R$ is
subsequently divulged.  The concatenation $(x,s)$ is the final private
random string generated by the protocol.  It is manifestly longer than
the initial one.  Moreover, if Snoop is honest, then Protocol
\ref{protA} uses 2 bits of $x$ while generating 2 new bits of
randomness each time the loop (i.e., Steps \ref{pst1}--\ref{pst3}) is
run.

Although $(x,s)$ is private with respect to the outside world, it is
not private with respect to the devices, which, being malicious, may
be programmed to remember their sequence of inputs and outputs.  Snoop
could then program her devices in the following way.  The first time
$x$ is input, the devices behave honestly, using genuine {\sc GHZ}
states and suitable measurement operators.  Alice's tests will then
all pass.  When $x$ is input again (which the devices know, because we
assume Snoop knows Alice's procedure), the devices can simply recall
the output they made in the first (honest) run.  With probability
$\frac{1}{4}$ the devices output these directly, otherwise they
randomly flip two of the three outputs. (The devices can be
pre-programmed with shared private randomness in order to do this.)
The outputs in this second run appear genuine from Alice's point of
view, but in fact contain no additional private randomness.
Therefore, the procedure cannot simply be repeated to generate an even
longer string.

\subsubsection{Security Against Classical Attacks}
Consider the situation where Alice performs the protocol as described,
while Snoop attempts to cheat.  In so doing, Snoop limits himself to
classical attacks (that is, to inserting known outcomes\footnote{Snoop
could distribute these outcomes according to some probability
distribution, but this will not help.  Additionally, he could make the
output depend on the input, but since when we bound the smooth
min-entropy we give Snoop the input, this also does not help.}).  If
he does this, his best attack has success probability $\frac{3}{4}$
per supposed {\sc GHZ} state, and gains him 2 bits of Alice's sequence,
$x'$.

Snoop can then have made a maximum of
$m=\frac{\log{\frac{1}{\epsilon}}}{\log\frac{4}{3}}$ attacks, except
with probability less than $\epsilon$.  So that his probability of
successful attack is less than $2\epsilon$, we require the hashing to
reduce the string length by
$2\left(\frac{1}{\log{\frac{4}{3}}}+1\right)\log{\frac{1}{\epsilon}}-2$
bits (see Equation (\ref{sm_min_ent})).  This is independent of the
number of {\sc GHZ} tests performed.  Provided the initial private
random string has length greater than twice this, it can be expanded,
except with probability less than $2\epsilon$.


In the cases where Snoop does not make an attack, two new pieces of
randomness are generated for each bit of $x_1$.  Therefore, against
classical attacks, this protocol increases the amount of private
randomness by a factor of $\lesssim\frac{3}{2}$, for large initial
amounts.

\subsubsection{Quantum Attacks}
Of course, limiting Snoop to classical attacks is an undesirable and
unrealistic assumption, especially given the fact that he is able to
produce Alice's quantum devices!  If Snoop performs a quantum attack,
then, before privacy amplification, the final state of the system
takes the form
\begin{equation}
\sum_{x',x_1}P_{X'X_1}(x',x_1)\ketbra{x'x_1}{x'x_1}_A\otimes\rho_Z^{x'x_1}.
\end{equation}
The length by which $x'$ needs to be reduced depends on
$H_{\infty}^{\epsilon}(\rho_{X'X_1Z}|X_1Z)$, and on $\epsilon$.  If
Alice wants an overall error probability less than $2\epsilon$, then
she can expand randomness provided that
$H_{\infty}^{\epsilon}(\rho_{X'X_1Z}|X_1Z)>2\log\frac{1}{2\epsilon}-2$
(see Equation (\ref{sm_min_ent})).

We have not been able to usefully bound
$H_{\infty}^{\epsilon}(\rho_{X'X_1Z}|X_1Z)$.  However, intuitively, we
expect that if a large number of {\sc GHZ} tests pass, Snoop's states must
be close to {\sc GHZ} states, except with probability exponentially small in
the number of tests.  In Appendix \ref{AppC}, we give a complete
description for the set of states that perfectly satisfy a {\sc GHZ} test.
Such states all generate 2 bits of private randomness per test.
Hence, we suspect that conditioned on $\scorpio$ {\sc GHZ} tests passing,
$H_{\infty}^{\epsilon}(\rho_{X'X_1Z}|X_1Z)$ is less than, but
approximately equal to $2\scorpio$.  In fact, to ensure our result, we
need a weaker conjecture, as follows.
\begin{conjecture}
\label{non-0}
If Protocol \ref{protA} is followed exactly by Alice, then for all
$\zeta>0$, $\epsilon>0$, there exists a sufficiently large integer,
$\scorpio$, such that conditioned on $\scorpio$ {\sc GHZ} tests passing,
$H_{\infty}^{\epsilon}(\rho_{X'X_1Z}|X_1Z)\geq\zeta$.
\end{conjecture}
If we accept this conjecture, then any desired length of additional
private random string can be generated using a sufficiently long
initial string (but at present, we do not know how to relate the
number of tests to the amount of additional randomness).  This
conjecture, together with Equation (\ref{sm_min_ent}), implies that we
can use Protocol \ref{protA} to generate $\tau$ additional random bits
except with probability 
\begin{equation}
\label{delta_defn}
\delta\leq\epsilon+\frac{1}{2}2^{-\frac{1}{2}\left(\zeta-\tau\right)}.
\end{equation}
Conjecture \ref{non-0} implies that for fixed $\zeta$, increasing
$\scorpio$ reduces $\epsilon$, while for fixed $\epsilon$, increasing
$\scorpio$ increases $\zeta$.  Hence, $\delta$ can be made arbitrarily
small for fixed $\tau$, by increasing $\scorpio$, which in turn
requires a longer initial string.

The capacity for generating any finite amount of additional randomness
may be useful in itself, but what is more useful is the ability to
take a string and expand it by an arbitrary amount.  In the next
section we give a protocol to do just that.

\subsection{Indefinite Expansion}
\label{indef}
If we accept Conjecture \ref{non-0}, then, except with a probability
exponentially small in the number of tests performed, the string
generated in Protocol \ref{protA} is private and random.  In this
section, we introduce a protocol that we conjecture allows a
sufficiently long initial random string to be expanded by an arbitrary
amount.

As we have mentioned, one cannot simply feed the original string, $x$,
twice into the same devices to double the amount of randomness gained.
On the other hand, if a second device triple is supplied by Snoop, and
can be assured no means of communication with the original (which is
reasonable given Assumption \ref{ass1}), then the string $(x,s)$
generated by the first triple is private and random with respect to
the second, and hence can be used as input.  One natural way to assure
independence is simply to provide spatial separation between the
device triples, in which case the same string $x$ (but not $(x,s)$)
can be used for each triple.  The overall protocol is as follows.

\begin{protocol}\qquad
\label{protNtriples}

\begin{enumerate}
\item Alice asks Snoop for $N$ device triples.
\item She places each device triple within its own sub-lab of her
laboratory such that no two can communicate.
\item Within each sub-lab, Alice uses her device triple to perform
protocol \ref{protA} with the same initial string, $x$, being used for
each.  The output generated in lab $i$ is string $s_i$, and we denote
the intermediate (non-hashed) string in this lab $x_i'$.  If any of
the {\sc GHZ} tests fail, the entire protocol aborts.
\item The strings $\{s_i\}$ are concatenated to form the final
output.
\end{enumerate}
\end{protocol}
\begin{figure}
\begin{center}
\includegraphics[width=\textwidth]{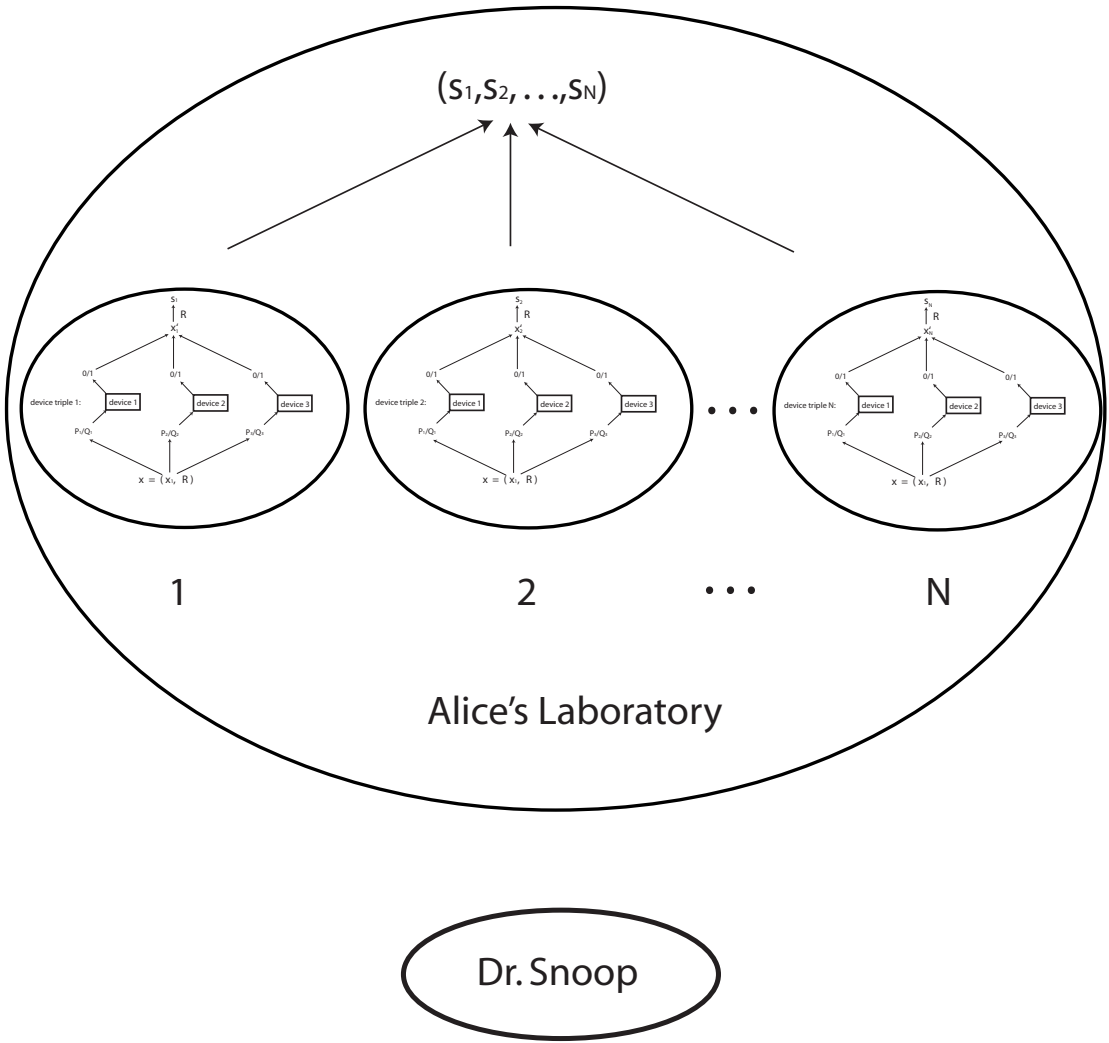}
\caption{Diagram of the steps in Protocol \ref{protNtriples}.  The
  same string, $x$, is used to generate the input to each device
  triple.  We have numbered each sub-lab in which instances of
  Protocol \ref{protA} occur.}
\label{fig:Rand_exp2}
\end{center}
\end{figure}
This protocol is illustrated in Figure \ref{fig:Rand_exp2}.

If we accept Conjecture \ref{non-0}, then each device triple, taken on
its own generates a non-zero amount of private randomness, except with
probability $\delta$, as defined by Equation (\ref{delta_defn}).  From
the discussion of privacy amplification in Section \ref{extractors},
this means that, for any system held by Snoop, he can distinguish
Alice's string from a uniform one with probability at most $\delta$.
This includes the case where, after the protocol has taken place,
Snoop learns $x$.  Since this must hold for any system held by Snoop,
we have that the strings $\{s_i\}$ are independent, since one possible
strategy for Snoop is to keep the other $N-1$ systems.  Hence, this
protocol generates $N$ times as much randomness as Protocol
\ref{protA}.  Thus, provided the initial private random string is
sufficiently long that it would generate a longer string in Protocol
\ref{protA}, it can be used to generate an arbitrarily large amount of
additional private randomness.

\section{Resource Considerations}
We have described two protocols for the expansion of random strings.
For a given initial string, the first protocol has limited potential
for expansion, while the second can be used to expand this string by
an arbitrary amount, but requires a large supply of device triples in
order to do so.  We consider the following resources:
\begin{enumerate}
\item The number of bits forming the initial string, $n$, and
\item The number of sub-laboratories Alice must form, $N$.
\end{enumerate}
Such resources limit the amount of additional randomness that can be
generated, as well as the probability of error achievable.

For a fixed initial string, Protocol \ref{protNtriples} allows an
arbitrary amount of randomness to be generated, provided that $n$ is
sufficient for the error tolerance required.  On the other hand, if
$N$ is fixed as well, then there is some limit on the amount of
expansion possible.  Since Protocol \ref{protA} is called as a
sub-protocol of Protocol \ref{protNtriples}, we look to enhance the
former in order to improve efficiency.

There are two ways in which one might increase the amount of
randomness generated over that given using Protocol \ref{protA}.  The
first is to use a more efficient extractor than the $universal_2$ hash
functions we have considered, so that the relative size of $x_1$ over
$R$ could be increased.  The second is to use a more efficient test to
generate the additional randomness.  This latter consideration is
discussed in the next section.

\subsection{Beyond The GHZ Test}
\label{effic}
Consider the task of using an $n$ bit initial string in some procedure
in order to maximize the length of additional random string generated,
while relying on $universal_2$ hashing for privacy amplification.  We
use $universal_2$ hash functions which require a random string equal
in length to the string being hashed.  Consider now a {\sc GHZ}-like
test whose output is $\nu$ times the length of the input.  In order
to use such a test to form a new string, the $n$ bit string is
partitioned into two strings, one of length $\frac{n}{1+\nu}$ and
one of length $\frac{\nu n}{1+\nu}$.  The first of these is
used, via the {\sc GHZ}-like test, to generate a string of length
$\frac{\nu n}{1+\nu}$ which is hashed using the second to form
the final string.  In this way, the original $n$ bit string has been
used to form one of length $n\left(1+\frac{\nu}{1+\nu}\right)$
(ignoring the reduction in length required for security, which, for
large $n$, represents an arbitrarily small fraction of the length).
In the limit $\nu\rightarrow\infty$, the original string can be
doubled in length.  This should be compared to an increase of
$\frac{3}{2}$ times if the original {\sc GHZ} test is used, or
approximately $1.4$ times if {\sc CHSH} is used. \footnote{In the {\sc
GHZ} case, choosing between each of the four quantities to test uses
two bits of randomness, while the amount of randomness gained from a
successful test is also two bits according to quantum theory (each of
the four possible outcomes from any of the measurements (e.g., for
$P_1P_2P_3$, the four are $---$, $++-$, $+-+$ and $-++$) are equally
likely).  Hence, for this test, $\nu=1$.}

Arbitrarily large values of $\nu$ are possible for appropriately
constructed tests.  One such construction was conceived by Pagonis,
Redhead and Clifton \comment{({\sc PRC}) \nomenclature[zPRC]{{\sc
PRC}}{Pagonis, Redhead and Clifton, the collaboration behind a
generalization of the {\sc GHZ} test}} \cite{PRC}.  They have
presented a series of Bell-type tests which extend the {\sc GHZ} test
to more systems.  In the seven system version, Alice asks for seven of
the two input, two-output devices discussed previously, and is to
consider the eight quantities
\begin{flushleft}
$P_1Q_2Q_3Q_4Q_5Q_6Q_7$, $Q_1P_2Q_3Q_4Q_5Q_6Q_7$,
$Q_1Q_2P_3Q_4Q_5Q_6Q_7$, $Q_1Q_2Q_3P_4Q_5Q_6Q_7$,\\
$Q_1Q_2Q_3Q_4P_5Q_6Q_7$, $Q_1Q_2Q_3Q_4Q_5P_6Q_7$,
$Q_1Q_2Q_3Q_4Q_5Q_6P_7$, $P_1P_2P_3P_4P_5P_6P_7$.
\end{flushleft}
She demands that the first seven are always $+1$, while the last
should be $-1$.  Again, it is easy to see that this is classically
impossible.  We conjecture that quantum mechanically, all states which
satisfy these requirements are essentially seven system analogues of
the {\sc GHZ} state, i.e.\
$\frac{1}{\sqrt{2}}(\ket{0000000}-\ket{1111111})$ (like in the {\sc
GHZ} case discussed in Appendix \ref{AppC}), although this remains
unproven\footnote{The proof provided for the {\sc GHZ} case does not
generalize directly.}.  For this test, 3 bits of randomness are
required to choose amongst the eight settings, while in a successful
implementation of the test on this state, 6 bits of randomness are
generated by the output.  Higher dimensional versions of this test
(see \cite{PRC}) lead to larger increases still.  In the $k$th version
of this test, $4k-1$ devices are required to measure one of $4k$
quantities.  Such a test generates $4k-2$ bits of randomness, and
hence has an associated value $\nu=\frac{4k-2}{\log{4k}}$.

Although these tests allow a larger amount of additional randomness
per bit of original string, there is a tradeoff in that they generate
lower detection probabilities in the event that Snoop cheats.  This is
easily illustrated by considering a classical attack.  For a {\sc GHZ} test,
a classical attack can escape detection with probability $\frac{3}{4}$
per test, while in the seven system generalization, this figure is
$\frac{7}{8}$ per test.  However, without a relation between the
smooth min-entropy and the number of tests, we cannot fully classify
the tradeoff.

\section{Discussion}
In this chapter, we have introduced two protocols that we conjecture
allow the expansion of a private random string using untrusted
devices.  The second of our protocols provides an arbitrarily long
private random string.  This may be a useful primitive on which to
base other protocols in the untrusted device scenario, and this is an
interesting avenue of further work.  Such a scenario is of interest in
that it allows us to reduce our assumptions.  More fundamentally, we
can think of nature as our untrusted adversary which provides devices.
One could then argue that our protocols strengthen the belief that
nature behaves in a random way\footnote{Of course, it is impossible to
rule out cosmic conspiracy.}.

The untrusted devices scenario is a realistic one, and will become
important if quantum computers become widespread.  The ordinary user
will not want to construct a quantum computer themselves and will
instead turn to a supplier, in the same way that users of classical
computers do today.  The protocols in this chapter seek to provide
such users a guarantee that the devices supplied are behaving in such
a way that their outputs are private and random, to within a
sufficient level of confidence.



\addcontentsline{toc}{chapter}{Conclusions}
\fancyhead[RO]{{\bf Conclusions}}
\def\baselinestretch{1}

\chapter*{Conclusions}


\def\baselinestretch{1.66}

The cryptographic power present within a model depends fundamentally
on the physical theory underlying it.  Non-relativistic classical
theory does not give much power and unproven technological assumptions
often have to be employed in order to make cryptographic tasks
possible.  Non-relativistic quantum theory permits key distribution,
but remains insufficient for a range of other tasks.

We have investigated quantum relativistic protocols, the most powerful
allowed by current theory.  Using such protocols, we have been able to
widen the class of tasks known to be possible to include variable bias
coin tossing.  However, many remain impossible.  The current state of
the field for two-party protocols is summarized in Table \ref{fns2}
(see page \pageref{fns2}).  Nature itself has a built-in limit on the
set of cryptographic tasks allowed.  For some, it is fundamentally
necessary to appeal to assumptions about the adversary in order that
they be achieved.  One might speculate that developments to our
current theory (e.g.\ a theory of quantum gravity) could be such that
they alter the set of allowed tasks.

We have also investigated cryptographic tasks outside the standard
model.  Specifically, we have dropped the usual assumption that each
party trusts all the devices within their laboratory.  In the
untrusted devices model, any quantum devices used are assumed to be
produced by a malicious adversary.  Even within this highly
restrictive scenario, some cryptographic procedures can succeed.  We
have discussed the task of expanding a private random string in
detail, giving two protocols which we conjecture do just that.

\bigskip
Throughout this thesis, we have sought unconditional security.  Our
goal has been to use a minimal set of assumptions in order to do
cryptography.  Aside from its obvious practical benefits, a
classification of tasks as possible or impossible is of intellectual
interest as a way of giving insight into fundamental physics itself.
However, when considering real-world cryptography, unconditional
security is unattainable in the way we have described.  Our first
assumption, that each party has complete trust in the security of
their laboratory, for example, is at best an assumption about the
power of an adversary, since an impenetrable laboratory is impossible
to realize.

Trust is something of a commodity in cryptography.  In practice, the
overwhelming majority of users are much more trusting than we have
allowed for.  They will, for instance, accept the functionality of
their devices on faith, taking the presence of a padlock symbol in the
corner of their browser window as a guarantee that their
communications are being encrypted.  Furthermore, they provide any
malicious code on their system with a high capacity channel (an
internet connection) with which to release private data.  Trusted
suppliers are hence a virtual necessity in any large-scale
cryptographic network.

Ultimately, it is not for us to say which assumptions a given user
should accept and which they should not.  Instead, we set up protocols
and clearly state the assumptions under which they are secure.  In
this way, the responsibility of deciding whether a given protocol is
of use in a particular situation is delegated to its user.




\appendix
\fancyhead[RO]{{\bf Appendix \thechapter}}
\chapter{Maximizing The Probability Of Distinguishing Between Two Quantum States}
\label{AppA}
Here we prove the following theorem \cite{Helstrom}:

\begin{theorem} 
Bob is in possession of one of two states whose density matrices are,
$\rho_0$ and $\rho_1$, for which the prior probability of $\rho_0$ is
$\eta_0$, and of $\rho_1$ is $\eta_1=1-\eta_0$.  The {\sc POVM} which is optimal to
distinguish these states does so with success probability
$\frac{1}{2}\left(1+{\rm tr}\left|\eta_0\rho_0-\eta_1\rho_1\right|\right)$.
\end{theorem}

\begin{proof}
Our proof follows a similar argument to Nielsen and
Chuang \cite{Nielsen&Chuang}, but extends their result to the case of
unequal prior probabilities.

Consider a {\sc POVM} described by elements $\{E_i\}_{i=1}^N$, which
satisfy $\sum_{i=1}^N E_i=\openone$.  Measurement with this {\sc POVM}
on the state provided generates outcomes according to one of two
probability distributions.  We have, $P_I(i)={\rm tr}(\rho_0 E_i)$ and
$Q_I(i)={\rm tr}(\rho_1 E_i)$, where $P_I$ occurs with probability
$\eta_0$ and $Q_I$ with probability $\eta_1$.  On measuring outcome
$i$, our best guess of the distribution will be the one with
$\max(\eta_0 P_I(i),\eta_1 Q_I(i))$, this guess being correct with
probability $\frac{\max(\eta_0 P_I(i),\eta_1 Q_I(i))}{\eta_0
P_I(i)+\eta_1 Q_I(i)}$.  The overall probability that we guess
correctly using this {\sc POVM} is then, $\sum_i\max(\eta_0 P_I(i),\eta_1
Q_I(i))$. Let us label the $\{P_I(i),Q_I(i)\}$ such that $P_I(i)\geq
Q_I(i)$ for $i=1,...,d$, and $P_I(i)<Q_I(i)$ for $i=d+1,...,N$.  Then,
\begin{eqnarray}
\nonumber\sum_{i=1}^N\max(\eta_0 P_I(i),\eta_1 Q_I(i))&=&\sum_{i=1}^d\eta_0 P_I(i)+\sum_{i=d+1}^N\eta_1 Q_I(i)\\
\nonumber&=&\sum_{i=1}^N\left|\eta_0 P_I(i)-\eta_1 Q_I(i)\right|+\sum_{i=1}^d\eta_1
Q_I(i)+\sum_{i=d+1}^N\eta_0 P_I(i)\\
&=&\frac{1}{2}\left(1+\sum_{i=1}^N\left|\eta_0 P_I(i)-\eta_1 Q_I(i)\right|\right)
\end{eqnarray}
Let us now define positive operators $\Upsilon_0$ and $\Upsilon_1$
with orthogonal support such that
$\eta_0\rho_0-\eta_1\rho_1=\Upsilon_0-\Upsilon_1$, and hence,
$\left|\eta_0\rho_0-\eta_1\rho_1\right|=\Upsilon_0+\Upsilon_1$.

We then have that,
\begin{eqnarray}
\nonumber\sum_{i=1}^N\left|\eta_0 P_I(i)-\eta_1
Q_I(i)\right|&=&\sum_{i=1}^N\left|E_i(\eta_0\rho_0-\eta_1\rho_1)\right|\\
\label{ineq}
&\leq &{\rm tr}\left|\eta_0\rho_0-\eta_1\rho_1\right|,
\end{eqnarray}
where the final inequality follows from the fact that,
\begin{eqnarray*}
\left|{\rm tr}\left(E_i\left(\Upsilon_0-\Upsilon_1\right)\right)\right|\leq{\rm
  tr}\left(E_i\left(\Upsilon_0+\Upsilon_1\right)\right)={\rm
  tr}\left(E_i\left|\eta_0\rho_0-\eta_1\rho_1\right|\right),
\end{eqnarray*}
and $\sum_i E_i=\openone$.

It remains to show that a {\sc POVM} exists that achieves equality in
(\ref{ineq}).  The relevant {\sc POVM} is $\{\Pi_0,\Pi_1\}$, where $\Pi_0$
is the projector onto the support of $\Upsilon_0$, and $\Pi_1$ is
likewise the projector onto the support of $\Upsilon_1$.  It is easy
to show that this {\sc POVM} has the desired properties.  We have hence
shown that the inequality (\ref{ineq}) can be saturated, hence the
result.
\end{proof}

As a corollary to this theorem, if the states to be distinguished have
equal priors, then they can be successfully distinguished with
probability at most $\frac{1}{2}(1+D(\rho_0,\rho_1))$, where
$D(\rho_0,\rho_1)\equiv\frac{1}{2}{\rm tr}|\rho_0-\rho_1|$ is the
trace distance between the two states.



\chapter{A Zero Knowledge Protocol For Graph Non-Isomorphism}
\label{AppD}

We use the task of providing a zero-knowledge proof for graph
non-isomorphism as an illustration that universally composable
security definitions can be satisfied, even in cryptographic protocols
in which one party responds after having received information from
another.  The protocol we use is classical and is found in \cite{GMW}.
We discuss its universally composable properties here.

A zero-knowledge proof is a protocol involving a verifier and a
prover.  It ensures that if some statement is true, and the protocol
is followed honestly, the prover is able to convince the verifier of
its truth, without revealing any other information.  Furthermore, if
the statement is false, it is impossible to convince the verifier that
it is true.

In this context, a graph is a series of nodes together with a defined
connectivity.  A zero-knowledge proof for graph non-isomorphism is one
in which the prover can convince a verifier that two graphs are
inequivalent under any permutation of their vertices.  In our
protocol, we assume that the prover has a device which solves the
graph isomorphism problem (i.e., a device which when give two graphs
decides whether they are isomorphic or not), but that the
(computationally bounded) verifier cannot solve this problem.
\begin{protocol}\qquad

We label the two graphs $G_0$ and $G_1$.  These are known to both
the prover and verifier.
\begin{enumerate}
\item The verifier picks either $G_0$ or $G_1$ at random and
  applies a random permutation, $\Xi$, to it.  This permuted graph is
  sent to the prover.
\item The prover tests to see whether the graph is a permutation of
  $G_0$ or $G_1$, and returns $0$ or $1$ to the verifier accordingly.
\item The verifier checks whether the prover was correct.  If so, this
  process is repeated until a sufficient confidence level is reached.
  If not, then no proof has been provided that the graphs are
  non-isomorphic.
\item We denote the outcome of the protocol $a=1$ if the proof is
  accepted, and $a=0$ if it is not.
\end{enumerate}
\end{protocol}

The ideal functionality has the following behaviour.
\begin{idf}\qquad
If the graphs are non-isomorphic, the prover can choose whether to
prove the non-isomorphism or not (i.e., whether the ideal will output
$a=1$ or $a=0$ to the verifier).  If they are isomorphic, the ideal
can only output $a=0$ to the verifier.
\end{idf}

Consider now a scenario in which the prover has access to an
additional device into which she inputs the permutation, and the
device returns either $c=0$ or $c=1$.  This device is analogous to the
additional device used by Bob in Section \ref{simulator_role} and uses
an algorithm unknown to the prover.  There, the additional device
broke the requirements for universally composable security.  However,
here it does not.  We show that a dishonest prover can use this device
and a device performing the ideal functionality in order to simulate
all the data she would gain by using this device in the real protocol.
The prover simply simulates each permutation and inputs the permuted
graph into the additional device.  The outcomes either correspond to
correct identification of the chosen graph, or they do not.  If they
do, the prover simply tells the ideal functionality to output $a=1$,
otherwise she tells it to output $a=0$.  

It is important that the inputs to the additional device are made
prior to the ideal being used.  Once the ideal has been executed, the
simulator cannot then generate pairs $(\Xi,c)$ that are correctly
distributed with $a$.



\chapter{The Complete Set Of Quantum States That Can Pass A GHZ Test}
\label{AppC}
The technique that we follow in this section is similar to that used
to find the complete set of states and measurements producing maximal
violation of the {\sc CHSH} inequality \cite{PopescuRohrlich}.

We seek the complete set of tripartite states (in finite dimensional
Hilbert spaces), and two-setting measurement devices that output
either $1$ or $-1$, such that, denoting the settings of device $i$ by
$P_i$ and $Q_i$, we have,
\begin{enumerate}
\item If all three detectors measure $P_i$, then the
  product of their outcomes is $+1$.
\item If two detectors measure $Q_i$ and one measures
  $P_i$, then the product of their outcomes is $-1$.
\end{enumerate}
These are equivalent to demanding
\begin{eqnarray}
\label{rel1}P_1\otimes P_2\otimes P_3\ket{\Psi}&=&-\ket{\Psi}\\
\label{rel2}Q_1\otimes Q_2\otimes P_3\ket{\Psi}&=&\ket{\Psi}\\
Q_1\otimes P_2\otimes Q_3\ket{\Psi}&=&\ket{\Psi}\\
\label{rel4}P_1\otimes Q_2\otimes Q_3\ket{\Psi}&=&\ket{\Psi},
\end{eqnarray}
where $\ket{\Psi}$ is the tripartite state.
We then have
\begin{eqnarray}
\nonumber
F\ket{\Psi}\equiv\frac{1}{4}\left(P_1\otimes Q_2\otimes Q_3+Q_1\otimes
P_2\otimes Q_3\,+\qquad\quad\right.\\
Q_1\otimes Q_2\otimes P_3-
\label{ineq2}
\left.P_1\otimes P_2\otimes
P_3\right)\ket{\Psi}&=&\ket{\Psi}.
\end{eqnarray}
$\ket{\Psi}$ is thus an eigenstate of $F$ with eigenvalue 1, so that
$F^2\ket{\Psi}=\ket{\Psi}$.  This is equivalent to
\begin{eqnarray}
\nonumber\left(i[P_1,Q_1]\otimes
i[P_2,Q_2]\otimes\openone+i[P_1,Q_1]\otimes\openone\otimes
i[P_3,Q_3]\,+\right.\quad\qquad\\\left.\openone\otimes i[P_2,Q_2]\otimes
i[P_3,Q_3]\right)\ket{\Psi}=&12\ket{\Psi}.
\end{eqnarray}

The maximum eigenvalue of $i[P_1,Q_1]$ is 2, hence
\begin{equation}
i[P_1,Q_1]\otimes i[P_2,Q_2]\otimes\openone\ket{\Psi}=4\ket{\Psi}
\end{equation}
and similar relations for the other permutations.  We hence have
\begin{equation}
i[P_1,Q_1]\otimes\openone\otimes\openone\ket{\Psi}=2\ket{\Psi}
\end{equation}
from which it follows that
\begin{equation}
\bra{\Psi}\left(\left\{P_1,Q_1\right\}\otimes\openone\otimes\openone\right)^2\ket{\Psi}=0
\end{equation}
and hence that
\begin{equation}
\left(\left\{P_1,Q_1\right\}\otimes\openone\otimes\openone\right)\ket{\Psi}=0.
\end{equation}
Consider the following Schmidt decomposition \cite{Nielsen&Chuang}:
$\ket{\Psi}=\sum_{i=1}^n\lambda_i\ket{i_1}\ket{i_{23}}$, where
$\lambda_i\geq 0$ $\forall$ $i$, and $n$ is the dimensionality of the
first system.  Then, if $\lambda_i\neq 0$ $\forall$ $i$, the
$\{\ket{i_1}\}$ are $n$ eigenstates of $\left\{P_1,Q_1\right\}$, each
having eigenvalue 0.  Since there are only $n$ eigenstates, we must
have $\left\{P_1,Q_1\right\}=0$.

If some of the $\lambda_i$ are zero, then we can define a projector
onto the non-zero subspace.  Call this $\Pi_1$, and define $p_1=\Pi_1
P_1\Pi_1$ and $q_1=\Pi_1 Q_1\Pi_1$.  Similarly, define projectors
$\Pi_2$ and $\Pi_3$, and hence operators $p_2,$ $q_2$ and $p_3$, $q_3$
by taking the Schmidt decomposition for systems (1,3) and 2, and (1,2)
and 3, respectively.  It is then clear that
\begin{equation}
\frac{1}{4}\left(p_1\otimes q_2\otimes q_3+q_1\otimes p_2\otimes
q_3+q_1\otimes q_2\otimes p_3-p_1\otimes p_2\otimes
p_3\right)\ket{\Psi}=\ket{\Psi}
\end{equation}
holds for the projected operators, and hence, these satisfy
$\{p_i,q_i\}=0$ for $i=1,2,3$.

The relationships, $p_i^2=\openone$, $q_i^2=\openone$, $\{p_i,q_i\}=0$
then apply for the Hilbert space restricted by $\{\Pi_i\}$.  These
imply that $p_i$, $q_i$ and $\frac{i}{2}[q_i,p_i]$ transform like the
generators of SU(2).  The operators may form a reducible
representation, in which case we can construct a block diagonal matrix
with irreducible representations on the diagonal.  The anticommutator
property means that only the two-dimensional representation can
appear, hence we can always pick a basis such that
$p_i=\openone_{d_i}\otimes\sigma_x{}_i$ and $q_i=\openone_{d_i}\otimes\sigma_y{}_i$
for some dimension, $d_i$, of identity matrix.  Our state then needs
to satisfy
\begin{equation}
\openone_{d_1}\otimes\sigma_x{}_1\otimes\openone_{d_2}\otimes\sigma_x{}_2\otimes\openone_{d_3}\otimes\sigma_x{}_3\ket{\Psi}=-\ket{\Psi},
\end{equation}
and similar relations for the other combinations analogous to
(\ref{rel2}--\ref{rel4}).
By an appropriate swap operation, this becomes
\begin{equation}
\openone_{d_1d_2d_3}\otimes\sigma_x{}_1\otimes\sigma_x{}_2\otimes
\sigma_x{}_3\ket{\Psi}=-\ket{\Psi},
\end{equation}
etc., which makes it clear that the system can be divided into
subspaces, each of which must satisfy the {\sc GHZ} relation (\ref{ineq2}).
In an appropriate basis, we can write
\begin{equation}
\ket{\Psi}=\left(\begin{array}{c}a_1\ket{\psi_{\text{GHZ}}}\\a_2\ket{\psi_{\text{GHZ}}}\\\vdots\end{array}\right),
\end{equation}
where $\ket{\psi_{\text{GHZ}}}=\frac{1}{\sqrt{2}}(\ket{000}-\ket{111})$
\footnote{This is the only solution to $(\sigma_x{}_1\otimes
\sigma_y{}_2\otimes \sigma_y{}_3+\sigma_y{}_1\otimes
\sigma_x{}_2\otimes \sigma_y{}_3+\sigma_y{}_1\otimes
\sigma_y{}_2\otimes \sigma_x{}_3-\sigma_x{}_1\otimes
\sigma_x{}_2\otimes \sigma_x{}_3)\ket{\psi_j}=4\ket{\psi_j}$, up to
global phase.}, and the complex
co-efficients $\{a_j\}$ simply weight each subspace and satisfy
$\sum_j |a_j|^2=1$.  We have hence
obtained the complete set of states and operators satisfying
(\ref{rel1}--\ref{rel4}), up to local unitaries.



\fancyhead[RO]{{\bf References}}
\addcontentsline{toc}{chapter}{References}

\bibliographystyle{unsrt}

\renewcommand{\bibname}{References}

\end{document}